\documentclass[11pt,lettersize]{article}

\usepackage{authblk}
\usepackage{fullpage}

\usepackage{amsmath}
\usepackage{amssymb}
\usepackage{amsthm}

\usepackage{graphicx} 
\usepackage[round]{natbib}  

\usepackage{complexity}
\usepackage{xspace}
\usepackage{xcolor}
\usepackage{tabularx}
\usepackage{booktabs}
\usepackage{thm-restate}

\usepackage{enumerate}
\usepackage{nicefrac}

\usepackage{algorithm}
\usepackage{algorithmic}

\usepackage{stackengine}

\usepackage{cleveref}
\AddToHook{cmd/appendix/before}{%
    \crefalias{section}{appendix}%
    \crefalias{subsection}{appendix}
}
\providecommand\creflastconjunction{, and }

\newcommand{\Crefor}[1]{\renewcommand\crefpairconjunction{ or }\renewcommand\crefmiddleconjunction{, }\renewcommand\creflastconjunction{, or }\Cref{#1}\renewcommand\crefpairconjunction{ and }\renewcommand\crefmiddleconjunction{, }\renewcommand\creflastconjunction{, and }}

\crefname{claim}{Claim}{Claims}
\crefname{observation}{Observation}{Observations}

\definecolor{CoalitionColor}{HTML}{1D42A6}
\definecolor{OxfordRed}{HTML}{AA1A2D}

\usepackage{figures/tikzit}

\tikzstyle{agent}=[fill=white, draw=black, shape=circle, minimum size=30pt]
\tikzstyle{agent_part_A}=[fill=red, draw=black, shape=circle, minimum size=30pt]
\tikzstyle{agent_part_B}=[fill=cyan, draw=black, shape=circle, minimum size=30pt]
\tikzstyle{utility}=[-triangle 45, draw=black, fill=none]
\tikzstyle{dashed-utility}=[fill=none, -triangle 45, dashed]
\tikzstyle{dashed_undirected}=[fill=none, triangle 45 - triangle 45, dashed]
\tikzstyle{text-node}=[fill=white, draw=none, shape=rectangle]
\tikzstyle{utility-symmetric}=[fill=none, draw=black, triangle 45 - triangle 45]
\tikzstyle{agent_part_C}=[fill={rgb,255: red,191; green,255; blue,0}, draw=black, shape=circle, minimum size=30pt]
\tikzstyle{dotted-utility}=[fill=none, -triangle 45, dotted, thick]
\tikzstyle{agent_group}=[fill=white, draw=black, shape=circle, ultra thick, minimum size=30pt]
\tikzstyle{large_agent}=[fill=white, draw=black, shape=circle, minimum size=37pt]
\tikzstyle{large_group}=[fill=white, draw=black, shape=circle, minimum size=37pt, ultra thick]
\tikzstyle{very_large_agent}=[fill=white, draw=black, shape=circle, minimum size=55pt]
\tikzstyle{very_large_group}=[fill=white, draw=black, shape=circle, minimum size=55pt, ultra thick]
\tikzstyle{negative_incentive}=[fill=none, draw=red, -triangle 45]
\tikzstyle{utility_intermediate}=[fill=none, draw=black]
\tikzstyle{new style 0}=[fill=none, draw=red]

\tikzstyle{utilityy}=[->]
\tikzstyle{dashed_undirectedd}=[-, fill={rgb,255: red,227; green,233; blue,245}, draw={rgb,255: red,30; green,98; blue,186}]
\tikzstyle{utility-symmetricc}=[-]
\tikzstyle{dotted-utilityy}=[->, dotted]
\tikzstyle{dashed-utilityy}=[-, dashed]
\tikzstyle{ddashed}=[-, dashed]

\usetikzlibrary{calc}

\newcommand{\convexpath}[2]{
  [   
  create hullcoords/.code={
    \global\edef\namelist{#1}
    \foreach [count=\counter] \nodename in \namelist {
      \global\edef\numberofnodes{\counter}
      \coordinate (hullcoord\counter) at (\nodename);
    }
    \coordinate (hullcoord0) at (hullcoord\numberofnodes);
    \pgfmathtruncatemacro\lastnumber{\numberofnodes+1}
    \coordinate (hullcoord\lastnumber) at (hullcoord1);
  },
  create hullcoords
  ]
  ($(hullcoord1)!#2!-90:(hullcoord0)$)
  \foreach [
  evaluate=\currentnode as \previousnode using \currentnode-1,
  evaluate=\currentnode as \nextnode using \currentnode+1
  ] \currentnode in {1,...,\numberofnodes} {
    let \p1 = ($(hullcoord\currentnode) - (hullcoord\previousnode)$),
    \n1 = {atan2(\y1,\x1) + 90},
    \p2 = ($(hullcoord\nextnode) - (hullcoord\currentnode)$),
    \n2 = {atan2(\y2,\x2) + 90},
    \n{delta} = {Mod(\n2-\n1,360) - 360}
    in 
    {arc [start angle=\n1, delta angle=\n{delta}, radius=#2]}
    -- ($(hullcoord\nextnode)!#2!-90:(hullcoord\currentnode)$) 
  }
}

\usepackage{subcaption}

\usepackage{todonotes}
\presetkeys{todonotes}{inline}{}

\newtheorem{theorem}{Theorem}
\newtheorem{corollary}[theorem]{Corollary}

\newtheorem{lemma}[theorem]{Lemma}

\newtheorem{observation}[theorem]{Observation}
\newtheorem{claim}[theorem]{Claim}

\newtheorem{example}[theorem]{Example}

\newcommand{\NN}{\mathbb{N}\xspace} 
\newcommand{\QQ}{\mathbb{Q}\xspace} 

\newenvironment{claimproof}
{
    
    \proof
}
{
    \endproof
    
}

\makeatletter
\def\renewtheorem#1{%
  \expandafter\let\csname#1\endcsname\relax
  \expandafter\let\csname c@#1\endcsname\relax
  \gdef\renewtheorem@envname{#1}
  \renewtheorem@secpar
}
\def\renewtheorem@secpar{\@ifnextchar[{\renewtheorem@numberedlike}{\renewtheorem@nonumberedlike}}
\def\renewtheorem@numberedlike[#1]#2{\newtheorem{\renewtheorem@envname}[#1]{#2}}
\def\renewtheorem@nonumberedlike#1{  
\def\renewtheorem@caption{#1}
\edef\renewtheorem@nowithin{\noexpand\newtheorem{\renewtheorem@envname}{\renewtheorem@caption}}
\renewtheorem@thirdpar
}
\def\renewtheorem@thirdpar{\@ifnextchar[{\renewtheorem@within}{\renewtheorem@nowithin}}
\def\renewtheorem@within[#1]{\renewtheorem@nowithin[#1]}
\makeatother

\newcommand{\problemdefijcai}[3]{
  \begin{center}
    \begin{minipage}{0.95\columnwidth}
      \noindent
      \textsc{#1}
      
      \vspace{2pt}
      \setlength{\tabcolsep}{3pt}
      \begin{tabularx}{\columnwidth}{@{}lX@{}}
        \textbf{Input:} 		& #2 \\
        \textbf{Question:} 	& #3
      \end{tabularx}
    \end{minipage}
  \end{center}
}

\newcommand{\bigO}{\mathcal{O}\xspace} 

\newcommand{\problemname}[1]{\textsc{#1}\xspace}
\newcommand{\RXC}{\problemname{RX3C}}
\newcommand{\XTC}{\problemname{X3C}}
\newcommand{\INDSET}{\problemname{IndSet}}
\renewcommand{\PCD}{\problemname{PCD}}
\newcommand{\NCD}{\problemname{NCD}}

\newcommand{\NS}{\ensuremath{\text{NS}}\xspace}
\newcommand{\IS}{\ensuremath{\text{IS}}\xspace}
\newcommand{\CNS}{\ensuremath{\text{CNS}}\xspace}
\newcommand{\CIS}{\ensuremath{\text{CIS}}\xspace}

\newcommand{\SMS}{\ensuremath{\text{SMS}}\xspace}

\newcommand{\VIS}{\ensuremath{\text{VIS}}\xspace}
\newcommand{\VOS}{\ensuremath{\text{VOS}}\xspace}
\newcommand{\SVS}{\ensuremath{\text{VS}}\xspace} 

\newcommand{\ASHG}{\ensuremath{\text{ASHG}}\xspace}
\newcommand{\ASHGs}{\ensuremath{\text{ASHGs}}\xspace}
\newcommand{\FHG}{\ensuremath{\text{FHG}}\xspace}
\newcommand{\FHGs}{\ensuremath{\text{FHGs}}\xspace}
\newcommand{\MFHG}{\ensuremath{\text{MFHG}}\xspace}
\newcommand{\MFHGs}{\ensuremath{\text{MFHGs}}\xspace}

\newcommand{\Fin}{\ensuremath{F_\mathrm{in}}\xspace}
\newcommand{\Fout}{\ensuremath{F_\mathrm{out}}\xspace}

\newcommand{\qin}{\ensuremath{q_{\mathrm{in}}}\xspace}
\newcommand{\qout}{\ensuremath{q_{\mathrm{out}}}\xspace}
\newcommand{\tin}{\ensuremath{t_{\mathrm{in}}}\xspace}
\newcommand{\tout}{\ensuremath{t_{\mathrm{out}}}\xspace}

\newcommand{\startpart}{\ensuremath{\pi_0}} 
\newcommand{\vf}{\ensuremath{v}\xspace} 
\newcommand{\uf}{\ensuremath{u}\xspace} 

\newcommand{\devimplies}{
    \ensuremath{
    \smash{
    \raisebox{-2.5pt}{
        \stackengine{3.5pt}{\hspace{1pt}\tiny\ensuremath{\rightharpoondown}}{\ensuremath{\subset}}{O}{c}{F}{F}{L}
    }}}
}

\title{Deviation Dynamics in Cardinal Hedonic Games\thanks{We dedicate this work to Dr. Jochen Zech. Rest in peace.}}

\author[1]{Valentin Zech}
\author[2]{Martin Bullinger}

\affil[1]{ \small Department of Computer Science, University of Oxford, UK}
\affil[2]{ \small School of Engineering Mathematics and Technology, University of Bristol, UK\protect\\ \vspace*{0.05cm} zech@vzech.de, martin.bullinger@bristol.ac.uk}

\date{}

\begin{document}

\maketitle

\begin{abstract}
Computing stable partitions in hedonic games is a challenging task because there exist games in which stable outcomes do not exist.
    Even more, these No-instances can often be leveraged to prove computational hardness results.
    We make this impression rigorous in a dynamic model of cardinal hedonic games by providing meta theorems.
    These imply hardness of deciding about the possible or necessary convergence of deviation dynamics based on the mere existence of No-instances.
    Our results hold for additively separable, fractional, and modified fractional hedonic games (\ASHGs, \FHGs, and \MFHGs).
    Moreover, they encompass essentially all reasonable stability notions based on single-agent deviations. 
    In addition, we propose dynamics as a method to find individually rational and contractually individual stable (\CIS) partitions in \ASHGs.  
    In particular, we find that \CIS dynamics from the singleton partition possibly converge after a linear number of deviations but may require an exponential number of deviations in the worst case.
\end{abstract}

\section{Introduction}\label{sec:intro}

The field of Computational Social Choice (COMSOC) is concerned with aggregating potentially conflicting individual preferences of different agents into a compromise solution \citep{BCE+14a}. 
With various applications to, among others, politics, multi-agent systems, and economic processes, coalition formation is among the primary areas of interest within COMSOC \citep{RaVo15a}. 
Here, a group of agents must be divided into distinct coalitions, with each agent having preferences for these divisions.

A common restriction on agents' preferences is that their utility depends only on which agents are present in their own coalition. 
This restriction describes the model of so-called \emph{hedonic games} \citep{DrGr80a}.
Since their introduction, they have been a constant area of interest in the literature on artificial intelligence and multi-agent systems 
\citep{AzSa15a,BER24a}.
Hedonic games have been successfully utilized to model many interesting real-world settings, such as research team formation \citep{AlRe04a}, allocation of indivisible goods \citep{Pete16b}, task allocation for wireless agents \citep{SHB+11a}, and community detection in social networks \citep{ABB+17a}. Further, they have proven to be a powerful theoretical model in the context of clustering \citep{FLN15a,AAK+22a,CLMP22a}, one of the central research topics in the realm of machine learning.

A prominent measure for the desirability of outcomes in hedonic games is stability, defined as the absence of beneficial deviations by agents to join other coalitions \citep{BoJa02a}. 
In certain scenarios, it is sensible to additionally require partial or unanimous consent of the otherwise affected agents, which give rise to a wide landscape of notions of stability \citep{AzSa15a}. 
We will focus on those defined by deviations of single agents. 

Some stability notions guarantee a stable outcome in any hedonic game, e.g., contractual individual stability where a deviation requires unanimous consent of all involved agents. 
For most stability notions, however, stable partitions are not guaranteed to exist, even in fairly restricted game classes.
This gives rise to the problem of deciding whether a given hedonic game admits a stable partition.
A common observation is that No-instances, i.e., games without a stable partition, can be used as gadgets to prove computational boundaries of the existence problem \citep[see, e.g.,][]{SuDi10a,ABS11c,PeEl15a,BBT23a}.\footnote{Two notable exceptions are aversion-to-enemies games and the locally egalitarian variant of hedonic games.
In both cases, the core is nonempty but an outcome in the core is \NP-hard to compute \citep{DBHS06a,BuKo21a}.}

The work discussed so far is only concerned with whether an outcome is stable or not, while it matters less how this outcome is obtained.
One natural way to model the process of obtaining stable outcomes are deviation dynamics, where the agents start in some initial state and then iteratively perform deviations as long as they have an incentive to do so, see, e.g., \citep{BFFMM18a,GaSa19a,BBW21b,BBT23a}.
Such dynamics have previously been utilized successfully, e.g., to show that partitions satisfying a specific stability notion always exist in a particular game class \citep{BoJa02a, BoEl20a, BBT23a,CaKi24a}, to study the complexity of computing stable outcomes \citep{GaSa19a} or to place an upper bound on the price of stability in terms of achieving high social welfare \citep{BFFMM18a,MMV20a}.
While dynamics are, therefore, a powerful general tool, scenarios in which  dynamics are guaranteed to converge offer a decentralized approach to reaching desirable partitions. Thus, they give rise to interesting questions in their own right. 
Specifically, \citet{BBW21b} ask whether, given a hedonic game and a starting partition, dynamics possibly or necessarily converge, i.e., reach a stable partition.

\subsection{Contribution}
We will make the intuition that No-instances lead to computational intractabilities explicit.
In contrast to previous work that explicitly constructs No-instances and uses them to prove individual hardness results \citep[see, e.g.,][]{SuDi10a,BBT23a}, we present meta-theorems that treat No-instances as a black box. 
This approach enables future hardness results to be derived by identifying a single suitable instance.
The meta theorems concern the intractability of possible and necessary convergence of dynamics, and apply to three prominent classes of hedonic games: additively separable \citep{BoJa02a}, fractional \citep{ABB+17a}, and modified fractional \citep{Olse12a} hedonic games.
They hold for most reasonable stability notions based on deviations between Nash deviations (which simply need to make the deviator better off) and contractual individual deviations (which additionally require the consent of all other agents).
We demonstrate the generality of our meta theorems by applying them for a general class of voting-based stability notions that encompass a wide range of known and new stability notions.

Finally, we zoom in on a special case of dynamics that necessarily converge, namely those based on contractual individual deviations for additively separable hedonic games.
When starting from the singleton partition, the resulting partition additionally is individually rational, i.e., at least as good for each agent as being on her own. 
We show that fast convergence is always possible. 
It is, however, unclear how to efficiently identify the associated deviations.
Simply running any sequence of deviations may take an exponential number of steps.
Nonetheless, we identify the structural reason behind this result, leading to a fixed-parameter tractability result based on the number of certain valuation pairs.

\subsection{Related Work}

Hedonic games were first introduced by \citet{DrGr80a}, and later popularized by \citet{BoJa02a}, \citet{BKS01a}, and \citet{CeRo01a}.
An overview is provided in the book chapters by \citet{AzSa15a} and \citet{BER24a}.

The axiomatic and computational properties of stability have been studied extensively in cardinal hedonic games \citep[see, e.g.,][]{DBHS06a,SuDi10a,ABS11c,Woeg13a,BFFMM18a,ABB+17a,BoEl20a,BBT23a}. 
\citet{SuDi10a}, specifically, provide a detailed overview of stability based on single-agent deviations in additively separable hedonic games. 
Related to our efforts to study the computational complexity of finding a partition that is individually rational and contractually individually stable,
\citet{ABS11c} provide an algorithm for computing a (not necessarily individually rational) partition that is contractually individually stable in additively separable hedonic games.
Further, \citet{PeEl15a} utilize a meta approach to show hardness for several game classes and stability notions simultaneously, similar to our unified theory.
In contrast to our investigation of deviation dynamics, their paper concerns the general existence of stable outcomes.

In this light, a recent trend has been to study the dynamic aspects of coalition formation based on beneficial deviations, which offer a decentralized approach to finding stable outcomes and can thus model specific real-world scenarios more realistically. 
Most related is the work by \citet{BBW21b} that studies the computational complexity of possible and necessary convergence of dynamics in a variety of game classes.
The only overlap with our work is the consideration of fractional hedonic games.
While \citet{BBW21b} only study individual stability, our meta theorems work for a much larger set of stability notions and additionally concerns other classes of cardinal hedonic games. 
Subsequently, \citet{BuSu24a} study possible and necessary convergence for the equivalent of Nash stability in a generalization of additively separable hedonic games. 

Further, \citet{BFFMM18a} study Nash stability in fractional hedonic games and, for instance, utilize dynamics to design an algorithm that approximates the maximum social welfare of a Nash stable outcome in polynomial time. 
\citet{GaSa19a} settle the complexity of deciding whether a stable partition exists in symmetric additively separable hedonic games by treating this question as local search problems. 
\citet{BBT23a} also study computational questions related to the existence of stable partitions, where all their positive results are obtained by proving convergence of dynamics. 
\citet{BBK23a} propose a version of hedonic games specifically adapted to a dynamic setting, where utilities change after a deviation takes place.
Their work also has implications for a fixed-utility setting:
In particular, they consider the computational complexity of convergence in a given time limit and prove hardness results for additively separable hedonic games.
In addition, \citet{HVW18a}, \citet{BuKo21a}, and \citet{FMM21a} study dynamics in hedonic games based on group deviations.
Finally, we note that similar dynamic approaches to finding stable solutions have been studied in the context of stable matchings \citep{AbRo95a,HVW18a,BrWi20a}.

Further, the study of dynamic processes has recently received increased interest in the research community of computational social choice and collective decision-making at large \citep[see, e.g.,][]{ZBET24a, EOT24a, ILNN24a, CaNa24a}.

\section{Preliminaries}
In this section, we introduce preliminaries.
We use the convention that $\NN$ is the set of nonnegative integers, including~$0$.
For 
$i \in \NN$, $i \ge 1$,
we denote $[i]:=\{1,\dots, i\}$.

\subsection{Hedonic Games}
We consider a finite set $N$ of $n := |N|$ agents. 
A nonempty subset of agents is called a \emph{coalition}.
We aim to partition the agents in $N$ into disjoint coalitions. 
A \emph{coalition structure} (or \emph{partition}) of $N$ is a subset $\pi \subseteq 2^N$ with $\bigcup_{C \in \pi} C = N$, where, for all $C, D \in \pi$, it holds that $C = D$, or $C \cap D = \emptyset$. 
Given an agent $a \in N$, we denote by $\pi(a)$ the coalition in $\pi$ that contains $a$. 
Let $\mathcal{N}_a = \{C \subseteq N \mid a \in C\}$ denote the set of all coalitions that $a$ can belong to. 
We refer to the partition $\pi = \{ \{a\} \mid a \in N\}$ as the \emph{singleton partition}, and to $\pi = \{N\}$ as the \emph{grand coalition}. Further, for each agent $a \in N$, we call $\{a\}$ the \emph{singleton coalition} of $a$.

A \emph{hedonic game} $G = (N, \succsim)$ consists of a set $N$ of agents, and a \emph{preference profile} $\succsim = (\succsim_a)_{a \in N}$ where $\succsim_a \subseteq \mathcal{N}_a \times \mathcal{N}_a$ is a complete, reflexive, and transitive binary relation called agent $a$'s \emph{preference relation} \citep{DrGr80a}.
Given two coalitions $C, D \in \mathcal{N}_a$, we write $C \succ_a D$ if $C \succsim_a D$ but not $D \succsim_a C$ (i.e., $a$ \emph{strictly prefers} $C$ over $D$).
We say that a partition $\pi$ is \emph{individually rational} if $\pi(a) \succsim_a \{a\}$ for each agent $a \in N$, i.e., no agent would strictly prefer to be in her respective singleton coalition.

Agents have preferences over partitions based on preferences over coalitions. 
Given two partitions $\pi, \pi'$ of $N$, we say that $\pi \succsim_a \pi'$ if and only if $\pi(a) \succsim_a \pi'(a)$. 
Further, we denote by $G - a$ the game with agent set $N \setminus \{a\}$ that is induced by $G$ by removing agent $a$. We write $\pi - a$ to mean the partition of $N \setminus \{a\}$ that resulted from~$\pi$ by removing $a$ from her coalition, formally, $\pi - a := \{ C \setminus \{a\} \mid C \in \pi, C\neq \{a\}\}$.

We consider classes of hedonic games in which preference relations evolve from cardinal utility functions, i.e., agents have numeric value for each coalition and preferences are based on comparing these values. 
Formally, a \emph{cardinal hedonic game} is given by the pair $(N,u)$ where $N$ is the agent set and $\uf = (\uf_a\colon \mathcal{N}_a\to\QQ)_{a\in N}$ a profile of \emph{utility functions}.
Then, $(N,u)$ induces the hedonic game $(N,\succsim)$ where, for every agent $a\in N$ and coalitions $C, D\in \mathcal N_a$, we define $C\succsim_a D$ if and only if $\uf_a(C)\ge \uf_a(D)$.
We say that $\uf_a(C)$ is $a$'s utility for coalition $C$ and extend this to utilities for partitions by setting $\uf_a(\pi) := \uf_a(\pi(a))$.

Cardinal hedonic games generally require to specify a utility for an exponentially large set of coalitions.
To avoid listing these all explicitly, several classes of cardinal hedonic games have been proposed where utility functions are represented succinctly by merely specifying valuations for single agents. 
Let $G = (N, \uf)$ be a cardinal hedonic game and let $(\vf_a\colon N\to \QQ)_{a \in N}$ be a collection of valuation functions.

\iftrue
Following \citet{BoJa02a}, $G$ is called an \textit{additively separable hedonic game} (\ASHG) if for all $a \in N, C\in \mathcal N_a$ it holds that $\uf_a(C) = \sum_{b\in C\setminus \{a\}}\vf_a(b)$. 
Following \citet{ABB+17a}, $G$ is called a \textit{fractional hedonic game} (\FHG) if for all $a \in N, C\in \mathcal N_a$ it holds that $\uf_a(C) = \sum_{b\in C\setminus \{a\}}\frac{\vf_a(b)}{|C|}$. 
Following \citet{Olse12a}, $G$ is called a \textit{modified fractional hedonic game} (\MFHG) if for all $a \in N$, it holds that $\uf_a(\{a\}) = 0$ and for all $C\in \mathcal N_a$, $C\neq \{a\}$ it holds that $\uf_a(C) = \sum_{b\in C\setminus \{a\}}\frac{\vf_a(b)}{|C|-1}$.
\fi

In other words, the utility in an \ASHG is the sum of valuations for agents in the considered coalition, and the utility in an \FHG and \MFHG is the average valuation, where \FHGs include the consideration of the agent herself. 
All three game classes are fully specified by the valuation functions and we therefore also represent an \ASHG, \FHG, or \MFHG $G$ by the pair $(N,\vf)$, where $\vf = (\vf_a\colon N\to \QQ)_{a\in N}$ is a profile of \textit{valuation functions}.

Note that the valuation functions of \ASHGs, \FHGs, and \MFHGs can be represented as a weighted directed graph, where the vertices are agents, and, given two agents $a, b \in N$, there is an edge from $a$ to $b$ with weight $\vf_a(b)$.

\subsection{Single-Agent Stability}\label{sec:stability}

We now formalize how to capture stability based on beneficial deviations by single agents.
Given a hedonic game $G = (N, \succsim)$, a \emph{single-agent deviation} of an agent $a \in N$ transforms a partition $\pi$ of $N$ into a partition $\pi'$ of $N$, where $\pi(a) \neq \pi'(a)$, and, for all agents $b \in N \setminus \{a\}$, it holds that $\pi(b) \setminus \{a\} = \pi'(b) \setminus \{a\}$. 
We denote such a deviation by $\pi \overset{a}{\rightarrow} \pi'$. 
Intuitively, agent $a$ deviates away from coalition $\pi(a)$, to join coalition $\pi'(a)$ (importantly, $\pi'(a)$ can be $a$'s singleton coalition), while all other coalitions remain unchanged.

A minimum requirement for the desirability of a deviation is whether the deviator is better off by performing this deviation.
A \emph{Nash deviation} is a single-agent deviation $\pi \overset{a}{\rightarrow} \pi'$ of an agent $a \in N$ such that $\pi'(a) \succ_a \pi(a)$. 
A partition $\pi$ which does not admit a Nash deviation is said to be \emph{Nash stable} (\NS), and $\pi$ is called an \NS partition. 

While Nash stability offers a very strong and desirable solution concept, \NS deviations completely disregard the opinion of members in the abandoned and welcoming coalition. 
In this light, several stability notions enforce additional requirements to be satisfied for a deviation to be valid.
We introduce a general class of such stability notions based on voting among the involved agents.

Let $C \subseteq N$ be a coalition and $a \in N$ an agent. Following \citet{BBT23a}, we define the \emph{favour-in set} $\Fin(C, a)$ and \emph{favour-out set} $\Fout(C, a)$ of $C$ with respect to $a$ as 
\begin{align*}
    \Fin(C, a) := \{ b \in C \setminus \{a\} \mid C \cup \{a\} \succ_b C \setminus \{a\} \} \text{ and}\\
    \Fout(C, a) := \{ b \in C \setminus \{a\} \mid C \setminus \{a\} \succ_b C \cup \{a\} \}\text.
\end{align*}
These capture the agents in $C$ that prefer $a$ inside or outside the coalition $C$.
Note that the definition is valid regardless of whether $a$ is part of $C$.

Let $\qout, \qin \in [0, 1]$ be two real numbers interpreted as quotas.
A Nash deviation $\pi \overset{a}{\rightarrow} \pi'$ of an agent $a \in N$ is called a \emph{$(\qout,\qin)$-vote deviation} if
\begin{enumerate}[1.]
    \item $|\Fout(\pi(a), a)|\ge \qout(|\Fin(\pi(a), a)| + |\Fout(\pi(a), a)|)$ and 
    \item $|\Fin(\pi'(a), a)| \ge \qin(|\Fin(\pi'(a), a)| + |\Fout(\pi'(a), a)|)$.
\end{enumerate}
Hence, such a deviation requires that at least a $\qout$-fraction of the nonindifferent members of the abandoned coalition and a $\qin$-fraction of the nonindifferent members of the welcoming coalition are strictly in favor of the deviation.
Now, a partition is said to be \emph{$(\qout, \qin)$-voting-stable} ($(\qout, \qin)$-\SVS) if it does not admit a $(\qout,\qin)$-vote deviation.

Our stability framework captures most single-deviation stability notions commonly considered in the literature.
If $\qout, \qin \in \{0, 1\}$, we obtain stability notions based on unanimous consent whenever consent is required.
Specifically, $(0,0)$-\SVS is \NS, $(0,1)$-\SVS is called \emph{individual stability} (\IS), $(1,0)$-\SVS is called \emph{contractual Nash stability} (\CNS), and $(1,1)$-\SVS is called \emph{contractual individual stability} (\CIS) \citep{BoJa02a,SuDi07b}. 
In addition, $(\qout, \qin)$-\SVS generalizes previously studied voting-based stability concepts: 
\citet{GaSa19a} consider $(0, \qin)$-\SVS and $(\qout, 0)$-\SVS under the names of \emph{vote-in} and \emph{vote-out stability} (\VIS and \VOS), and \citet{BBT23a} consider $(0, \frac 12)$-\SVS, $(\frac 12,0)$-\SVS, and $(\frac 12,\frac 12)$-\SVS.
\citet{BBT23a} call the latter \emph{separate-majorities stability} (\SMS).
Among all of these, only \CIS guarantees the existence of stable partitions.

Given a stability notion $\chi$, we refer to the corresponding deviations and stable partitions as $\chi$ deviations and $\chi$ partitions, respectively.

Given two stability notions $\chi$ and $\chi'$, we write $\chi \devimplies \chi'$ if every $\chi$ deviation is also a $\chi'$ deviation. 
For instance, for every $\qout, \qin \in [0, 1]$, it holds that $(\qout, \qin)\text{-\SVS} \devimplies \NS$ and $\CIS \devimplies (\qout, \qin)$-\SVS.

\subsection{Standard Stability Notions}

In the last section, we introduced a class of specific stability notions based on voting.
To state our meta theorems, we propose a novel condition to capture an even more general class of stability notions between \NS and \CIS. 
These are defined for cardinal hedonic games and should satisfy two properties:
\begin{enumerate}[1.]
    \item the feasibility of deviations only depends on the utility changes of the involved agents, not their identities, i.e., deviations are anonymously hedonic,
    \item deviations that are stronger than feasible deviations are also feasible, i.e., deviations are monotonic. 
\end{enumerate}
We formalize this in the following.

Let $G = (N, \uf)$ be a cardinal hedonic game, let $a \in N$ be an agent, and let $\pi, \pi'$ be two partitions of $N$. 
We refer to $\mathit{uc}_G(a, \pi, \pi') := (\uf_a(\pi), \uf_a(\pi'))$ as the \emph{utility-change tuple} of $a$ with respect to $G$, $\pi$ and $\pi'$. 
Further, we refer to the multisets
    $\mathit{UC}_G^{\mathrm{out}}(a, \pi, \pi') := \{\mathit{uc}_G(b, \pi, \pi') \mid b \in \pi(a) \setminus \{a\} \}$
and
    $\mathit{UC}_G^{\mathrm{in}}(a, \pi, \pi') := \{\mathit{uc}_G(b, \pi, \pi') \mid b \in \pi'(a) \setminus \{a\} \}$
as the \emph{utility-change-out multiset} and \emph{utility-change-in multiset} of $a$ with respect to $G$, $\pi$ and $\pi'$, respectively. 
We denote by $\mathcal{UC}$ the set of all utility-change multisets (that is, both utility-change-out and utility-change-in multisets). 
For all functions $uc_G$, $\mathit{UC}_G^{\mathrm{out}}$, and $\mathit{UC}_G^{\mathrm{in}}$, we will omit the game $G$ whenever it is clear from the context. 

We say that a stability notion $\chi$ is \emph{anonymously hedonic} if there exists a polynomial-time computable function $f_\chi: \mathcal{UC} \times \mathcal{UC} \times (\QQ \times \QQ) \rightarrow \{0, 1\}$, such that for all single-agent deviations $\pi \overset{a}{\rightarrow} \pi'$, it holds that

\begin{equation*}
    f_\chi(\mathit{UC}^{\mathrm{out}}(a, \pi, \pi'), \mathit{UC}^{\mathrm{in}}(a, \pi, \pi'), \mathit{uc}(a, \pi, \pi'))
    = 
    \begin{cases}
        & 1, \text{ if } \pi \overset{a}{\rightarrow} \pi' \text{ is a } \chi\text{ deviation},\\
        & 0, \text{ otherwise}.
    \end{cases}
\end{equation*}

Simply put, the validity of a deviation with respect to an anonymously hedonic stability notion solely depends on the changes in the utility of abandoned and welcoming coalitions and that of the deviator. In particular, this captures all stability notions that are implied by \NS, and further only depend on the sizes of the favor-in and favor-out sets of the abandoned and welcoming coalitions. However, the class of anonymously hedonic stability notions allows for more nuanced requirements. 
For example, a deviation may be allowed if it is an \NS deviation, and increases the utilitarian welfare, defined as $\sum_{a \in N} \uf_a(\pi)$ for partition $\pi$.

Next, given two multisets $X, Y \in \mathcal{UC}$, we say that $X$ \emph{dominates} $Y$, written $Y \trianglelefteq X$, if it holds that $|Y| \leq |X|$ and:
$$
    \forall (y, y') \in Y, (x, x') \in X: y'-y \leq x'-x.
$$

Now, an anonymously hedonic stability notion $\chi$ is \emph{monotone} if, for all $X, X', Y, Y' \in \mathcal{UC}$, and $z, z' \in \QQ \times \QQ$, where $X \trianglelefteq X'$, $Y \trianglelefteq Y'$, and $\{z\} \trianglelefteq \{z'\}$, it holds that:
$$
     f_\chi(X, Y, z) \leq f_\chi(X', Y', z'), 
$$
i.e., whenever a deviation is allowed with parameters $X, Y, z$, then it must also be allowed with parameters $X', Y', z'$. 
We will refer to anonymously hedonic monotone stability notions as \emph{standard stability notions}.

We illustrate standard stability with an example.

\begin{example}
    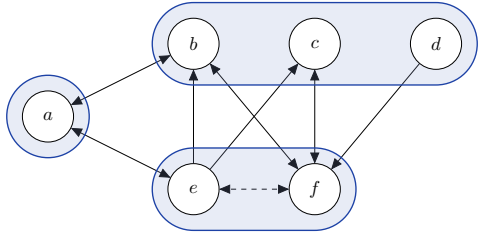
\begin{figure}[t]
        \centering
            \resizebox{.4\textwidth}{!}{
                \begin{tikzpicture}
	\begin{pgfonlayer}{nodelayer}
		\node [style=agent] (0) at (-9, 0) {$a$};
		\node [style=agent] (1) at (-3, 3) {$b$};
		\node [style=agent] (2) at (2, 3) {$c$};
		\node [style=agent] (3) at (7, 3) {$d$};
		\node [style=agent] (4) at (-3, -3) {$e$};
		\node [style=agent] (5) at (2, -3) {$f$};
	\end{pgfonlayer}
	\begin{pgfonlayer}{edgelayer}
		\draw [style=utility-symmetric] (0) to (1);
		\draw [style=utility-symmetric] (0) to (4);
		\draw [style=utility] (3) to (5);
		\draw [style=utility] (4) to (1);
		\draw [style=utility] (4) to (2);
		\draw [style=utility-symmetric] (5) to (2);
		\draw [style=utility-symmetric] (1) to (5);
		\draw [style={dashed_undirected}] (4) to (5);

        \draw[thick,CoalitionColor, fill=CoalitionColor!50, fill opacity=0.2]  \convexpath{1, 3}{1.7cm};
        \draw[thick,CoalitionColor, fill=CoalitionColor!50, fill opacity=0.2]  \convexpath{4, 5}{1.7cm};
        \draw[thick,CoalitionColor, fill=CoalitionColor!50, fill opacity=0.2] (0) circle[radius=1.7cm];
	\end{pgfonlayer}
\end{tikzpicture}
                }
        
        \caption{Illustration of an \ASHG. Blue boxes indicate the initial partition. 
        A straight arrow from an agent $x$ to an agent $y$ indicates $\vf_x(y) = 1$ while a dashed arrow indicates $\vf_x(y) = -1$.
        Missing arrows indicate a valuation of~$0$.\label{fig:standard_stability_example}}
    \end{figure}

    Consider the \ASHG depicted in \Cref{fig:standard_stability_example} with agents acting according to a standard stability notion. 
    First, assume that agent $b$ has a deviation to join $\{a\}$. 
    Then one can verify that both $b$ and $c$ must also be allowed to deviate to join $\{e, f\}$.
    However, we cannot infer any further deviations. In particular, despite the fact that the abandoned agent of a deviation of an agent in $\{e, f\}$ strictly increases her utility, we cannot infer whether, e.g., $e$ can deviate to join $\{a\}$, since $|\{e, f\}| < |\{b, c, d\}|$. 

    Next, assume that $f$ cannot deviate to join $\{b, c, d\}$. 
    Then, it holds that $e$ can also not deviate to join $\{b, c, d\}$ or $\{a\}$. However, perhaps somewhat surprisingly, we cannot say whether, e.g., agent $d$ can deviate to join $\{e, f\}$, again since $|\{e, f\}| < |\{b, c, d\}|$. 
\end{example}

We note that our voting-based stability notions are standard stability notions, as the relevant favor-in and favor-out sets can be reconstructed with the information captured in the utility-change multisets.
We defer the formal proof to \Cref{app:VSstandard}.

\begin{restatable}{proposition}{VSstandard}\label{prop:VSstandard}
    Let $\qout, \qin \in [0,1]$.
    Then, $(\qout,\qin)$-\SVS is a standard stability notion.
\end{restatable}

We remark that our notion of monotonicity does not capture all stability notions between NS and CIS, e.g., it fails to capture some notions that rely on 
the egalitarian welfare.

\subsection{Deviation Dynamics}

We are ready to introduce the central concept of this paper, which we will utilize to formulate our decision problems.

Stability notions naturally induce \emph{dynamics}, where, given a hedonic game and a starting partition of the agents, we iteratively obtain successor partitions by letting agents perform deviations from the current partition in alignment with the stability notion. 

Formally, let $\chi$ be a stability notion, let $G = (N, \succsim)$ be a hedonic game with a set $N$ of agents, and let $\startpart$ be a partition of $N$. 
Then, an \emph{execution of the $\chi$ dynamics} of $(G, \startpart)$ is a finite or infinite sequence $(\pi_i)_{0 \leq i \leq t}$ of partitions, i.e., $t \in \NN \cup \{+\infty\}$, together with a corresponding sequence $(a_i)_{1 \leq i \leq t}$ of deviating agents, such that for every $1 \leq i \leq t$, it holds that $\pi_{i-1} \overset{a_i\ }{\rightarrow}_\chi \pi_{i}$, i.e., $\pi_i$ evolves from $\pi_{i-1}$ by a $\chi$ deviation of~$a_i$. 
We say that an execution of the $\chi$ dynamics of $(G, \startpart)$ \emph{converges} if $\pi_t$ is a $\chi$ partition.

We say that the $\chi$ dynamics of $(G, \startpart)$ \emph{possibly converges} if some execution of $(G, \startpart)$ converges.
Moreover, we say that the $\chi$ dynamics of $(G, \startpart)$ 
\emph{necessarily converges} if every execution of the $\chi$ dynamics of $(G, \startpart)$ is finite.
This means that we necessarily reach a $\chi$ partition if we continue applying $\chi$ deviations.
By contrast, if the $\chi$ dynamics of $(G, \startpart)$ 
does not converge necessarily, there have to be executions where the same partition is reached infinitely often.
In this case, we say that the dynamics \emph{cycles}.

As computational decision problems, possible and necessary convergence can be captured as follows.

\problemdefijcai
{Possible Convergence of Dynamics ($\chi$-PCD)}
{A hedonic game $G$ and a starting partition $\startpart$.}
{Is there a sequence of $\chi$ deviations on $G$ that results in a $\chi$ partition when starting from~$\startpart$?}

\problemdefijcai
{Necessary Convergence of Dynamics ($\chi$-NCD)}
{A hedonic game $G$ and a starting partition $\startpart$.}
{Is every sequence of $\chi$ deviations on $G$ finite when starting from $\startpart$?
}

Typically, we consider $\chi$-\PCD and $\chi$-\NCD for a specific class of hedonic games, such as \ASHGs.

We conclude with the simple observation that \CIS dynamics necessarily converge.
This follows immediately because we operate on a finite game and every \CIS deviation increases the utilitarian welfare $\sum_{a \in N} \uf_a(\pi)$ \citep{ABS11c}.

\begin{observation}\label{obs:CISconvergence}
    Every execution of the \CIS dynamics converges necessarily.
\end{observation}

\section{Presentation of Meta Theorems}

We now present our meta theorems.
A proof sketch can be found in \Cref{sec:proofsketch} and the full proof is provided in \Cref{app:proofs}.
Our first theorem states that the existence of a cycling dynamics implies hardness of deciding about possible convergence of dynamics.

\begin{restatable}{theorem}{unifiedPCD}\label{thm:unified_PCD}
    Let $\chi$ be a standard stability notion with $\chi \devimplies \NS$ and $\CIS \devimplies \chi$.
    Assume that there exists an \ASHG, \FHG, or \MFHG $G_{\chi}$ and partition $\pi_\chi$ such that the $\chi$ dynamics of $(G_\chi, \pi_\chi)$ must cycle. Then $\chi$-\PCD is \NP-hard for the game class of $G_{\chi}$ (e.g., for \ASHGs if $G_{\chi}$ is an \ASHG).
\end{restatable}

Moreover, if there exists an instance in which the dynamics can cycle but necessarily converge after the removal of a \emph{singleton} coalition, we obtain hardness of deciding about necessary convergence of dynamics.

\begin{restatable}{theorem}{unifiedNCD}\label{thm:unified_NCD}
    Let $\chi$ be a standard stability notion with $\chi \devimplies \NS$ and $\CIS \devimplies \chi$.
    Assume that there exists an \ASHG, \FHG, or \MFHG $G_{\chi}$ and partition $\pi_\chi$ that contains a singleton coalition $\{a\} \in \pi_\chi$, such that the $\chi$ dynamics can cycle on $(G_\chi, \pi_\chi)$, but necessarily converge on $(G_\chi - a, \pi_\chi - a)$. 
    Then $\chi$-\NCD is \coNP-hard for the game class of $G_{\chi}$.
\end{restatable}

The precondition for the required game in \Cref{thm:unified_NCD} may seem intricate, but it is quite weak.
For instance, it is satisfied whenever there exists a game in which the dynamics starting from the singleton partition can cycle.
Indeed, in this case, one can obtain the desired game by iteratively removing agents until the dynamics from the singleton coalition necessarily converges.
Then, the penultimate game in this procedure satisfies the prerequisites of \Cref{thm:unified_NCD}.
Moreover, both theorems hold whenever there exists an instance without a stable partition.
In this case, the dynamics from any starting partition (e.g., the singleton partition) must cycle.
We state the latter observation in the following corollary.

\begin{corollary}
    Let $\chi$ be a standard stability notion with $\chi \devimplies \NS$ and $\CIS \devimplies \chi$.
    Assume that there exists an \ASHG, \FHG, or \MFHG $G_{\chi}$ without a $\chi$ partition.
    Then, $\chi$-\PCD{} is \NP-hard and $\chi$-\NCD is \coNP-hard for the game class of~$G_{\chi}$.
\end{corollary}

We can directly apply our corollary for established stability notions of which it is known that instances without stable partitions exist.
For instance, there exist ASHGs without an IS or CNS (and, therefore, no NS) partition (\citealp{BoJa02a}, Example~5; \citealp{SuDi07b}, Example~2).
Our meta theorems (\Cref{thm:unified_PCD,,thm:unified_NCD}) apply uniformly to all standard stability notions, including \NS, \IS, \CNS, \VIS, \VOS, and \SMS in \ASHGs, \FHGs, and \MFHGs.
All of these also follow from \Cref{thm:consequences} below.

In fact, we now 
demonstrate the applicability of our meta theorems for any 
deviation concept between \NS deviations and voting-based notions weaker than \CIS deviations.
More precisely, consider $(\qout, \qin)$-\SVS for any $\qout, \qin \in [0, 1]$.
In case that $\qout = \qin = 1$, this is \CIS, for which dynamics necessarily converge (\Cref{obs:CISconvergence}).
In all other cases, we show that \Cref{thm:unified_NCD,thm:unified_PCD} can be applied for all three game classes.
We thus obtain a dichotomy that separates \CIS from other voting-based stability notions.

\begin{restatable}{theorem}{consequences}\label{thm:consequences}
    Let $\chi$ be a standard stability notion such that $\chi \devimplies \NS$ and $(\qout, \qin)$-\SVS $\devimplies \chi$ for some $\qout,\qin \in [0, 1]$.
    Then, $\chi$-\PCD{} is \NP-hard and $\chi$-\NCD is \coNP-hard for \ASHGs, \FHGs, and \MFHGs if $\qout < 1$ or $\qin < 1$.
\end{restatable}

The full proof of \Cref{thm:consequences} is presented in \Cref{app:voting}.
It relies on constructing two games for which we apply \Cref{thm:unified_PCD,thm:unified_NCD} once each.
We further distinguish whether for the relevant stability notion $\chi$ it holds that $(\qout, 1)$-\SVS $\devimplies \chi$ or $(1, \qin)$-\SVS $\devimplies \chi$.
All constructed games consist of a large set of deviating agents and a small set of gadget agents that never perform deviations (and, in fact, their valuation function is the $0$-function, under which all coalitions yield an identical utility). 
Starting from a predetermined partition, there always exists precisely one deviating agent that can perform a permissible $\chi$ deviation, while no other deviation is possible that is even an \NS deviation. 
Performing this deviation yields a partition that is identical up to a permutation of agents.
Hence, we establish inevitable cycling, and, therefore, games suitable to apply \Cref{thm:unified_PCD}.
The starting partitions can then be turned into partitions satisfying the preconditions of \Cref{thm:unified_NCD} by removing the first deviator from her coalition and placing her in a singleton coalition.

\section{Proof Sketch of Meta Theorems}\label{sec:proofsketch}

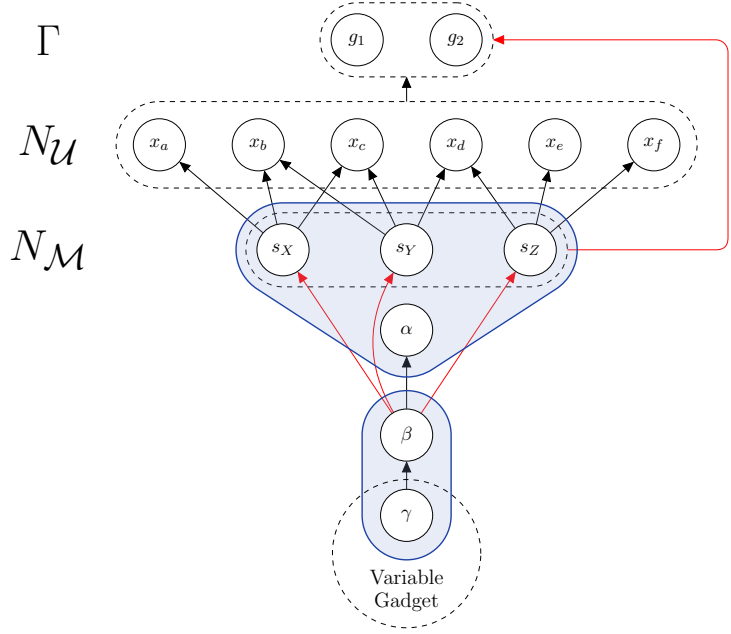
\begin{figure}[t]
        \centering
            \resizebox{.6\textwidth}{!}{
                \begin{tikzpicture}
	\begin{pgfonlayer}{nodelayer}
		\node [style=agent] (0) at (0, -8.5) {$\gamma$};
		\node [style=agent] (1) at (0, -5.25) {$\beta$};
		\node [style=agent] (2) at (0, -1) {$\alpha$};
		\node [style=agent] (3) at (-5, 2.25) {$s_X$};
		\node [style=agent] (4) at (0, 2.25) {$s_Y$};
		\node [style=agent] (5) at (5, 2.25) {$s_Z$};
		\node [style=agent] (6) at (-10, 6.5) {$x_a$};
		\node [style=agent] (7) at (-6, 6.5) {$x_b$};
		\node [style=agent] (8) at (-2, 6.5) {$x_c$};
		\node [style=agent] (9) at (2, 6.5) {$x_d$};
		\node [style=agent] (10) at (6, 6.5) {$x_e$};
		\node [style=agent] (11) at (10, 6.5) {$x_f$};
		\node [style=agent] (12) at (-2, 10.75) {$g_1$};
		\node [style=agent] (13) at (2, 10.75) {$g_2$};
		\node [style=none] (14) at (-14.5, 10.75) {\Huge$\Gamma$};
		\node [style=none] (15) at (-14.5, 6.5) {\Huge$N_{\mathcal{U}}$};
		\node [style=none] (16) at (-14.5, 2.25) {\Huge$N_{\mathcal{M}}$};
		\node [style=none] (17) at (-2, 12.25) {};
		\node [style=none] (18) at (-2, 9.25) {};
		\node [style=none] (19) at (2, 9.25) {};
		\node [style=none] (20) at (2, 12.25) {};
		\node [style=none] (22) at (-3.5, 10.75) {};
		\node [style=none] (23) at (3.5, 10.75) {};
		\node [style=none] (25) at (-10, 8.25) {};
		\node [style=none] (26) at (-10, 4.75) {};
		\node [style=none] (27) at (10, 4.75) {};
		\node [style=none] (28) at (10, 8.25) {};
		\node [style=none] (33) at (-5, 3.75) {};
		\node [style=none] (34) at (-5, 0.75) {};
		\node [style=none] (35) at (5, 0.75) {};
		\node [style=none] (36) at (5, 3.75) {};
		\node [style=none] (37) at (-6.5, 2.25) {};
		\node [style=none] (43) at (0, -11) {\large{Variable}};
		\node [style=none] (44) at (0, -12) {\large{Gadget}};
		\node [style=none] (45) at (13, 10.25) {};
		\node [style=none] (46) at (12.5, 2.25) {};
		\node [style=none] (50) at (6.5, 2.25) {};
		\node [style=none] (51) at (13, 2.75) {};
		\node [style=none] (52) at (12.5, 10.75) {};
		\node [style=none] (55) at (-11.75, 6.5) {};
		\node [style=none] (57) at (0, 8.25) {};
		\node [style=none] (58) at (0, 9.25) {};
		\node [style=none] (59) at (11.75, 6.5) {};
        \node[draw, circle, dashed,minimum size=3cm] at ($(0)!.5!(44)+(0,.2)$) {};
	\end{pgfonlayer}
	\begin{pgfonlayer}{edgelayer}
		\draw [style=utility] (1) to (2);
		\draw [style={negative_incentive}, bend left] (1) to (4);
		\draw [style={negative_incentive}] (1) to (5);
		\draw [style={negative_incentive}] (1) to (3);
		\draw [style=utility] (0) to (1);
		\draw [style=utility] (3) to (8);
		\draw [style=utility] (3) to (7);
		\draw [style=utility] (3) to (6);
		\draw [style=utility] (4) to (7);
		\draw [style=utility] (4) to (8);
		\draw [style=utility] (4) to (9);
		\draw [style=utility] (5) to (9);
		\draw [style=utility] (5) to (10);
		\draw [style=utility] (5) to (11);
		\draw [style=ddashed] (20.center)
			 to [bend left=45] (23.center)
			 to [bend left=45] (19.center)
			 to (18.center)
			 to [bend left=45] (22.center)
			 to [bend left=45] (17.center)
			 to cycle;
		\draw [style=ddashed] (25.center)
			 to (28.center)
			 to [bend left=45] (59.center)
			 to [bend left=45] (27.center)
			 to (26.center)
			 to [bend left=45] (55.center)
			 to [bend left=45] cycle;
		\draw [style=ddashed] (36.center)
			 to [bend left=45] (50.center)
			 to [bend left=45] (35.center)
			 to (34.center)
			 to [bend left=45] (37.center)
			 to [bend left=45] (33.center)
			 to cycle;
		\draw [style={negative_incentive}] (50.center)
			 to (46.center)
			 to [bend right=45, looseness=1.25] (51.center)
			 to (45.center)
			 to [bend right=45] (52.center)
			 to (23.center);
		\draw [style=utility] (57.center) to (58.center);

        \draw[thick,CoalitionColor, fill=CoalitionColor!50, fill opacity=0.2]  \convexpath{2, 3, 5}{1.9cm};
        \draw[thick,CoalitionColor, fill=CoalitionColor!50, fill opacity=0.2]  \convexpath{0, 1}{1.8cm};
	\end{pgfonlayer}
\end{tikzpicture}
                }
        
        \caption{Illustration of the reduction. 
        A covering instance $(\mathcal{U}, \mathcal{M})$ is represented by agents $N_{\mathcal U}$ and $N_{\mathcal{M}}$.
        Here, we have $\mathcal{U} = \{a, \ldots, f\}$ and $\mathcal{M} = \{X, Y, Z\}$ with $X = \{a, b, c\}$, $Y = \{b, c, d\}$, and $Z = \{d, e, f\}$.
        Black and red arrows indicate potential utility increases and decreases, respectively.
        Important coalitions of the starting partition are indicated in blue. 
        In Yes-instances, dynamics can lead to agent $\gamma$ ending up in a singleton coalition.\label{fig:meta-sketch}}
\end{figure}

In this section, we outline the proofs of \Cref{thm:unified_PCD,thm:unified_NCD}.
Both rely on a reduction from \textsc{Restricted Exact Cover By $3$-Sets} (\RXC).
An instance of \RXC consists of a finite set of elements $\mathcal{U} = \{e_1, \ldots, e_{3 h}\}$ 
and a family $\mathcal{M} = \{M_1, \ldots, M_{3 h}\}$ subsets of $\mathcal{U}$ of size~$3$ such that every element of $\mathcal{U}$ belongs to exactly three sets in $\mathcal{M}$.
An instance is a Yes-instance if and only if there is a selection of exactly $h$ sets from $\mathcal{M}$ whose union is $\mathcal{U}$. 
\RXC is known to be \NP-complete \citep{Karp72a,Gonz85a}.

Both proofs are performed in two steps: first, we encode the combinatorial structure of an \RXC instance as deviation dynamics, then we use the games assumed by the respective theorem as a gadget.
The first step is the same for both theorems and is outlined in \Cref{fig:meta-sketch}.
Given an instance $(\mathcal{U}, \mathcal{M})$ of \RXC, we introduce sets $N_{\mathcal U}$ and $N_{\mathcal{M}}$ of element agents and set agents representing $\mathcal{U}$ and $\mathcal{M}$, respectively.
Set agents receive a positive utility from the element agents corresponding to their contained elements.
At the top, there is a set $\Gamma$ of grouping agents, identical in size to the number of sets in an exact cover, e.g., $2$ agents if $|\mathcal U| = 6$.
Further down, there are special agents $\alpha$ and $\beta$.
The latter has a very high valuation for~$\alpha$ but dislikes set agents. 
At the bottom, there is a variable gadget containing a dedicated agent $\gamma$ who is the only agent that can interact with the other gadget agents through deviations.

In \Cref{fig:meta-sketch}, black arrows indicate deviation incentives, while red arrows represent deviation obstacles.
The two important coalitions of the starting partition are indicated in blue.
Generally, agents perform deviations ``upwards.'' 

Element agents can freely join the coalitions of grouping agents, which can in principle lead to coalitions containing any grouping agent and any subset of element agents.
However, set agents can only join the coalition of a grouping agent if it contains exactly the agents corresponding to its contained elements.\footnote{Initially, coalitions of element agent contain an additional restricting agent that prevents set agents from joining.
These are omitted from the figure for simplicity.}
Once this happens, a coalition is created from and towards which no more deviations happen.

Over time, the coalition of $\alpha$ contains less and less set agents.
This allows $\beta$ to join this coalition
if and only if the deviated set agents correspond to an exact cover of $\mathcal U$. 
This in turn allows the abandoned $\gamma$ to engage in deviations within the gadget.
In the deviation sequence up to this step, almost all performed deviations are \CIS deviations and, therefore, $\chi$ deviations.
The only deviation that is possibly not a \CIS deviation is when $\beta$ joins $\alpha$.
When performing this deviation, it is the only time in the proof that we use that we need a standard stability notion.

By specifying the variable gadget, we can leverage this general reduction to prove \Cref{thm:unified_PCD,thm:unified_NCD}.
For possible convergence, we use 
the game in which cycling dynamics must happen.
For each of the coalitions of the starting partition causing necessary cycling, we append a copy of the construction in \Cref{fig:meta-sketch}.
If the source instance was a No-instance, then agents of type~$\gamma$ (in the multiple copies) never end up in singleton coalitions.
Hence, the gadget agents have to cycle inevitably.
If, however, the source instance was a Yes-instance, then agents of type~$\gamma$ can join the coalitions from the gadget with \CIS deviations, leading to a stable partition.
Hence, dynamics possibly converge if and only if the source instance was a Yes-instance.

We now turn to necessary convergence. 
Note that \coNP-hardness for necessary convergence is identical to \NP-hardness of the question whether dynamics possibly cycle.
We now use the possibly cycling game with its dedicated agent~$a$ as a variable gadget and 
identify $\gamma$ with~$a$.
Hence, if the source instance was a No-instance, dynamics can never change the coalition of $a$, and, therefore, dynamics have to converge in the variable gadget.
Otherwise, if the source instances was a Yes-instance, agent $a$ can initiate cycling once she is in a singleton coalition.

\section{Contractual Individual Stability}\label{sec:CIS}

As \CIS dynamics necessarily converge in any hedonic game (cf.~\Cref{obs:CISconvergence}), 
\CIS-\PCD and \CIS-\NCD are trivially polynomial-time solvable.
Moreover, \citet{ABS11c} provide an algorithm to compute \emph{some} \CIS partition in polynomial time for \ASHGs.\footnote{\citet{BDEG25b} correct an inaccuracy in this algorithm.}
Unfortunately, their algorithm 
fails to produce partitions that satisfy individual rationality, i.e., some agents might have a large negative utility. 
Notably, as \CIS deviations preserve individual rationality, \CIS dynamics from the singleton coalition guarantee the existence of individually rational \CIS partitions.

\begin{observation}\label{obs:CIS_dyn_maintain_IR}
    Let $G$ be a hedonic game together with an individually rational partition  $\startpart$. 
    Then, any execution of the \CIS dynamics of $(G,\startpart)$ converges to an individually rational \CIS partition.
\end{observation}

By contrast, it is \NP-hard to decide whether \CIS dynamics lead to individually rational outcomes, when starting from a general partition.
This result holds even for fairly restricted valuations, e.g., to $\{-1,1\}$.
We defer all missing proofs in this section to \Cref{app:CIS}.

\begin{restatable}{theorem}{CisIrNp}\label{thm:CIS_IR_NP}
    Let $f^+: \NN \rightarrow \QQ^+$ and $f^- : \NN \rightarrow \QQ^-$ be two functions with $f^+(n) \geq |f^-(n)|$ for all $n \in \NN$. It is \NP-hard to decide whether the \CIS dynamics in an \ASHG, \FHG, or \MFHG 
    can converge to an individually rational partition from a given starting partition $\pi$, even when valuations are restricted to $\{f^-(n), f^+(n)\}$ for games with $n$ agents.
\end{restatable}

Hence, for each hedonic game, 
one can compute an individually rational \CIS partition by running \CIS dynamics from the singleton partition. 
However, it is not clear whether one can efficiently find a short converging sequence of \CIS deviations, i.e., a sequence that consists of polynomially many steps. 
We, therefore, dedicate the remainder of this section to this question, and focus our attention on \ASHGs.

First, we show that short converging sequences taking a linear number of \CIS deviations always exist.

\begin{theorem}\label{thm:CIS_dyn_shortcut_singleton}
    Let $G$ be an \ASHG and let $\pi$ be a \CIS partition that was reached through an execution of the \CIS dynamics on $G$ when starting from the singleton partition. 
    Then $\pi$ can be reached 
    from the singleton partition after exactly $|N| - |\pi|$ \CIS deviations. 
\end{theorem}

\begin{proof}
    Consider an execution of the \CIS dynamics on $G$ when starting from the singleton partition. 
    Our proof relies on the following claim which is proved in the appendix. 

    \begin{restatable}{claim}{CISshortcut}\label{claim:CIS_dyn_shortcut_singleton_exactly_one}
        Every coalition $C$ in $\pi$ contains exactly one agent that never deviated in the execution 
        of the \CIS dynamics.
    \end{restatable}

    We denote the agents that never deviate to reach $\pi$ as per \Cref{claim:CIS_dyn_shortcut_singleton_exactly_one} as the \emph{owners} of their respective coalitions in $\pi$.
    Moreover, given an arbitrary agent $a \in N$, we denote by $o_a$ the owner of the coalition $\pi(a)$.
    Now, given the original (possibly exponential length) sequence of \CIS deviations that resulted in $\pi$, consider the last deviation of each agent. 
    We construct a new, shortened sequence of $|N| - |\pi|$ deviations, where each agent $a$ that is not the owner of a coalition performs exactly one deviation from her singleton coalition to join $o_a$. 
    We order this new deviation sequence by when the agents performed their last deviation in the original sequence.
    It is clear that this new deviation sequence results in the same partition $\pi$ after exactly $|N| - |\pi|$ steps.
    
    It remains to show that the new sequence consists only of \CIS deviations. 
    As each agent deviates from her singleton coalition, no agent will ever be blocked from leaving. 
    Now, given a nonowner agent $a$, let $C_{\mathrm{new}}$ be the coalition that she joins in the new sequence and let $C_{\mathrm{ori}}$ be the coalition that she joins in the original sequence.
    Observe that $C_{\mathrm{new}} \subseteq C_{\mathrm{ori}}$ must hold. 
    Then, $a$ not being blocked from joining $C_{\mathrm{new}}$ directly follows from the fact that the original sequence consists only of \CIS deviations. 
    Further, in case there exists an agent $b \in C_{\mathrm{ori}} \setminus C_{\mathrm{new}}$ with $\vf_a(b) > 0$, then $b$ must have deviated from a coalition that contains $a$ in the original sequence, which cannot have been a \CIS deviation. 
    Hence, $\vf_a(b)\le 0$, and thus $\uf_a(C_{\mathrm{new}})\ge \uf_a(C_{\mathrm{ori}}) > 0$ must hold, where the strict inequality follows because $C_{\mathrm{ori}}$ was reached in a \CIS dynamics starting from the singleton partition by a deviation of $a$.
    Therefore, the deviation of $a$ is a \CIS deviation.
    Since $a$ was chosen arbitrarily, this concludes the proof. 
\end{proof}

An additional observation from the last theorem is that in the constructed dynamics
every agent deviates at most once.
However, finding this sequence needed knowledge of a possibly much longer sequence.
This raises the question whether all \CIS dynamics starting from the singleton partition are short.
We answer this question negatively by 
constructing a family of instance where \CIS dynamics can have exponential length
with respect to the game size.

\begin{restatable}{theorem}{CisExp}\label{thm:CIS_EXP}
    Let $\chi$ be a stability notion with $\CIS \devimplies \chi$. 
    Then the $\chi$ dynamics starting from the singleton partition may take an exponential number of steps
    with respect to the game's input size. 
\end{restatable}

It remains an interesting open problem to determine the complexity of computing an individually rational \CIS partition (even without using dynamics).
We make first progress towards this question by identifying the structural reason behind \Cref{thm:CIS_EXP}.
The games constructed in its proof
heavily rely on valuations that are positive in one direction but~$0$ in the other. 
If we bound the number of agents with such valuations, we can efficiently compute individually rational \CIS partitions.
To this end, for an \ASHG $G = (N,\uf)$, define $s(G) := |\{a \in N \mid \exists b \in N : \vf_a(b) > 0 \land \vf_b(a) = 0\}|$.

The proof idea is as follows.
We construct the desired \CIS dynamics in three phases.
Define $X := \{a \in N \mid \exists b \in N : \vf_a(b) > 0 \land \vf_b(a) = 0\}$, i.e., $|X| = s(G)$.
In the first phase, the agents not in $X$ deviate.
After at most one deviation each, a partition is reached in which these agents cannot deviate again.
In the second phase, agents in $X$ deviate at most once, joining best coalitions containing agents not in $X$.
The first two phases comprise at most $n$ deviations.
In the third phase, arbitrary \CIS deviations are performed.
It can be shown that, after the second phase, deviations can only be performed by agents in $X$, joining other agents in $X$.
Hence, this can lead to at most $s(G)^{s(G)}$ unique partitions.

\begin{restatable}{theorem}{CISffpt}\label{thm:CIS_FFPT_asymm_0}
    An execution of the \CIS dynamics starting from the singleton partition taking at most $s(G)^{s(G)}+n$ deviations can be computed in polynomial time with respect to the game's input size.
\end{restatable}

\section{Conclusion}

We presented a meta approach to determine the computational complexity of deciding whether the deviation dynamics possibly or necessarily converge in a hedonic game based on the mere existence of simple No-instances.
Our results encompass all standard stability notions based on deviations between \NS and \CIS deviations. 
Moreover, they hold for the prominent game classes of additively separable, fractional, and modified fractional hedonic games.
We also investigated the computational complexity of finding an individually rational \CIS partition in an \ASHG. Here, dynamics may converge in a linear number of steps, but we can only efficiently extract the deviations for fast convergence when restricting the number of certain valuation pairs.

Natural directions for future work include reevaluating our hardness results for restricted domains of valuations, such as, utilities based on \emph{friend-and-enemy} evaluations \citep{DBHS06a}, different classes of hedonic games, including ordinal models, 
or stability notions that rely on group deviations. Further, while \citet{BBK23a} discuss the structure of outcomes and running time of simulations for \NS dynamics, an interesting direction would be a comprehensive experimental evaluation for a broader set of stability notions.
Finally, an intriguing open question is the computational complexity of computing an individually rational \CIS partition, and the applicability of our established results to game classes other than \ASHGs.

\section*{Acknowledgments}

Most of this work was done when Martin Bullinger was at the University of Oxford.
Martin Bullinger was supported by the AI Programme of The Alan Turing Institute.

\appendix

\section*{Appendix}

In the appendix, we provide additional material, such as missing proofs.

\section{Proof of Proposition~\ref{prop:VSstandard}}\label{app:VSstandard}

In this appendix, we provide the proof that voting-based stability notions are standard stability notions. 

\VSstandard* 

\begin{proof}
    Let $\qout, \qin \in [0,1]$.
    Consider an agent $a\in N$ and a single-agent deviation $\pi \overset{a}{\rightarrow} \pi'$.
    For $b\in N$, define $\Delta\mathit{uc}_G(b, \pi, \pi') := \uf_b(\pi') - \uf_b(\pi)$, which only depends on $\mathit{uc}_G(b, \pi, \pi')$.
    
    Define $f: \mathcal{UC} \times \mathcal{UC} \times (\QQ \times \QQ) \rightarrow \{0, 1\}$ such that for all $X, Y \in \mathcal{UC}$, and $z \in \QQ \times \QQ$, we have that 
    $f(X,Y,z) = 1$ if and only if
    \begin{itemize}
        \item $\Delta z > 0$,
        \item $|\{x\in X\colon \Delta x > 0\}|\ge \qout(|\{x\in X\colon \Delta x < 0\}| + |\{x\in X\colon \Delta x > 0\}|)$, and 
        \item $|\{y\in Y \colon \Delta y > 0\}| \ge \qin (|\{y\in Y\colon \Delta y < 0\}| + |\{y\in Y\colon \Delta y > 0\}|)$.
    \end{itemize}

    We first show that $f$ precisely encapsulates $(\qout,\qin)$-\SVS.
    Note that it holds that 
    \begin{itemize}
        \item $\Fout(\pi(a), a) = |\{x\in \mathit{UC}_G^{\mathrm{out}}(a, \pi, \pi')\colon \Delta x > 0\}|$,
        \item $\Fin(\pi(a), a) = |\{x\in \mathit{UC}_G^{\mathrm{out}}(a, \pi, \pi')\colon \Delta x < 0\}|$,
        \item $\Fout(\pi'(a), a) = |\{x\in \mathit{UC}_G^{\mathrm{in}}(a, \pi, \pi')\colon \Delta x < 0\}|$, and 
        \item $\Fin(\pi'(a), a) = |\{x\in \mathit{UC}_G^{\mathrm{in}}(a, \pi, \pi')\colon \Delta x > 0\}|$.
    \end{itemize}
    Hence, $\pi \overset{a}{\rightarrow} \pi'$ is a $(\qout,\qin)$-\SVS deviation if and only if $f(\Fout(\pi(a), a),\Fin(\pi(a), a),\mathit{uc}_G(a, \pi, \pi')) = 1$.
    It follows that $(\qout,\qin)$-\SVS is an anonymously hedonic stability notion.

    Moreover, consider $X, X', Y, Y' \in \mathcal{UC}$, and $z, z' \in \QQ \times \QQ$ such that $X \trianglelefteq X'$, $Y \trianglelefteq Y'$, and $\{z\} \trianglelefteq \{z'\}$. 
    If $f(X,Y,z) = 0$, then $f(X,Y,z) \le f(X',Y',z')$ is immediate.
    Assume, therefore, that $f(X,Y,z) = 1$.
    Since $\{z\} \trianglelefteq \{z'\}$, it holds that $\Delta z'\ge \Delta z > 0$.
    The second condition in the definition of $f$ holds for $X'$ if $\qout = 0$.
    If $\qout > 0$, then
    $\{x\in X\colon \Delta x < 0\} = \emptyset$ or there exists $\hat x\in \{x\in X\colon \Delta x > 0\}$.
    Hence, $\max\{\Delta x\colon x\in X\}\ge 0$.
    Since, $X \trianglelefteq X'$, it follows that $\min\{\Delta x'\colon x'\in X'\}\ge 0$, and therefore the second condition in the definition of $f$ is satisfied for $X'$.
    
    Finally, the third condition in the definition of $f$ is satisfied for $Y'$ by an analogous argument.
    We conclude that $f(X',Y',z') = 1$.
    Hence, $f$ is monotone.
\end{proof}

\section{Proof of Theorems~\ref{thm:unified_PCD} and \ref{thm:unified_NCD}}\label{app:proofs}

In this section, we will provide the full proof of \Cref{thm:unified_PCD,thm:unified_NCD}.
Both proofs use the same overall construction, which we will introduce first, and analyze in subsequent lemmas. 
The reduction is from \problemname{Exact Cover by Three Sets (\XTC)}, which is defined as follows. 

\problemdefijcai
{Exact Cover By $3$-Sets (\XTC)}
{A finite set of elements $\mathcal{U}$ and a family $\mathcal{M}$ of subsets of $\mathcal{U}$ of size~$3$.}
{Is there a selection of exactly $\nicefrac{|\mathcal{U}|}{3}$ sets from $\mathcal{M}$ whose union is $\mathcal{U}$, i.e., is there an exact cover of $\mathcal{U}$ with sets from $\mathcal{M}$?}

It is known that \XTC is \NP-complete \citep{Karp72a}.
We use the following variant 
that assumes further restrictions on the structure of the set $\mathcal{M}$. 
This variation is known to remain \NP-complete \citep{Gonz85a}.

\problemdefijcai
{Restricted Exact Cover By $3$-Sets (\RXC)}
{A finite set of elements $\mathcal{U} = \{e_1, \ldots, e_{3 h}\}$ 
and a family $\mathcal{M} = \{M_1, \ldots, M_{3 h}\}$ of subsets of $\mathcal{U}$ of size~$3$ such that every element of $\mathcal{U}$ belongs to exactly three sets in $\mathcal{M}$.}
{Is there a selection of exactly $h$ sets from $\mathcal{M}$ whose union is $\mathcal{U}$, i.e., is there an exact cover of $\mathcal{U}$ with sets from $\mathcal{M}$?}

Throughout the remaining section, we assume that $\chi$ is a standard stability notion such that $\chi \devimplies \NS$ and $\CIS \devimplies \chi$.

\subsection{Reduction from \RXC}
\label{sec:hardnessproof:construction}

\begin{figure*}[ht!]
        \centering
        
        \begin{subfigure}[b]{\textwidth}
            \centering
            \resizebox{.85\textwidth}{!}{
                \begin{tikzpicture}
	\begin{pgfonlayer}{nodelayer}
		\node [style=agent] (0) at (0, -16) {$\gamma$};
		\node [style=agent] (1) at (0, -11) {$\beta$};
		\node [style=agent] (2) at (0, -2) {$\alpha$};
		\node [style=agent] (3) at (-10, 1) {$s_X$};
		\node [style=agent] (4) at (0, 1) {$s_Y$};
		\node [style=agent] (5) at (10, 1) {$s_Z$};
		\node [style=agent] (6) at (-16, 8) {$x_a$};
		\node [style=agent] (7) at (-12, 8) {$x_b$};
		\node [style=agent] (8) at (-4, 8) {$x_c$};
		\node [style=agent] (9) at (4, 8) {$x_d$};
		\node [style=agent] (10) at (12, 8) {$x_e$};
		\node [style=agent] (11) at (16, 8) {$x_f$};
		\node [style=agent] (12) at (-5, 16) {$g_1$};
		\node [style=agent] (13) at (5, 16) {$g_2$};
		\node [style=none] (14) at (-19.5, 16) {\Huge$\Gamma$};
		\node [style=none] (15) at (-19.5, 8) {\Huge$N_{\mathcal{U}}$};
		\node [style=none] (16) at (-19.5, 1) {\Huge$N_{\mathcal{M}}$};
	\end{pgfonlayer}
	\begin{pgfonlayer}{edgelayer}
		\draw [style=utility] (0) to node [midway, fill=white] {$T_A$} (1);
		\draw [style=utility, bend left=15] (1) to node [midway, fill=white, rotate=90] {$(h'+1) \cdot nT_A$} (2);
		\draw [style=utility, bend left=15] (2) to node [midway, fill=white] {$T_A$} (1);
		\draw [style=utility, bend left=45] (1) to node [midway, fill=white] {$-nT_A$} (3);
		\draw [style=utility, bend right=45] (1) to node [midway, fill=white] {$-nT_A$} (5);
		\draw [style=utility, bend left=45] (1) to node [midway, fill=white] {$-nT_A$} (4);
		\draw [style=utility] (3) to node [midway, fill=white] {$nT_A$} (6);
		\draw [style=utility] (3) to node [midway, fill=white] {$nT_A$} (7);
		\draw [style=utility] (3) to node [pos=0.35, fill=white] {$nT_A$} (8);
		\draw [style=utility] (4) to node [pos=0.35, fill=white] {$nT_A$} (7);
		\draw [style=utility] (4) to node [midway, fill=white] {$nT_A$} (8);
		\draw [style=utility] (4) to node [midway, fill=white] {$nT_A$} (9);
		\draw [style=utility] (5) to node [midway, fill=white] {$nT_A$} (9);
		\draw [style=utility] (5) to node [midway, fill=white] {$nT_A$} (10);
		\draw [style=utility] (5) to node [midway, fill=white] {$nT_A$} (11);
		\draw [style=utility, bend right, looseness=0.50] (3) to node [midway, fill=white] {$T_A$} (1);
		\draw [style=utility, bend left=45] (4) to node [midway, fill=white] {$T_A$} (1);
		\draw [style=utility, bend left, looseness=0.50] (5) to node [midway, fill=white] {$T_A$} (1);
		\draw [style=utility, bend left=15, looseness=1.25] (3) to node [midway, fill=white] {$-2nT_A$} (12);
		\draw [style=utility, bend left] (3) to node [midway, fill=white] {$-2nT_A$} (13);
		\draw [style=utility, bend right=15] (4) to node [pos=0.6, fill=white] {$-2nT_A$} (12);
		\draw [style=utility, bend left=15] (4) to node [pos=0.6, fill=white] {$-2nT_A$} (13);
		\draw [style=utility, bend right] (5) to node [midway, fill=white] {$-2nT_A$} (12);
		\draw [style=utility, bend right=15, looseness=1.25] (5) to node [midway, fill=white] {$-2nT_A$} (13);

        \draw[thick,CoalitionColor, fill=CoalitionColor!50, fill opacity=0.2]  \convexpath{0, 1}{1.8cm};
        \draw[thick,CoalitionColor, fill=CoalitionColor!50, fill opacity=0.2]  \convexpath{3, 5, 2}{1.8cm};
	\end{pgfonlayer}
\end{tikzpicture}
                }
            \caption{Schematic of the overall construction, where we omit the agents in $R \cup D \cup (A \setminus \{\gamma\})$ and the corresponding valuations. Further, we omit all outgoing valuations of agents in $N_\mathcal{U}$, and the negative valuations of set agents for element agents that do not belong to the corresponding sets. 
            Technically, $\vf_\beta(\alpha)$ is only well-defined for \RXC source instances, which we indicate by replacing $2h$ by $h'$. 
            One can assume $h' = |\mathcal{M}| - \nicefrac{|\mathcal{U}|}{3}$. We display the nonadapted instance, where the valuation between agents $\beta$ and $\gamma$ is not flipped.
            }
            \label{subfig:unified_reduction_a}
        \end{subfigure}
        
        \caption{Illustration of the reduction. 
        For the sake of simplicity, the depicted reduction is from \XTC instead of \RXC. However, the schematic is analogous apart from a small change to the valuation $\vf_\beta(\alpha)$. The reduced instance for the source instance $(\mathcal{U}, \mathcal{M})$ is displayed, where $\mathcal{U} = \{a, \ldots, f\}$, and $\mathcal{M} = \{X, Y, Z\}$ with $X = \{a, b, c\}$, $Y = \{b, c, d\}$, and $Z = \{d, e, f\}$. A directed edge from agent $p$ to agent $d$ represents the valuation $\vf_p(d)$. Whenever two or more displayed agents belong to the same coalition in the starting partition~$\startpart$, we indicate this by blue boxes.}
    \end{figure*}
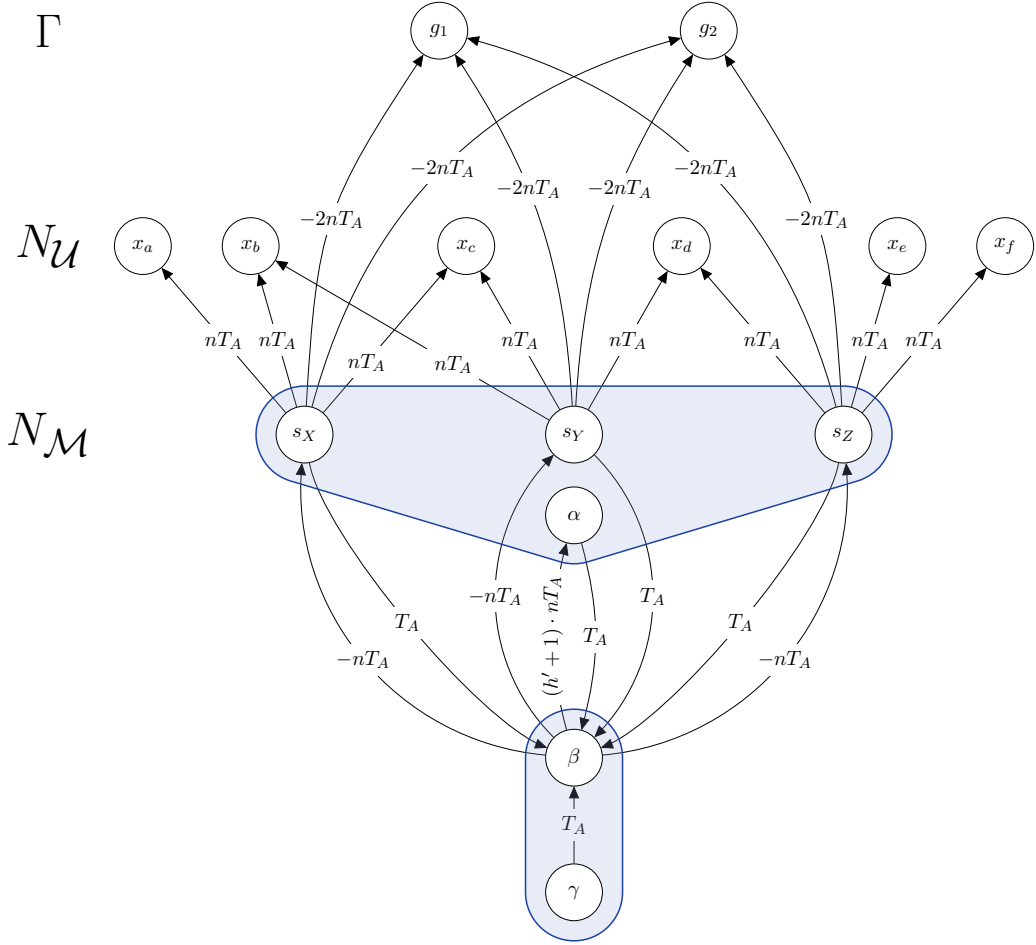
    \begin{figure*}[ht!]\ContinuedFloat
        \centering
        \begin{subfigure}[b]{.55\textwidth}
            \centering
            \resizebox{1\textwidth}{!}{
                \begin{tikzpicture}
	\begin{pgfonlayer}{nodelayer}
		\node [style=agent] (0) at (0, -7) {$s_X$};
		\node [style=agent] (2) at (-9, -1) {$x_a$};
		\node [style=agent] (3) at (9, -1) {$x_d$};
		\node [style=agent] (4) at (0, 10) {$g_1$};
		\node [style=agent] (5) at (-9, 4) {$r_a$};
		\node [style=agent] (6) at (9, 4) {$r_d$};
		\node [style=agent] (7) at (0, -12) {$\alpha$};
		\node [style=none] (8) at (-13, 10) {\Huge$\Gamma$};
		\node [style=none] (9) at (-13, 4) {\Huge$R$};
		\node [style=none] (10) at (-13, -1) {\Huge$N_\mathcal{U}$};
		\node [style=none] (11) at (-13, -7) {\Huge$N_\mathcal{M}$};
	\end{pgfonlayer}
	\begin{pgfonlayer}{edgelayer}
		\draw [style=utility-symmetric] (3) to (2);
		\draw [style=utility] (0) to node [pos=0.8, fill=white] {$-2nT_A$} (4);
		\draw [style=utility-symmetric] (0) to node [midway, fill=white] {$-nT_A$} (3);
		\draw [style=utility-symmetric] (0) to node [midway, fill=white] {$nT_A$} (2);
		\draw [style=utility] (0) to node [pos=.4, fill=white] {$-nT_A$} (6);
		\draw [style=utility] (0) to node [pos=.4, fill=white] {$-nT_A$} (5);
		\draw [style=utility] (3) to (4);
		\draw [style=utility] (2) to (4);
		\draw [style=utility] (2) to node [midway, fill=white] {$-n^2T_A$} (7);
		\draw [style=utility] (3) to node [midway, fill=white] {$-n^2T_A$} (7);
		\draw [style=dashed-utility] (3) to (5);
		\draw [style=dashed-utility] (2) to (6);

        \draw[thick,CoalitionColor, fill=CoalitionColor!50, fill opacity=0.2]  \convexpath{2, 5}{1.8cm};
        \draw[thick,CoalitionColor, fill=CoalitionColor!50, fill opacity=0.2]  \convexpath{3, 6}{1.8cm};
        \draw[thick,CoalitionColor, fill=CoalitionColor!50, fill opacity=0.2]  \convexpath{0, 7}{1.8cm};
	\end{pgfonlayer}
\end{tikzpicture}
                }
            \caption{Schematic of the construction with focus only on the agents in $\{g_1, r_a, r_d, x_a, x_b, s_X, \alpha\}$, where $a \in X$ but $d \not\in M$. Unlabeled edges represent a valuation of $1$, and dashed edges represent a valuation of $-1$.}
            \label{subfig:unified_reduction_b}
        \end{subfigure}
        
        \caption{Illustration of the reduction. For the sake of simplicity, the depicted reduction is from \XTC instead of \RXC. However, the schematic is analogous apart from a small change to the valuation $\vf_\beta(\alpha)$. The reduced instance for the source instance $(\mathcal{U}, \mathcal{M})$ is displayed, where $\mathcal{U} = \{a, \ldots, f\}$, and $\mathcal{M} = \{X, Y, Z\}$ with $X = \{a, b, c\}$, $Y = \{b, c, d\}$, and $Z = \{d, e, f\}$. A directed edge from agent $p$ to agent $d$ represents the valuation $\vf_p(d)$. Whenever two or more displayed agents belong to the same coalition in the starting partition~$\startpart$, we indicate this by blue boxes.}
        \label{fig:unified_reduction}
    \end{figure*}

    Consider an \RXC instance $\mathcal{I} = (\mathcal{U}, \mathcal{M})$, where $|\mathcal U| = 3h$. 
    The reduction is illustrated in \Cref{fig:unified_reduction}.
    We construct an \ASHG, \FHG, or \MFHG $G = (N, \vf)$ as follows. 
    We define the set $N$ of agents as $N_{\mathcal{M}} \cup N_{\mathcal{U}} \cup R \cup \Gamma \cup D \cup A \cup \{\alpha, \beta\}$, where 
    \begin{itemize}
        \item $N_{\mathcal{M}} = \{ s_M \}_{M \in \mathcal{M}}$ is a set of \emph{set agents}, 
        \item $N_{\mathcal{U}} = \{ x_e \}_{e \in \mathcal{U}}$ is a set of \emph{element agents},
        \item $R = \{r_e\}_{e \in \mathcal{U}}$ is a set of \emph{restricting agents},
        \item $\Gamma = \{g_i\}_{i \in [h]}$ is a set of \emph{grouping agents}, 
        \item $A$ is a set of \emph{gadget agents} which contains a dedicated agent $\gamma$, and
        \item $D$ is a (possibly empty) set of $\max(0, |A| - 2h)$ \emph{dummy agents}.
    \end{itemize}

    We assume that valuations among agents in $A$ are already defined. 
    Based on this, define $$t_{A} := 1 + \sum_{p, d \in A} |\vf_p(d)|\quad \text{and}\quad T_A := nt_A\text.$$ 
    There, $n$ is, as usual, the number of agents.
    
    We set $\vf_\gamma(\beta) = T_A$ and $\vf_\beta(\gamma) = 0$. 
    Moreover, we define valuations among the agents in $N \setminus A$ as follows:  
    \begin{enumerate}[1.]
        \item Let $\vf_\beta(\alpha) = (2h + 1) \cdot nT_A$, and $\vf_\beta(s_M) = -nT_A$ for every $M \in \mathcal{M}$.
        \item Let $\vf_\alpha(\beta) = T_A$.
        \item For every $d \in D$, let $\vf_d(\beta) = T_A$.
        \item For every $M \in \mathcal{M}$, $e \in M$, $e' \in \mathcal{U} \setminus M$, and $g \in \Gamma$, let 
        \begin{enumerate}[(a)]
            \item $\vf_{s_M}(x_e) = nT_A$, 
            \item $\vf_{s_M}(r_e) = \vf_{s_M}(r_{e'}) = \vf_{s_M}(x_{e'}) = -nT_A$,
            \item $\vf_{s_M}(g) = -2nT_A$, and
            \item $\vf_{s_M}(\beta) = T_A$.
        \end{enumerate}
        \item For every $e, e' \in \mathcal{U}$ with $e \neq e'$, $M, M' \in \mathcal{M}$ with $e \in M$ and $e \not\in M'$, and $g \in \Gamma$, let 
        \begin{enumerate}[(a)]
            \item $\vf_{x_e}(x_{e'}) = \vf_{x_e}(g) = 1$, 
            \item $\vf_{x_e}(r_{e'}) = -1$,
            \item $\vf_{x_e}(\alpha) = -n^2T_A$,
            \item $\vf_{x_e}(s_M) = nT_A$, and
            \item $\vf_{x_e}(s_{M'}) = -nT_A$.
        \end{enumerate}
    \end{enumerate}

    Based on the construction so far, we define $$t_{N \setminus A} := 1 + \sum_{p, d \in N \setminus A} |\vf_p(d)|\quad \text{and}\quad T_{N \setminus A} := nt_{N \setminus A}\text.$$ 
    
    We define valuations such that, with the exception of coalition $\{\beta, \gamma\}$, no two agents $(p, d) \in A \times (N \setminus A)$ can ever be part of a joint individually rational coalition, and we refer to these valuations by \emph{sub-game restricting valuations}. Specifically, for all $(p, d) \in A \times (N \setminus A)$ with $(p, d) \neq (\beta, \gamma)$, let $\vf_p(d) = -T_{A}$, and $\vf_d(p) = -T_{N \setminus A}$.

    Finally, we set all valuations that have not yet been specified to~$0$.

    For the reader's convenience, valuations are often chosen much larger than necessary for our construction to work as intended, simplifying many of our following arguments.

    We set the starting partition $\startpart$ to 
    \begin{align}
        & \{ \{g\} \mid g \in \Gamma\} \cup \{\{x_e, r_e\} \mid e \in \mathcal{U}\} \cup \notag\\
        & \{\{D \cup N_{\mathcal{M}} \cup \{\alpha\}\}, \{\beta, \gamma\}\} \cup \pi_A\text,\label{eq:startpart}
    \end{align}

    where $\pi_A$ is an arbitrary but fixed partition of the agents in $A \setminus \{\gamma\}$, whose exact composition is left open.
    
    Let $\Lambda \subseteq N_{\mathcal{M}}$ be an arbitrary subset of exactly $2h$ set agents, and consider the coalitions $C = \Lambda \cup D \cup \{\alpha\}$ and $\{\beta, \gamma\}$.
    It will later become crucial whether $\beta$ can perform a $\chi$ deviation from $\{\beta, \gamma\}$ to join a coalition of type $C$.
    First, note that the utility-change tuples that are relevant for whether $\beta$ can deviate towards $C$ do not depend on the composition of $C$ (they are identical regardless of the agent set~$\Lambda$).
    Hence, since we consider a standard stability notion, we can check whether such a deviation is permitted under our stability notion, by arbitrarily fixing $\Lambda$.
    
    Now, note that a deviation of $\beta$ leaving $\{\beta,\gamma\}$ and joining $C$ is an NS deviation because $\beta$ would increase her utility from $0$ to $nT_A$.
    However, the deviation might not be a $\chi$ deviation, e.g., because other involved agents might block the deviation.
    We can check this by applying the polynomial-time computable function $f_{\chi}$ associated to $\chi$ for the utility-change tuples of this deviation.
    In case that $\beta$ cannot perform a $\chi$ deviation to join $C$, we amend our construction by changing the valuations between $\beta$ and $\gamma$ and setting $\vf_\gamma(\beta) = 0$ and $\vf_\beta(\gamma) = T_A$ (we ``reverse'' the direction of the weighted edge between agents $\beta$ and $\gamma$).
    This will ensure that the deviation of $\beta$ would be a CIS deviation and hence a $\chi$ deviation.\footnote{The exact reason for this case distinction will become apparent from the proof later on, see \Cref{lem:unified_reduction_beta_can_deviate,lem:unified_reduction_gamma_doesnt_deviate}. 
    In particular, if we change the valuations, we will make use of the fact that $\beta$ was not allowed to perform the deviation to join $C$ before the alteration of valuations.}

    We remark that the reduced instance $(G, \startpart)$ can be constructed in polynomial time with respect to the source instance of \RXC, once the subgame induced by the agents in $A$ has been fixed.

\subsection{Investigation of Dynamics in the Reduced Instance}
\label{sec:hardnessproof:behavior}

We will now in detail consider dynamics in the reduced instance.
Throughout this section, we refer to $G$ as the reduced game and $\startpart$ the chosen starting partition.
We first capture the behavior induced by the sub-game restricting valuations.
Since these are sufficiently negative, no agent in $A$ performing an \NS deviation can ever join an agent in $N\setminus A$ and vice versa. 
We can use this to prove the following lemma.

\begin{lemma}\label{obs:unified_reduction_no_connecting_coalition}
    In every execution of the $\chi$ dynamics of $(G, \startpart)$, the only coalition ever containing an agent from $A$ and an agent from $N\setminus A$ 
    is $\{\beta, \gamma\}$. 
\end{lemma}

\begin{proof}
    Assume for contradiction that agents $p\in A$ and $d\in N\setminus A$ where $(d,p)\neq (\beta, \gamma)$ end in a joint coalition at some point.
    Recall that $\vf_p(d) = -T_{A}$, and $\vf_d(p) = -T_{N \setminus A}$.
    Hence, when this happens for the first time, the agent of $d$ and $p$ that performs the deviation does not increase her utility.
    Thus the performed deviation is not an \NS deviation.
    This is a contradiction as every $\chi$ deviation is an \NS deviation.
\end{proof}

The previous lemma implies that no agent can ever join the coalition $\{\beta,\gamma\}$.
Hence, this coalition can only change if either $\beta$ or $\gamma$ performs a deviation herself.
The next two lemmas reason about their deviations.
This is where we make use of the potential adaptation of the constructed game. 
It is important to recall that, whenever agent $\beta$ can perform a $\chi$ deviation in the original construction to leave $\{\beta, \gamma\}$ and join some coalition $C = \Lambda \cup D \cup \{\alpha\}$ with $\Lambda \subseteq N_{\mathcal{M}}$ that contains exactly $2h$ set agents, then we use this original construction. Otherwise, if this is not the case, then we adapt the original construction by ``reversing'' the direction of the weighted edge between $\beta$ and $\gamma$.
First, we consider deviations by $\beta$.

\begin{lemma}\label{lem:unified_reduction_beta_can_deviate}
    In every execution of the $\chi$ dynamics of $(G, \startpart)$, agent $\beta$ can deviate to leave a coalition that contains $\gamma$, to join a coalition $C = \Lambda \cup D \cup \{\alpha\}$, where $\Lambda \subseteq N_{\mathcal{M}}$ with $|\Lambda| = 2h$.
\end{lemma}
\begin{proof}
    By \Cref{obs:unified_reduction_no_connecting_coalition}, we can assume that the abandoned coalition of the deviation of $\beta$ is $\{\beta, \gamma\}$. 
    If we did not adapt $G$, then $\beta$ can perform the deviation by assumption. 

    Otherwise, we show that such a deviation of $\beta$ is a \CIS deviation and hence a $\chi$ deviation. To see this, observe that:
    \begin{enumerate}[1.]
        \item The deviation is an \NS deviation because it holds that $\uf_\beta(\{\beta, \gamma\}) \leq T_A$ (achieved in case of an \ASHG or \MFHG) while $\uf_\gamma(C \cup \{\beta\}) \geq \frac{(2h+1)\cdot nT_a - 2hnT_A}{n-1} > T_A$ (observe that $|C \cup \{\beta\}| < n$).
        \item Agent $\gamma$ has a valuation of $0$ for $\beta$, and is, therefore, indifferent between coalitions $\{\gamma\}$ and $\{\beta, \gamma\}$.
        \item For each agent $p \in C$, it holds that $\uf_p(C) = 0$ while $\uf_p(C \cup \{\beta\}) \geq \frac{T_A}{n} > 0$. \qedhere
    \end{enumerate}
\end{proof}

Next, we consider deviations of $\gamma$ abandoning $\beta$.

\begin{lemma}\label{lem:unified_reduction_gamma_doesnt_deviate}
    In every execution of the $\chi$ dynamics of $(G, \startpart)$, agent $\gamma$ cannot deviate to leave a coalition that contains $\beta$.
\end{lemma}
\begin{proof}
    Again, because of \Cref{obs:unified_reduction_no_connecting_coalition}, for a deviation of $\gamma$ away from $\beta$, we can assume that the abandoned coalition of the deviation of $\gamma$ is $\{\beta, \gamma\}$, and the welcoming coalition is some $C_A \subseteq A$.
    
    In case we made no adaptation to $G$, it holds that $\uf_\gamma(\{\beta, \gamma\}) \geq \frac{T_A}{2} > t_A$, and, by definition of $t_A$, agent $\gamma$ can never improve her utility by deviating to coalition $C_A$.

    Otherwise, we show the statement by virtue of $\chi$ being a standard stability notion and $\beta$ not being able to deviate from $\{\beta, \gamma\}$ in the nonadapted construction to join a coalition $C = \{\Lambda \cup D \cup \{\alpha\}\}$ with $\Lambda \subseteq N_\mathcal{M}$ and $|\Lambda| = 2h$.

    Let $\bar G$ and $\bar G'$ refer to the nonadapted and adapted game, respectively. 
    Further, let $\bar \pi$ be a partition of $N$ with $\{\{\beta, \gamma\}, C, C_A\} \subseteq \bar \pi$, and let $\bar \pi^\beta_\rightarrow$ and $\bar \pi^\gamma_\rightarrow$ be the partitions that resulted from $\bar \pi$ after $\beta$ and $\gamma$ performed their deviations in $\bar G$ and $\bar G'$ to join $C$ and $C_A$, respectively. Let us compare the relevant utility-change multisets:

    \begin{enumerate}[1.]
        \item The utility-change multisets for the abandoned coalition $\{\beta, \gamma\}$ of $\beta$'s deviation in $\bar G$, namely, $X = \mathit{UC}^\mathrm{out}_{\bar G}(\beta, \bar \pi, \bar \pi^\beta_\rightarrow)$, and $\gamma$'s deviation in $\bar G'$, namely $X' = \mathit{UC}^\mathrm{out}_{\bar G'}(\gamma, \bar \pi, \bar \pi^\gamma_\rightarrow)$, are identical. Specifically, they are both $\{(\nicefrac{T_A}{\delta}, 0)\}$ with $\delta \in \{1, 2\}$ (dependent on the game type). 
        Hence $X$ dominates $X'$.
        
        \item Consider the utility-change multisets with respect to the two welcoming coalitions, namely, $Y = \mathit{UC}^\mathrm{in}_{\bar G}(\beta, \bar \pi, \bar \pi^\beta_\rightarrow)$, and $Y' = \mathit{UC}^\mathrm{in}_{\bar G'}(\gamma, \bar \pi, \bar \pi^\gamma_\rightarrow)$. Now, since $D \cup \Lambda \subseteq \Tilde{\pi}_{\overset{y}{\rightarrow}}(y)$ with $|D \cup \Lambda| \geq |A| - 2h + 2h \geq |A|$, and $C_A \subseteq A$, it must hold that $|Y| \geq |Y'|$. Further, for every agent $p \in C$, it holds that $\uf_p(C) = 0$ in all game classes, while $\uf_p(C \cup \{\beta\}) \geq \nicefrac{T_A}{n} \geq t_A$. On the other hand, any agent in $C_A$ can have a valuation of at most $t_A$ for agent $\gamma$, and hence only experience an increase of $t_A$ by $\gamma$ joining $C_a$. Thus, it must hold that $Y$ dominates $Y'$. 
        
        \item Let $z = \mathit{uc}_{\bar G}(\beta, \bar \pi, \bar \pi^\beta_\rightarrow) = (0, \uf_\beta(C \cup \{\beta\})$ with $\uf_\beta(C \cup \{\beta\}) \geq \nicefrac{nT_A}{n} = T_A$, and let $z' = \mathit{uc}_{\bar G'}(\gamma, \bar \pi, \bar \pi^\gamma_\rightarrow) = (0, \uf_\gamma(C_A \cup \{\gamma\})$ with $\uf_\gamma(C_A \cup \{\gamma\}) \leq t_A$. Then $\{z\}$ dominates $\{z'\}$.
    \end{enumerate}

    But then, by the definition of a standard stability notion, $f_\chi(X, Y, z) \geq f_\chi(X', Y', z')$ must hold for the above-defined sets. Now, since $\beta$ cannot deviate from $\{\beta, \gamma\}$ to join $C$ in $\bar G$, i.e., $f_\chi(X, Y, Z) = 0$, agent $\gamma$ cannot deviate from $\{\beta, \gamma\}$ in $\bar G'$, from which we can immediately follow the statement.
\end{proof}

In the following lemma, we examine the $\chi$ dynamics of $(G, \startpart)$ in greater detail, where we fully characterize all coalitions that can result from the dynamics.

\crefname{enumi}{Type}{Types}
\Crefname{enumi}{Type}{Type}
\begin{lemma}\label{lem:unified_reduction_coalitions_from_dynamics}
    Let $\pi'$ be a partition of $N$ that resulted from $(G, \startpart)$ through an execution of the $\chi$ dynamics. 
    Then each coalition in $\pi'$ is of one of the following types:
    \begin{enumerate}[I.]
        \item $\{r_e\}$,\label{enum:unified_reduction_1}
        \item $\{x_e, r_e\}$,\label{enum:unified_reduction_2}
        \item $\{g\} \cup N_\mathcal{U}'$ for some $g \in \Gamma$, and $N_\mathcal{U}' \subseteq N_\mathcal{U}$,\footnote{Here, $N_\mathcal{U}' = \emptyset$ is explicitly allowed.}\label{enum:unified_reduction_3}
        \item $\{g, x_e, x_f, x_g, s_M\}$ for some $g \in \Gamma$, $s_M \in N_\mathcal{M}$ with $M = \{e, f, g\}$, \label{enum:unified_reduction_4}
        \item $\{\alpha\} \cup \Lambda \cup D$, where $\Lambda \subseteq N_\mathcal{M}$ with $|\Lambda| \geq 2h$, \label{enum:unified_reduction_5}
        \item $\{\alpha, \beta\} \cup \Lambda \cup D$, where $\Lambda \subseteq N_\mathcal{M}$ with $|\Lambda| = 2h$,\label{enum:unified_reduction_6}
        \item $\{\beta, \gamma\}$, and \label{enum:unified_reduction_7}
        \item $C_A \subseteq A$. \label{typeA} \label{enum:unified_reduction_8}
    \end{enumerate}
\end{lemma}
\begin{proof}
    We show the statement by induction over the number of deviations. 
    First, the statement is true for the initial partition $\startpart$:
    Agents in $\Gamma$ are in coalitions of \cref{enum:unified_reduction_3}
    and the other agents are in coalitions of \Crefor{enum:unified_reduction_2,,enum:unified_reduction_5,,enum:unified_reduction_7,,enum:unified_reduction_8}.
    
    Next, let $\pi'$ be a partition of $N$ that resulted from $(G, \startpart)$ through an execution of the $\chi$ dynamics, and assume that each coalition in $\pi'$ is of one of the above types.
    Therefore, let $\pi''$ be a partition and let $p \in N$ be a deviator such that $\pi' \overset{p}{\rightarrow}_\chi \pi''$, i.e., we increase the execution of the dynamics by another deviation.
    We will show that both the abandoned coalition $\pi'(p) \setminus \{p\}$ and the welcoming coalition $\pi''(p)$ are of one of the above types.

    Note that $p \not \in \Gamma \cup R$, as these agents are indifferent over all coalitions they can possibly be in, and, thus, none of their deviations would be an \NS and, therefore, $\chi$ deviation. 
    Also note that there can only ever be one coalition of \cref{enum:unified_reduction_5,,enum:unified_reduction_6} and one of \cref{enum:unified_reduction_6,,enum:unified_reduction_7}.

    Now, assume that $p = x_e$ for some $e \in \mathcal{U}$. 
    Then $p$ cannot be part of a coalition of \cref{enum:unified_reduction_4} in $\pi'$, as the presence of the relevant set agent ensures that $p$ cannot increase her utility by deviating. 
    Hence, $p$ is part of a coalition of \Crefor{enum:unified_reduction_2,,enum:unified_reduction_3} in $\pi'$, and has a utility of at least $0$ in all game classes. 
    Thus, the deviation must lead to a strictly positive utility for~$p$.
    In particular, this implies that the welcoming coalition cannot be of \Crefor{enum:unified_reduction_1,,enum:unified_reduction_2,,enum:unified_reduction_7,,enum:unified_reduction_8}, as $p$ has no positive valuation for any agent in these coalitions and, therefore, no overall positive utility. 
    It can also not be of \Crefor{enum:unified_reduction_5,,enum:unified_reduction_6}, as the presence of $\alpha$ ensures that $p$ has strictly negative utility for such a coalition. 
    Next, it cannot be of \cref{enum:unified_reduction_4}, as such a coalition contains $s_M$ for some $M \in \mathcal{M}$ together with all corresponding element agents.
    Hence, for the $e\in \mathcal{U}$ with $p = x_e$, it holds that $e \not\in M$, and $p$ would obtain a negative utility from joining the coalition of $s_M$.
    Finally, if the welcoming coalition is of \cref{enum:unified_reduction_3}, then it is so after the deviation. 
    This concludes the consideration of the case $p = x_e$.

    In case $p = s_M$ for some $M \in \mathcal{M}$, then we claim that $\pi''(p)$ must be of \cref{enum:unified_reduction_4}.
    First, assume that $\pi'(p)$ is of \cref{enum:unified_reduction_4}. But then, $p$ has a utility of at least $\nicefrac{nT_A}{n} \geq T_A$ in partition $\pi'$, and it is clear that $p$ cannot gain by deviating, as the only agent outside of $\pi'(p)$ that $p$ has a positive valuation for is $\beta$, with $\vf_{p}(\beta) = T_A$. 
    Hence, $\pi'(p)$ must be of \Crefor{enum:unified_reduction_5,,enum:unified_reduction_6}, and, since there is only one coalition of \Crefor{enum:unified_reduction_5,,enum:unified_reduction_6}, she cannot join a coalition of \Crefor{enum:unified_reduction_5,,enum:unified_reduction_6}. 
    Now, $p$ cannot deviate to a coalition of \Crefor{enum:unified_reduction_1,,enum:unified_reduction_2}, as 
    the presence
    of a restricting agent would lead to a negative utility.
    Also, she cannot deviate to a coalition of \cref{enum:unified_reduction_4}, as one of the present element agents would not correspond to the set $M$, and $p$ can have a utility of at most $0$ for any such coalition.
    Moreover, she cannot deviate to a coalition of \Crefor{enum:unified_reduction_7,,enum:unified_reduction_8} due to \Cref{obs:unified_reduction_no_connecting_coalition}.
    As we have excluded all other cases, the welcoming partition $\pi''(p) \setminus \{p\}$ must be of \cref{enum:unified_reduction_3}.
    Further, it has to exactly contain those element agents corresponding to the set $M$ with $p = s_M$.
    Otherwise, the negative valuation of $s_M$ for agents in $\Gamma$ would lead to $p$ not performing an \NS deviation.
    However, if $\pi'(p)$ contains at most $2h$ set agents, then each of the remaining set agents not contained in  $\pi'(p)$ form coalitions of \cref{enum:unified_reduction_4} in $\pi'$. 
    Hence, $\pi'$ would not contain a coalition of \cref{enum:unified_reduction_3}.
    Hence, $\pi'(p)$ must contain at least $2h+1$ set agents and, therefore, be of \cref{enum:unified_reduction_5}. 
    But then, the abandoned coalition is still of \cref{enum:unified_reduction_5} after the deviation, and the welcoming partition $\pi''(p)$ is of \cref{enum:unified_reduction_4}.
    This concludes the consideration of the case $p = s_M$.

    Next, note that $\alpha$ can only ever have an incentive to deviate to join $\beta$. 
    However, if $\beta$ is in a coalition of \cref{enum:unified_reduction_7}, then $\alpha$ cannot deviate because of \Cref{obs:unified_reduction_no_connecting_coalition}, and otherwise, $\beta$ is already in the same coalition as $\alpha$. 
    It follows that $\alpha$ never deviates and, therefore, $p \neq \alpha$.

    In case that $p = \beta$, the abandoned coalition cannot be of \cref{enum:unified_reduction_6}, as then, $p$ would already have a utility of $\nicefrac{nT_A}{n} \geq T_A$ in $\pi'$, while $\gamma$ is the only outside agent that $p$ \emph{might} have a positive valuation for, with $\vf_p(\gamma) \leq T_A$. 
    Thus, the abandoned coalition must be of \cref{enum:unified_reduction_7} and, after the deviation, $\pi''(p) \setminus \{p\} = \{\beta, \gamma\} \setminus \{\beta\} = \{\gamma\} \subseteq A$, which is of \cref{enum:unified_reduction_8}.
    Moreover, deviating from $\{\beta,\gamma\}$, $p = \beta$ can only have incentive to join $\alpha$, i.e., the welcoming coalition must be of \cref{enum:unified_reduction_5} containing $\Lambda \subseteq N_\mathcal{M}$. 
    However, if $|\Lambda| > 2h$ then $\sum_{d \in \pi''(p)} \vf_p(d) \leq (2h+1) \cdot nT_A - (2h+1) \cdot nT_A \leq 0$, and $p$ does not have an incentive to deviate. 
    In addition, all set agents outside of this coalition must be in a coalition of \cref{enum:unified_reduction_4}, and there are at most $h$ such coalitions.
    Thus, we have $|\Lambda| \ge 2h$ and, therefore,
    $|\Lambda| = 2h$ must hold.
    We conclude that the welcoming coalition is of \cref{enum:unified_reduction_6} after the deviation.

    Finally, let us consider the case where $p \in A$. 
    If $p = \gamma$, 
    then, by \Cref{lem:unified_reduction_gamma_doesnt_deviate}, $\pi'(p)\neq \{\beta,\gamma\}$. 
    Hence, whether $p = \gamma$ or not, $p$ must abandon a coalition of \cref{enum:unified_reduction_8}, which remains such a coalition (if nonempty after the deviation).
    Also, $\beta$ then is in a coalition of \cref{enum:unified_reduction_6}.
    By \Cref{obs:unified_reduction_no_connecting_coalition}, $p$ cannot join a coalition containing an agent outside of $A$.
    Therefore, the welcoming coalition must be of \cref{enum:unified_reduction_8} as well.
\end{proof}

Next, we show the defining behavior of the constructed instance, namely that whether agent $\beta$ can deviate from $\{\beta, \gamma\}$ is directly corresponding to whether the source instance of \RXC is a Yes-instance.

\begin{lemma}\label{lem:unified_reduction_correspondance_to_RX3C}
    There exists an execution of the $\chi$ dynamics of $(G, \startpart)$ where agent $\beta$ can deviate from $\{\beta, \gamma\}$ if and only if the source instance $\mathcal{I} = (\mathcal{U}, \mathcal{M})$ of \RXC is a Yes-instance.
    Moreover, if $\mathcal I$ is a No-instance, then $\{\beta, \gamma\}$ is part of every occurring partition in every execution of the $\chi$ dynamics of $(G, \startpart)$.
\end{lemma}
\begin{proof}
    Note that, since $\{\beta, \gamma\} \in \startpart$ and because of \Cref{obs:unified_reduction_no_connecting_coalition}, no other agent can ever deviate to join $\{\beta, \gamma\}$.
    Additionally, because of \Cref{lem:unified_reduction_gamma_doesnt_deviate}, agent $\gamma$ can never abandon $\{\beta, \gamma\}$. 
    Thus, it suffices to show that $\beta$ can deviate from $\{\beta, \gamma\}$ in some execution of the $\chi$ dynamics if and only if the source instance $\mathcal{I} = (\mathcal{U}, \mathcal{M})$ of \RXC is a Yes-instance.

    ($\Rightarrow$) Assume that $\beta$ can deviate from $\{\beta, \gamma\}$. 
    By \Cref{lem:unified_reduction_coalitions_from_dynamics}, this deviation must result in a coalition of \cref{enum:unified_reduction_6}, i.e., a coalition $C = \{\alpha, \beta\} \cup \Lambda \cup D$ with $\Lambda \subseteq N_\mathcal{M}$ and $|\Lambda| = 2h$. But then, there must be exactly $h$ set agents who left their initial coalition with $\alpha$ to join a coalition of \cref{enum:unified_reduction_4}. Since such a coalition consists of some grouping agent and all element agents that correspond to the relevant sets, it is easy to see that these set agents correspond to an exact cover of $\mathcal{U}$ with $h$ sets from $\mathcal{M}$.

    ($\Leftarrow$) Assume that $\mathcal{I}$ is a Yes-instance, i.e., there is a set $\mathcal{M}' = \{M_1, \ldots, M_h\}$ such that $\bigcup_{M \in \mathcal{M}'} M = \mathcal{U}$. Note that, since $\CIS \devimplies \chi$, every $\CIS$ deviation also is a $\chi$ deviation. 
    We thus provide a sequence of \CIS deviations, followed by a single $\chi$ deviation of $\beta$ to join the coalition of $\alpha$. 
    For all \CIS deviations, we will argue that (i) it is an \NS deviation, (ii) the favor-in set of the abandoned coalition is empty, 
    and (iii) the favor-out set of the joined coalition is empty. 
    The sequence of deviations can now be given as follows (unless otherwise specified, within each step, deviations are performed in an arbitrary order): 
    \begin{enumerate}[1.]
        \item For each $i\in [h]$ 
        and $e \in M_i$, agent $x_e$ deviates to join the coalition of grouping agent $g_i$. 
        Since $M'$ is an exact cover of $\mathcal{U}$ of size~$h$, no agent in $N_\mathcal{U}$ is asked to join two different agents in $\Gamma$. 
        These are \CIS deviations since:
        \begin{enumerate}[(i)]
            \item agent $x_e$ increases her utility from $0$ to at least $\nicefrac{1}{2}$,
            \item agent $x_e$ leaves the coalition $\{x_e, r_e\}$, and $\vf_{r_e}(x_e) = 0$, and
            \item the coalition of $g_i$ is a subset of $\{g_i\} \cup \{x_{e'}\}_{e' \in M_i}$, where $g_i$ is indifferent of $x_e$'s deviation, and the agents $\{x_{e'}\}_{e' \in M_i}$, in case of an \ASHG or \FHG, are strictly in favor of the deviation, or indifferent in case of an \MFHG.
        \end{enumerate}
        
        \item For each $i\in [h]$, 
        agent $s_{M_i}$ joins the coalition of $g_i$. 
        These are \CIS deviations since:
        \begin{enumerate}[(i)]
            \item agent $s_{M_i}$ increases her utility from $0$ to at least $\nicefrac{(-2nT_A + 3nT_A)}{5} \geq T_A$, 
            \item agent $s_{M_i}$ 
            leaves a coalition that is a subset of $\{\alpha\} \cup N_\mathcal{M} \cup D$, where all these agents have a valuation of $0$ for $s_M$ and utility of $0$ for the whole coalition and are hence indifferent to the deviation in all game classes, and
            \item the coalition of $g_i$ is $\{g_i, x_e, x_f, x_g\}$ with $M_i = \{e, f, g\}$, where $g_i$ is indifferent to the deviation, and the set agents strictly increase their utility from at most $3$ to at least $\nicefrac{nT_A}{5} \geq T_A$.
        \end{enumerate}

        \item Agent $\beta$ performs a $\chi$ deviation to join the coalition of $\alpha$, which, at this point in the dynamics, is of the form $\{\alpha\} \cup \Lambda \cup D$, where $\Lambda \subseteq N_\mathcal{M}$ with $|\Lambda| = 2h$. This is a $\chi$ deviation, as shown in \Cref{lem:unified_reduction_beta_can_deviate}. \qedhere
    \end{enumerate}
\end{proof}

Finally, we show that the $\chi$ dynamics of the subgame restricted to the agents in $N \setminus A$ must converge when starting from the partition $\startpart$.

\begin{lemma}\label{lem:unified_reduction_finite_deviations}
    In every execution of the $\chi$ dynamics of $(G, \startpart)$, every agent in $N \setminus A$ can only perform a finite number of deviations.
\end{lemma}
\begin{proof}
    With the arguments given in \Cref{lem:unified_reduction_coalitions_from_dynamics}, one can verify that no agent in $\Gamma \cup R \cup D \cup \{\alpha\}$ ever has an incentive to deviate.
    Similarly, agent $\beta$ has no incentive to abandon a coalition of \cref{enum:unified_reduction_6} and can, therefore, deviate at most once.
    Moreover, a set agent in $N_\mathcal{M}$ does not have an incentive to deviate from a coalition of \cref{enum:unified_reduction_4}, and hence deviates at most once.

    It thus suffices to consider deviations of element agents.
    Given a partition $\bar \pi$ of $N$, consider the following potential function:
    $$
        \Phi(\bar \pi) := \sum_{g \in \Gamma} |\bar \pi(g) |^2.
    $$

    Intuitively, $\Phi$ describes the sum of squared coalition sizes of all grouping agents. Due to \Cref{lem:unified_reduction_coalitions_from_dynamics}, it is easy to see that for any $\bar \pi$ that resulted from $\pi$ through the $\chi$ dynamics, the value $\Phi(\bar \pi)$ is at most $|\Gamma| \cdot (5-1)^2 = 16h$.
    
    We claim that each deviation of an element agent strictly increases the value of potential function $\Phi$. 
    As every other agent in $N \setminus A$ can perform at most one deviation, this immediately implies the statement. 

    Let $\pi'$ be a partition that resulted from $\pi$ through an execution of the $\chi$ dynamics, and let $x \in N_\mathcal{U}$ be an element agent that performs a deviation $\pi' \overset{x}{\rightarrow_\chi} \pi''$ for some $\pi''$.
    
    If $x$ deviates from a coalition of \cref{enum:unified_reduction_2}, an increase of $\Phi$ is immediate. 
    Next, a deviation of $x$ from a coalition of \cref{enum:unified_reduction_4} would result in a forbidden coalition type, so it cannot happen. 
    It remains to consider deviations of $x$ from coalitions of \cref{enum:unified_reduction_3}.
    
    First, $x$ has no incentive to deviate from a coalition of \cref{enum:unified_reduction_3} to a coalition of \cref{enum:unified_reduction_1}, as she receives strictly positive utility for any of the former and a utility of~$0$ for the latter. 
    Further, as a coalition of \cref{enum:unified_reduction_4} cannot be formed through a deviation of $x$, this leaves a deviation of $x$ from a coalition $C$ of \cref{enum:unified_reduction_3} to a different coalition of \cref{enum:unified_reduction_3}. 
    Note that $x$ only increases her utility via such a deviation in case $|\pi'(x)| < |\pi''(x)|$. 
    Let $i := |\pi'(x)|$, and let $j := |\pi''(x)|$. 
    Then it holds that
    \begin{align*}
        & \Phi(\pi'') - \Phi(\pi') \\
        = & (|\pi'(x) \setminus \{x\}|^2 + |\pi''(x)|^2)- (|\pi'(x)|^2 + |\pi''(x) \setminus \{x\}|^2) \\
        = & ((i-1)^2 + j^2) - (i^2 + (j-1)^2) \\
        = & 2 \cdot (j-i)
        > 0\text.
    \end{align*}
    This completes the proof.
\end{proof}

\subsection{Hardness Results}

We are now ready to leverage our reduction to prove our hardness results stated in \Cref{thm:unified_NCD,thm:unified_PCD}.
First, we prove the theorem for possible convergence by utilizing the reduced games from \Cref{sec:hardnessproof:construction} as gadgets in a larger construction.

\unifiedPCD*

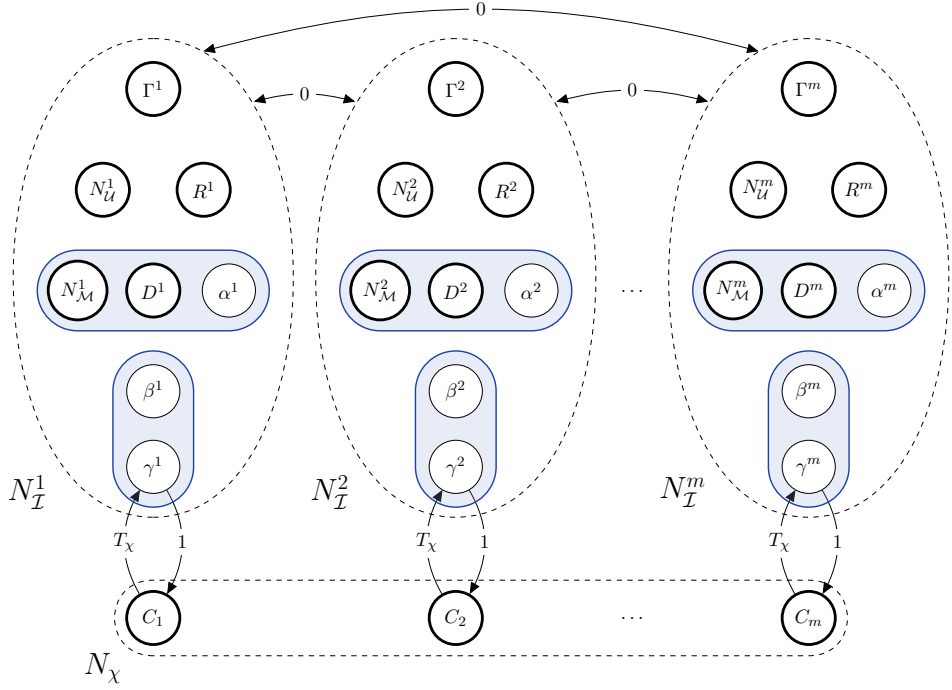
\begin{figure*}[ht]
    \centering
    \resizebox{.85\textwidth}{!}{
        \begin{tikzpicture}
	\begin{pgfonlayer}{nodelayer}
		\node [style={agent_group}] (0) at (-15, 2) {$N_\mathcal{M}^1$};
		\node [style={agent_group}] (1) at (-12, 2) {$D^1$};
		\node [style=agent] (2) at (-9, 2) {$\alpha^1$};
		\node [style=agent] (3) at (-12, -2) {$\beta^1$};
		\node [style=agent] (4) at (-12, -5) {$\gamma^1$};
		\node [style={agent_group}] (5) at (-14, 6) {$N_\mathcal{U}^1$};
		\node [style={agent_group}] (6) at (-10, 6) {$R^1$};
		\node [style={agent_group}] (7) at (-12, 10) {$\Gamma^1$};
		\node [style={agent_group}] (8) at (-3, 2) {$N_\mathcal{M}^2$};
		\node [style={agent_group}] (9) at (0, 2) {$D^2$};
		\node [style=agent] (10) at (3, 2) {$\alpha^2$};
		\node [style=agent] (11) at (0, -2) {$\beta^2$};
		\node [style=agent] (12) at (0, -5) {$\gamma^2$};
		\node [style={agent_group}] (13) at (-2, 6) {$N_\mathcal{U}^2$};
		\node [style={agent_group}] (14) at (2, 6) {$R^2$};
		\node [style={agent_group}] (15) at (0, 10) {$\Gamma^2$};
		\node [style={agent_group}] (16) at (11, 2) {$N_\mathcal{M}^m$};
		\node [style={agent_group}] (17) at (14, 2) {$D^m$};
		\node [style=agent] (18) at (17, 2) {$\alpha^m$};
		\node [style=agent] (19) at (14, -2) {$\beta^m$};
		\node [style=agent] (20) at (14, -5) {$\gamma^m$};
		\node [style={agent_group}] (21) at (12, 6) {$N_\mathcal{U}^m$};
		\node [style={agent_group}] (22) at (16, 6) {$R^m$};
		\node [style={agent_group}] (23) at (14, 10) {$\Gamma^m$};
		\node [style={agent_group}] (24) at (-12, -11) {$C_1$};
		\node [style={agent_group}] (25) at (0, -11) {$C_2$};
		\node [style={agent_group}] (26) at (14, -11) {$C_m$};
		\node [style=none] (27) at (7, 2) {$\ldots$};
		\node [style=none] (28) at (7, -11) {$\ldots$};
		\node [style=none] (35) at (14, 12) {};
		\node [style=none] (36) at (14, -7) {};
		\node [style=none] (37) at (0, 12) {};
		\node [style=none] (38) at (0, -7) {};
		\node [style=none] (39) at (-12, 12) {};
		\node [style=none] (40) at (-12, -7) {};
		\node [style=none] (41) at (-8, 9.5) {};
		\node [style=none] (42) at (-4, 9.5) {};
		\node [style=none] (43) at (4, 9.5) {};
		\node [style=none] (44) at (10, 9.5) {};
		\node [style=none] (45) at (-10, 11.5) {};
		\node [style=none] (46) at (12, 11.5) {};
		\node [style=none] (47) at (-12, -9.5) {};
		\node [style=none] (48) at (-12, -12.5) {};
		\node [style=none] (49) at (14, -9.5) {};
		\node [style=none] (50) at (14, -12.5) {};
		\node [style=none] (51) at (-17, -6) {\LARGE$N_{\mathcal I}^1$};
		\node [style=none] (52) at (-5, -6) {\LARGE$N_{\mathcal I}^2$};
		\node [style=none] (53) at (9, -6) {\LARGE$N_{\mathcal I}^m$};
		\node [style=none] (54) at (-14, -13) {\LARGE$N_\chi$};
	\end{pgfonlayer}
	\begin{pgfonlayer}{edgelayer}
		\draw [style=ddashed] (35.center)
			 to [bend right=90] (36.center)
			 to [bend right=90] cycle;
		\draw [style=ddashed] (37.center)
			 to [bend right=90] (38.center)
			 to [bend right=90] cycle;
		\draw [style=ddashed, bend right=90] (39.center) to (40.center);
		\draw [style=ddashed, bend right=90] (40.center) to (39.center);
		\draw [style=utility-symmetric, bend left=15] (43.center) to node [midway, fill=white] {$0$} (44.center);
		\draw [style=utility-symmetric, bend left=15] (41.center) to node [midway, fill=white] {$0$} (42.center);
		\draw [style=utility-symmetric, bend left=15] (45.center) to node [midway, fill=white] {$0$} (46.center);
		\draw [style=ddashed] (48.center)
			 to [bend left=90, looseness=1.75] (47.center)
			 to (49.center)
			 to [bend left=90, looseness=1.75] (50.center)
			 to cycle;
		\draw [style=utility, bend left] (4) to node [midway, fill=white] {$1$} (24);
		\draw [style=utility, bend left] (24) to node [midway, fill=white] {$T_\chi$} (4);
		\draw [style=utility, bend left] (12) to node [midway, fill=white] {$1$} (25);
		\draw [style=utility, bend left] (25) to node [midway, fill=white] {$T_\chi$} (12);
		\draw [style=utility, bend left] (20) to node [midway, fill=white] {$1$} (26);
		\draw [style=utility, bend left] (26) to node [midway, fill=white] {$T_\chi$} (20);

        \draw[thick,CoalitionColor, fill=CoalitionColor!50, fill opacity=0.2]  \convexpath{3, 4}{1.6cm};
        \draw[thick,CoalitionColor, fill=CoalitionColor!50, fill opacity=0.2]  \convexpath{11, 12}{1.6cm};
        \draw[thick,CoalitionColor, fill=CoalitionColor!50, fill opacity=0.2]  \convexpath{19, 20}{1.6cm};
        \draw[thick,CoalitionColor, fill=CoalitionColor!50, fill opacity=0.2]  \convexpath{0, 2}{1.6cm};
        \draw[thick,CoalitionColor, fill=CoalitionColor!50, fill opacity=0.2]  \convexpath{8, 10}{1.6cm};
        \draw[thick,CoalitionColor, fill=CoalitionColor!50, fill opacity=0.2]  \convexpath{16, 18}{1.6cm};
	\end{pgfonlayer}
\end{tikzpicture}
        }
    \caption{Illustration of the reduction for the proof of \Cref{thm:unified_PCD}.}
    \label{fig:PCD_reduction}
\end{figure*}

\begin{proof}
    Consider an \ASHG, \FHG, or \MFHG $G_\chi = (N_\chi, \vf_\chi)$ and a starting partition $\pi_{\chi}$ such that  every execution of the $\chi$ dynamics of $(G_{\chi}, \pi_{\chi})$ cycles.
    Let $m := |\pi_\chi|$ be the number of coalitions in $\pi_\chi$ and assume that $\pi_{\chi} = \{C_1,\dots, C_m\}$.
    For our hardness reduction, we will utilize several copies of the game constructed in \Cref{sec:hardnessproof:construction}.
    An illustration is provided in \Cref{fig:PCD_reduction}.
    Specifically, we consider one copy for each coalition in $\pi_{\chi}$.
    All coalitions introduce a disjoint set of agents except for the gadget agents that share $N_\chi$ and only have an individual dedicated agent.
    
    Formally, consider an \RXC instance $\mathcal{I} = (\mathcal{U}, \mathcal{M})$.
    We construct a reduced game $G = (N,\vf)$ where $N = N_\chi\cup \bigcup_{i\in [m]}N_{\mathcal I}^i$.
    For $i\in [m]$, $N_{\mathcal I}^i = N_{\mathcal M}^i \cup N_{\mathcal U}^i \cup R^i \cup \Gamma^i \cup D^i \cup \{\gamma^i\}$ where all sets are chosen according to the reduced instance corresponding to $\mathcal I$ constructed in \Cref{sec:hardnessproof:construction}.

    Valuations among agents in $N_\chi$ are according to $\vf_\chi$.
    Next, define $T_\chi := 1 + |N_\chi|\sum_{p, d \in N_\chi} |\vf_p(d)|$.
    For $i\in [m]$ and $p\in N_\chi$, we have 
    $$\vf_{\gamma^i}(p) = \begin{cases}
        1 & p \in C_i\\
        0 & p\in N_\chi \setminus C_i
    \end{cases}$$
    and
    $$
    \vf_{p}(\gamma^i) = \begin{cases}
        T_\chi & p \in C_i\text,\\
        0 & p\in N_\chi \setminus C_i\text.
    \end{cases}\text.$$
    Moreover, among the agent set $N_{\mathcal I}^i\cup N_\chi$, we define valuations according to the construction in \Cref{sec:hardnessproof:construction}, where we identify the variable subset of agents with $A^i = \{\gamma^i\}\cup N_\chi$.
    Note that these sets are not disjoint: we have $\bigcup_{i\in [m]}A^i = N_\chi$.
    However, the valuations are still well defined because we assume the given valuations among agents $A^i$ as just defined.
    Finally, for $1\le i < j\le m$, $p\in N_{\mathcal I}^i$, and $d \in N_{\mathcal I}^j$, we define $\vf_p(d) = 0$.

    We set the starting partition $\startpart$ to the union of the starting partition defined in \Cref{sec:hardnessproof:construction} for the agents in $N^i$ and $\pi_\chi$ for the agents in $N_\chi$, i.e.,
    
    \begin{align*}
        \startpart &= \pi_\chi \cup \bigcup_{i\in [m]} \{ \{g^i\} \mid g^i \in \Gamma^i\} \cup \{\{x^i_e, r^i_e\} \mid e \in \mathcal{U}\} \\
        &\cup \{\{D^i \cup N_{\mathcal{M}}^i \cup \{\alpha^i\}\}, \{\beta^i, \gamma^i\}\}\text.        
    \end{align*}

    For $i\in [m]$, the subgame of $G$ induced by the agent set $N_{\mathcal I}^i\cup N_\chi$ is identical to the reduced instance constructed in \Cref{sec:hardnessproof:construction} and the starting partition restricted to this agent set is identical to the respective starting partition.
    Note that the partition $\startpart$ is individually rational for all agents in $N_{\mathcal I}^i$ and, by \Cref{lem:unified_reduction_coalitions_from_dynamics}, remains individually rational for these agents while deviations only happen among coalitions among $N_{\mathcal I}^i\cup N_\chi$.
    However, while coalitions for these agents are individually rational, they have no incentive to join the coalition of an agent in $N_{\mathcal I}^j$ for $j\in [m]\setminus \{i\}$.
    Hence, throughout any execution of the $\chi$ dynamics of $(G,\startpart)$, agents in $N_{\mathcal I}^i$ and $N_{\mathcal I}^j$ will not form joint coalitions.
    Hence, we can apply all of our results from \Cref{sec:hardnessproof:behavior}, especially \Cref{lem:unified_reduction_coalitions_from_dynamics,lem:unified_reduction_correspondance_to_RX3C,lem:unified_reduction_finite_deviations}, to the subgame induced by $N_{\mathcal I}^i\cup N_\chi$.

    We now claim that the $\chi$ dynamics of $(G, \startpart)$ can converge if and only if the source instance $\mathcal{I}$ of \RXC is a Yes-instance.

    ($\Rightarrow$) Assume that $\mathcal{I}$ is a No-instance.
    Let $i\in [m]$. 
    By \Cref{lem:unified_reduction_correspondance_to_RX3C}, in any execution of the $\chi$ dynamics of $(G, \startpart)$, $\{\gamma^i,\beta^i\}$ will remain a coalition throughout.
    But then, the agents in $N_\chi$ must remain in coalitions among themselves. 
    As a consequence, the $\chi$ dynamics in the subgame $(G_\chi, \pi_\chi)$ must cycle, and thus, must also cycle in $(G, \startpart)$.

    ($\Leftarrow$) Assume that $\mathcal{I}$ is a Yes-instance. Then, as shown in \Cref{lem:unified_reduction_correspondance_to_RX3C}, there exists a sequence of $\chi$ deviations that results in all agents in $\{\gamma^i\mid i\in [m]\}$ being in their respective singleton coalitions.

    Next, for $i\in [m]$, we let each agent $\gamma_i$ deviate to her corresponding coalition $C_i \in \pi_\chi$. 
    These are \CIS deviations (and thus $\chi$ deviations), since $\gamma_i$ leaves her singleton coalition, has positive valuations for all agents in $C_i$ (and hence strictly increases her utility), and since the utility of all agents in $C_i$ strictly increases in all game classes based on the definition of $T_\chi$. 
    It is easy to see that, after these deviations, no agent in $N_\chi$, nor an agent in $\{\gamma^i\mid i\in [m]\}$ can increase their utility by performing any further deviation.
    Further, by \Cref{lem:unified_reduction_finite_deviations}, each agent in $N_{\mathcal{I}}^i$ can only perform a finite number of deviations. 
    Hence, the $\chi$ dynamics of $(G, \startpart)$ must converge subsequently.
\end{proof}

Next, we prove our theorem for necessary convergence.

\unifiedNCD*

\begin{proof}
    Assume that there exists a game $G_\chi = (N_\chi,\vf_\chi)$ and partition $\pi_\chi$ with the properties of the statement of the theorem.
    We show \NP-hardness of the complement problem, i.e., whether it is possible for the $\chi$ dynamics to cycle. 

    Given an \RXC instance $\mathcal{I} = (\mathcal{U}, \mathcal{M})$, consider the constructed instance $(G, \startpart)$ from \Cref{sec:hardnessproof:construction}. 
    We now specify the subgame by setting $A = N_{\chi}$ and all valuations among the agents in $N_{\chi}$ are defined as in $\vf_{\chi}$. 
    Further, let the dedicated agent $\gamma$ in $A$ be $a$, where $\{a\} \in \pi_\chi$, and we set the subpartition $\pi_A$ that was referenced in \Cref{eq:startpart} in \Cref{sec:hardnessproof:construction} to $\pi_\chi \setminus \{a\}$. 

    We claim that the $\chi$ dynamics of the constructed instance $(G, \pi)$ can cycle if and only if the source instance $\mathcal{I}$ of \RXC is a Yes-instance.

    ($\Rightarrow$) Assume that $\mathcal{I}$ is a No-instance. Then, by \Cref{obs:unified_reduction_no_connecting_coalition,,lem:unified_reduction_correspondance_to_RX3C}, agent $a = \gamma$ must remain in the coalition $\{\beta, a\}$. Now, again because of \Cref{obs:unified_reduction_no_connecting_coalition}, all agents in $N_\chi \setminus \{a\}$ must remain in coalitions among themselves, and, by definition, the $\chi$ dynamics of the sub-game $G_\chi$ when starting on $\pi$ must converge. Together with \Cref{lem:unified_reduction_finite_deviations}, this directly implies that the $\chi$ dynamics of $(G, \pi)$ must converge.
    
    ($\Leftarrow$) Assume that $\mathcal{I}$ is a Yes-instance. Then, because of \Cref{lem:unified_reduction_correspondance_to_RX3C}, in the $\chi$ dynamics of $G$, agent $\beta$ can deviate to join $\alpha$. But then, agent $a$ is left in her singleton coalition, and the agents in $N_\chi$ are partitioned as in $\pi_\chi$. Now, by definition of $G_\chi$, the $\chi$ dynamics can cycle.
\end{proof}

\section{Proof of Theorem~\ref{thm:consequences}}\label{app:voting}

In this section, we present the full proof of \Cref{thm:consequences} restated as follows.

\consequences* 

We want to apply \Cref{thm:unified_NCD,thm:unified_PCD}, and, therefore, have to construct games with the desired properties.

We further distinguish whether for the relevant stability notion $\chi$ it holds that $(\qout, 1)$-\SVS $\devimplies \chi$ or $(1, \qin)$-\SVS $\devimplies \chi$.
We start with the former
and construct two games for the respective applications of \Cref{thm:unified_NCD,thm:unified_PCD}.

\begin{lemma}\label{lem:SVS_quota_out_no_convergence}
    Let $\chi$ be a stability notion such that $\chi \devimplies \NS$ and $(\qout, 1)$-\SVS  $\devimplies \chi$ for some $\qout \in [0, 1[$. 
    Then, there exists an \ASHG, \FHG, and \MFHG $G = (N, \vf)$, and partition $\pi$ of $N$ such that the $\chi$ dynamics on $G$ must cycle when starting on $\pi$.
\end{lemma}

\begin{proof}
    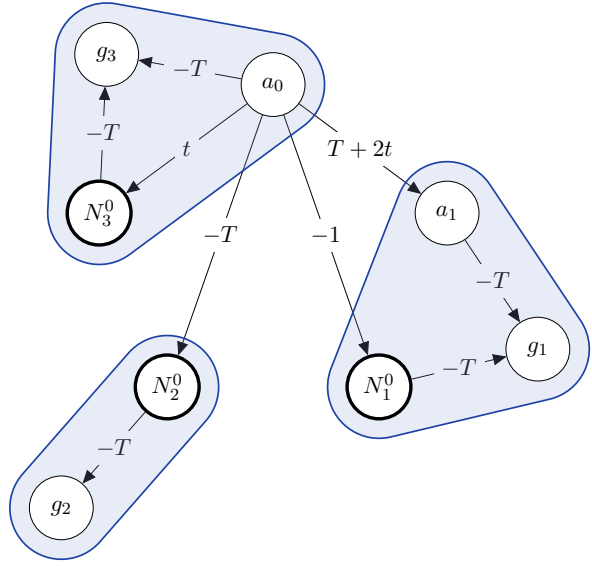
\begin{figure}[t]
        \centering
        \resizebox{.5\textwidth}{!}{
            \begin{tikzpicture}
	\begin{pgfonlayer}{nodelayer}
		\node [style=agent] (0) at (0, 6) {$a_0$};
		\node [style=agent] (1) at (5.75, 1.75) {$a_1$};
		\node [style={agent_group}] (2) at (3.5, -4) {$N_1^0$};
		\node [style={agent_group}] (3) at (-3.5, -4) {$N_2^0$};
		\node [style={agent_group}] (4) at (-5.75, 1.75) {$N_3^0$};
		\node [style=agent] (5) at (-5.5, 7) {$g_3$};
		\node [style=agent] (6) at (8.75, -2.75) {$g_1$};
		\node [style=agent] (7) at (-7, -8) {$g_2$};
	\end{pgfonlayer}
	\begin{pgfonlayer}{edgelayer}
		\draw [style=utility] (0) to node [midway, fill=white] {$T+2t$} (1);
		\draw [style=utility] (0) to node [midway, fill=white] {$-1$} (2);
		\draw [style=utility] (0) to node [midway, fill=white] {$-T$} (3);
		\draw [style=utility] (0) to node [midway, fill=white] {$t$} (4);
		\draw [style=utility] (0) to node [midway, fill=white] {$-T$} (5);
		\draw [style=utility] (3) to node [midway, fill=white] {$-T$} (7);
		\draw [style=utility] (2) to node [midway, fill=white] {$-T$} (6);
		\draw [style=utility] (1) to node [midway, fill=white] {$-T$} (6);
		\draw [style=utility] (4) to node [midway, fill=white] {$-T$} (5);

        \draw[thick,CoalitionColor, fill=CoalitionColor!50, fill opacity=0.2]  \convexpath{0, 4, 5}{1.7cm};
        \draw[thick,CoalitionColor, fill=CoalitionColor!50, fill opacity=0.2]  \convexpath{1, 6, 2}{1.7cm};
        \draw[thick,CoalitionColor, fill=CoalitionColor!50, fill opacity=0.2]  \convexpath{3, 7}{1.7cm};
	\end{pgfonlayer}
\end{tikzpicture}
            }
        \caption{Illustration of our construction in the proof of \Cref{lem:SVS_quota_out_no_convergence}. 
        A bold node represents a group of agents. 
        A weighted directed edge from an agent $a$ to an agent $b$ represents the valuation $\vf_a(b)$ and, in the case of a group of agents, applies to all group members individually. 
        The initial partition $\pi$ is indicated in blue. 
        We omit all valuations of agents in $A \setminus \{a_0\}$ to other agents in $A$, as well as the valuations from agents in $A$ to agents in $\Gamma$ that are not part of the same coalition in $\pi$.}
        \label{fig:SVS_quota_out_no_convergence}
    \end{figure}

    Since $\qout < 1$, there exists $t'\in \mathbb N$ such that $t'\ge \qout (1 + t')$. Let $\tout := \min\{t'\in \mathbb N\colon t'\ge \qout (1 + t')\}$.
    In other words, $\tout$ is the smallest number of agents required to favor some inside agent to leave so that this agent can leave if exactly one agent is against this deviation. 
    Now, define $$t := \max\{3, \tout + 1\}\quad \text{and} \quad T := t^2 - t\text.$$

    We construct a game $G = (N, \vf)$ with agent set $N = A \cup \Gamma$, where $A = \{ a_0, \ldots, a_{3t + 1} \}$ is a set of $m := 3t + 2$ \emph{deviating agents}, and $\Gamma = \{g_1, g_2, g_3\}$ is a set of \emph{grouping agents}. 
    In this proof, we read indices modulo $m$ mapping to the representative in $\{0,\dots,m-1\}$.
    For every integer $0 \le i \le 3t + 1$, let
    \begin{itemize}
        \item $N_1^i := \{ a_{i + j} \mid j \in \{2, \ldots, t\}\}$,
        \item $N_2^i := \{ a_{i + j} \mid j \in \{t+1, \ldots, 2t+1\}\}$, and
        \item $N_3^i := \{ a_{i + j} \mid j \in \{2t+2, \ldots, 3t+1\}\}$.
    \end{itemize}
    We will use these sets to define partitions.
    They essentially subdivide the deviating agents into three intervals of agents that encompass all deviating agents except for $a_i$ and $a_{i + 1}$.
    It holds that $|N_1^i| = t-1$, $|N_2^i| = t+1$, and $|N_3^i| = t$.
    
    For $0\le i \le 3t+1$, agent $a_i$ has the following valuations: 
    \begin{enumerate}[1.]
        \item Let $\vf_{a_i}(a_{i + 1}) = T+2t$.
        \item For each $a \in N_1^i$, let $\vf_{a_i}(a) = -1$.
        \item For each $a \in N_2^i$, let $\vf_{a_i}(a) = -T$.
        \item For each $a \in N_3^i$, let $\vf_{a_i}(a) = t$.
        \item For each $g \in \Gamma$, let $\vf_{a_i}(g) = -T$.
    \end{enumerate}
    All other valuations (i.e., the outgoing valuations of agents in $\Gamma$) are $0$.
    We illustrate our construction in \cref{fig:SVS_quota_out_no_convergence}.

    For $0\le i\le 3t+1$, define partition $\pi_i = \{C_1^i, C_2^i, C_3^i\}$ where $C_1^i = \{a_{i+1}, g_1\} \cup N_1^i$, $C_2^i = \{g_2\} \cup N_2^i$, and $C_3^i = \{a_i, g_3\} \cup N_3^i$.
    We claim that in $\pi_i$, agent $a_i$ has a $(\qout, 1)$-\SVS deviation to join $C_1^i$, whereas there exists no other \NS deviation (including by $a_i$).
    Therefore, there exists a unique $\chi$ deviation in $\pi_i$ as $\chi \devimplies \NS$ and $(\qout, 1)$-\SVS  $\devimplies \chi$.
    
    Note that the deviation by $a_i$ to join $C_1^i$ results in the partition $\pi_{i+2t+1}$, which is identical in terms of valuations up to shifting indices by $2t+1$.
    Hence, proving our claim establishes that the dynamics starting from $\pi_0$ must cycle.
    Without loss of generality, we prove our claim for the case $i=0$, and set $\pi := \pi_0$, and $C_j := C_j^i$ for $j\in [3]$.

    First, we will show that $a_0$ can perform a $(\qout, 1)$-\SVS deviation to join coalition $C_1$. 
    Observe that $\sum_{a \in C_3\setminus \{a_0\}} \vf_{a_0}(a) = t \cdot t - T = t$, while $\sum_{a \in C_1} \vf_{a_0}(a) = T+2t - T + (t-1) \cdot (-1) = t+1$.
    Hence, since $|C_3| = |C_1 \cup \{a_0\}|$, in all game classes, it holds that $a_0$ strictly prefers $C_1\cup\{a_0\}$ over $C_3$.
    
    Now, let $\bar N_3^0 = N_3^0 \setminus \{a_{3t+1}\}$ and consider $a \in \bar N_3^0$.
    Then, $\sum_{b \in C_3\setminus\{a\}} \vf_a(b) \geq T + t - T + (t-1) \cdot (-1) = 1$, but $\vf_a(a_0) = -1$. Hence, in all game classes, it holds that $a$ prefers $C_1 \setminus \{a_0\}$ over $C_1$. 
    Moreover, $g_3$ is indifferent over all coalitions, and only $a_{3t+1}$ votes against the deviation of $a_0$. 
    Hence, by the choice of~$t$, $a_0$ is allowed to abandon $C_3$. 
    
    Next, let us consider the welcoming coalition $C_1$. 
    Observe that, again, $g_1$ is indifferent between all coalitions. Further, in the case of an \ASHG and since all agents in $\{a_1\} \cup N_1^0$ have valuation $t$ for $a_0$, they will never object to $a_0$ joining. Otherwise, in case of an \FHG or \MFHG, given some $1\le i\le t-1$, 
    observe that $\uf_{a_i}(C_1) = \frac{T+2t - T + (i-1) \cdot t - (t-2-(i-1))}{\delta} = \frac{i \cdot (t+1) + 1}{\delta} \leq \frac{t \cdot (t+1)}{\delta}$, for some $\delta \in \{t+1, t+2\}$, and thus, $\uf_{a_i}(C_1) \leq t$. Then, since $\vf_{a_i}(a_0) = t$, no agent in $\{a_1\} \cup N_1^0 \setminus \{a_t\}$ objects to $a_0$ joining. Finally, it holds that $\sum_{b \in C_1\setminus \{a_t\}} \vf_{a_t}(b) = -T + (t-1) \cdot t = 0$, and, as $\vf_{a_t}(a_0) = t$, agent $a_t$ will also not object. 
    Together, the deviation of $a_0$ is a $(\qout, 1)$-\SVS deviation.

    It remains to show that no other \NS deviation is available.
    First, it is clear that $a_0$ cannot perform any different deviation, as she is individually rational in $C_1$ and has strictly negative utility for $C_2$. 
    Further, no grouping agent $g \in \Gamma$ can ever perform an \NS deviation, as they have a utility of~$0$ for all coalitions. 

    Next, for an agent in $a \in N_3^0$, it holds that $\sum_{b \in C_3\setminus \{a\}} \vf_a(b) \geq T+2t - T - t = t$. On the other hand, observe that $\sum_{b \in C_1} \vf_a(b) \leq 2 \cdot (-T) + (t-1) \cdot (-1) < 0$, and $\sum_{b \in C_2} \vf_a(b) \leq 2 \cdot (-T) + t \cdot t = -2t^2 + 2t + t^2 = 2t - t^2$, and, since $t \geq 3$, it follows that $\sum_{b \in C_2} \vf_a(b) < 0$. 
    Thus, $a$ has no incentive to deviate in any of the considered game classes.

    For an agent $a \in N_1^0$, it holds that $\sum_{b \in C_1} \vf_a(b) \geq (t-1) \cdot t - T = 0$, and it holds that $\sum_{b \in C_1} \vf_{a_1}(b) = T+2t - T + (t-1) \cdot (-1) \geq t$. 
    By contrast, for any agent $a \in N_1^0 \cup \{a_1\}$, we have $\sum_{b \in C_3} \vf_a(b) \leq 2 \cdot (-T) + t \cdot t < 0$, using $t\ge 3$ once again, and $\sum_{b \in C_2} \leq 2 \cdot (-T) + T + 2t + t \cdot (-1) = t - T < 0$. Thus, no agent in $C_1$ has an incentive to deviate.

    Finally, for the agents in $N^0_2$, it holds that $\sum_{b \in C_2\setminus \{a_{t+1}\}} \vf_{a_{t+1}}(b) = T+2t -T + (t-1) \cdot (-1) \geq t$, and $\sum_{b \in C_2\setminus \{a_{2t+1}\}} \vf_{a_{2t+1}}(b) = -T + t \cdot t = t$, with this value being at least as high for any other agent in $N_2^0$. 
    However, for $a \in N_2^0$, we have $\sum_{b \in C_3} \vf_a(b) \leq T + 2t + t \cdot (-1) + 2 \cdot (-T) = 2t - T = 3t - t^2 < 0$, using $t \geq 3$, and agent $a$ has no \NS deviation to join $C_3$. 
    Finally, it holds that $\sum_{b \in C_1} \vf_a(b) \leq t \cdot t - T = t$, and $a$ has no incentive to deviate in all game classes, since $|C_1 \cup \{a\}| = |C_2|$.

    This shows that no other \NS deviation exists, which completes the proof.
\end{proof}

The game constructed in the previous proof is constructed in a way that the first deviating agent could also start in a singleton coalition.
Then, the first deviation is still a permissible $\chi$ deviation, and the game can still cycle.
Moreover, if this agent is removed from the game, the $\chi$ dynamics can be shown to quickly converge.
In this way, we can also construct the game required to apply \Cref{thm:unified_NCD}.

\begin{lemma}\label{lem:SVS_quota_out_no_convergence_adapted}
    Let $\chi$ be a stability notion such that $\chi \devimplies \NS$ and $(\qout, 1)$-\SVS  $\devimplies \chi$ for some $\qout \in [0, 1[$. 
    Then, there exists an \ASHG, \FHG, and \MFHG $G = (N, \vf)$, and partition $\pi$ of $N$ that contains a singleton coalition $\{a\} \in \pi_\chi$, such that the $\chi$ dynamics can cycle on $(G_\chi, \pi_\chi)$, but must converge on $(G_\chi - a, \pi_\chi - a)$.
\end{lemma}

\begin{proof}
    Consider the game $G = (N, \vf)$ constructed in the proof of \Cref{lem:SVS_quota_out_no_convergence}. 
    We alter $\pi^0$ to obtain a starting partition $\pi'$ by placing agent $a_0$ in her singleton coalition while leaving all other agents untouched, i.e., we define 
    $\pi' = \{\{a_0\}\}\cup \{C_1,C_2,C_3\}$ with $C_1 = \{a_1, g_1\} \cup N_1^0$, $C_2 = \{g_2\} \cup N_2^0$, and $C_3 = \{g_3\} \cup N_3^0$.
    We claim that the altered partition satisfies the assertion.

    First, to see that the $\chi$ dynamics of $(G, \pi')$ can cycle, observe that, in the first deviation of the $\chi$ dynamics of $(G, \pi^0)$, agent $a_0$ had a strictly positive utility for the welcoming coalition. 
    Thus, $a_0$ can still increase her utility by deviating from her singleton coalition to join $C_1$, and the deviation is furthermore still a $(\qout, 1)$-\SVS deviation, as the support in coalition $C_1$ remains unchanged. 
    But then, the resulting partition is identical to a partition in the cycling dynamics in the proof of \Cref{lem:SVS_quota_out_no_convergence} after the first deviation, and, going forward, the $\chi$ dynamics must cycle. 

    Next, we claim that the $\chi$ dynamics of $(G - a_0, \pi' - \{a_0\})$ must converge.
    First, one can verify that, from the initial state, only agent $a_{2t+1}$ can perform an \NS deviation to join $C_3$. 
    Afterwards, only agent $a_{t}$ can perform an \NS deviation to join $C_2$. 
    However, in the resulting partition, agent $a_{3t+1}$ cannot deviate to join $C_1$ (as she could in the original construction), because the positive valuation for agent $a_0$ is missing as an incentive.
    In fact, based on the arguments given in the proof of \Cref{lem:SVS_quota_out_no_convergence}, it is easy to see that no other \NS deviation is possible at this stage.
    Hence, the \NS dynamics, and, therefore, the $\chi$ dynamics converges.
\end{proof}

We continue with constructing games for deviation concepts implying $(1, \qin)$-\SVS deviations.
The constructed example is similar, and even simpler, as we now have deviating agents that only alternate between two coalitions.

\begin{lemma}\label{lem:SVS_quota_in_no_convergence}
    Let $\chi$ be a stability notion such that $\chi \devimplies \NS$ and $(1, \qin)$-\SVS  $\devimplies \chi$ for some $\qin \in [0, 1[$. 
    Then, there exists an \ASHG, \FHG, and \MFHG $G = (N, \vf)$, and partition $\pi$ of $N$ such that the $\chi$ dynamics on $G$ must cycle when starting on $\pi$.
\end{lemma}

\begin{proof}
    \begin{figure}
        \centering
        \resizebox{.5\textwidth}{!}{
            \begin{tikzpicture}
	\begin{pgfonlayer}{nodelayer}
		\node [style=agent] (0) at (4.5, 0) {$a_0$};
		\node [style=agent] (3) at (240:4.5) {$a_{t+1}$};
		\node [style={agent_group}] (4) at (0,0) {$N_2^0$};
		\node [style=agent] (5) at (120:4.5) {$g_2$};
		\node [style=agent] (6) at (16, 0) {$g_1$};
		\node [style={agent_group}] (1) at ($(0)!.5!(6)$) {$N_1^0$};
	\end{pgfonlayer}
	\begin{pgfonlayer}{edgelayer}
		\draw [style=utility] (0) to node [midway, fill=white] {$-t$} (3);
		\draw [style=utility] (0) to (4);
		\draw [style=utility] (0) to (5);
		\draw [style=utility] (1) to (6);
		\draw [style=utility] (4) to (5);
		\draw [style=utility] (3) to (5);
		\draw [style=utility, bend left] (0) to (6);
		\draw [style=utility, bend left] (4) to (6);
		\draw [style=utility, bend right] (1) to (5);
		\draw [style=utility] (0) to node [midway, fill=white] {$0$} (1);

        \draw[thick,CoalitionColor, fill=CoalitionColor!50, fill opacity=0.2]  \convexpath{0, 3, 5}{1.7cm};
        \draw[thick,CoalitionColor, fill=CoalitionColor!50, fill opacity=0.2]  \convexpath{1, 6}{1.7cm};
	\end{pgfonlayer}
\end{tikzpicture}
            }
        \caption{Illustration of our construction in the proof of \Cref{lem:SVS_quota_in_no_convergence}. 
        A bold node represents a group of agents. 
        A weighted directed edge from an agent $a$ to an agent $b$ represents the valuation $\vf_a(b)$ and, in the case of a group of agents, applies to all group members individually. 
        Unweighted edges represent a valuation of~$1$.
        The initial partition $\pi$ is indicated in blue. 
        We omit all valuations of agents in $A \setminus \{a_0\}$ to other agents in $A$.
        }
        \label{fig:SVS_quota_in_no_convergence}
    \end{figure}
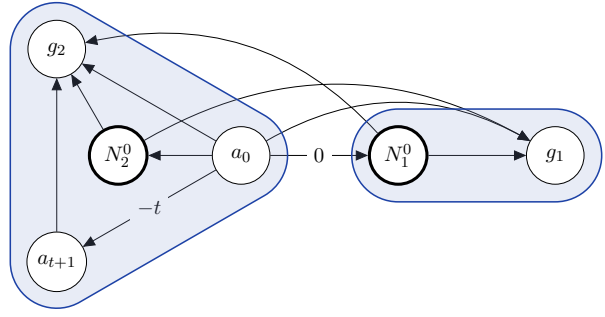

    We start by defining a threshold similar to the proof of \Cref{lem:SVS_quota_out_no_convergence}.
    Since $\qin < 1$, there exists $t'\in \mathbb N$ such that $t'\ge \qin (1 + t')$. 
    Let $\tin := \min\{t'\in \mathbb N\colon t'\ge \qin (1 + t')\}$.
    In other words, $\tin$ is the smallest number of agents required to favor some outside agent to join a coalition so that this agent can join if exactly one agent of the coalition is against this deviation. 
    We define $$t := \max\{3, \tin + 1\}\text.$$

    We construct a game $G = (N, \vf)$ with agent set $N = A \cup \Gamma$, where $A = \{ a_0, \ldots, a_{2t} \}$ is a set of $m := 2t + 1$ \emph{deviating agents}, and $\Gamma = \{g_1, g_2\}$ is a set of \emph{grouping agents}. 
    As in the proof of \Cref{lem:SVS_quota_out_no_convergence}, we read indices modulo $m$ mapping to the representative in $\{0,\dots,m-1\}$.
    
    For every integer $0\le i \le 2t$, define
    \begin{itemize}
        \item $N_1^i := \{ a_{i +_{[n]} j} \mid j \in \{1, \ldots, t\}\}$ and
        \item $N_2^i := \{ a_{i +_{[n]} j} \mid j \in \{t+2, \ldots, 2t\}\}$.
    \end{itemize}

    We will use these sets to define partitions.
    They essentially subdivide the deviating agents into two intervals of agents that encompass all deviating agents except for $a_i$ and $a_{i+t+1}$.
    It holds that $|N_1^i| = t$ and $|N_2^i| = t-1$.
    
    For $0\le i\le 2t$, agent $a_i$ has the following valuations: 
    \begin{enumerate}[1.]
        \item Let $\vf_{a_i}(a_{i + t+1}) = -t$.
        \item For each $a \in N_2^i$, let $\vf_{a_i}(a) = 1$.
        \item For each $g \in \Gamma$, let $\vf_{a_i}(g) = 1$.
    \end{enumerate}
    All other valuations are $0$. 
    We remark that the two grouping agents are indifferent over all possible coalitions.
    The construction is illustrated in \Cref{fig:SVS_quota_out_no_convergence}.

    For $0\le i\le 2t$, define partition $\pi_i = \{C_1^i,C_2^i\}$ where $C_1^i = \{g_1\} \cup N_1^i$, and $C_2^i = \{a_i, a_{i + t+1}, g_2\} \cup N_2^i$.
    We claim that in $\pi_i$, agent $a_i$ has a $(1,\qin)$-\SVS deviation to join $C_1^i$, whereas there exists no other \NS deviation (including by $a_i$).
    Therefore, there exists a unique $\chi$ deviation in $\pi_i$ as $\chi \devimplies \NS$ and $(1,\qin)$-\SVS  $\devimplies \chi$.
    
    Note that the deviation by $a_i$ to join $C_1^i$ results in the partition $\pi_{i+t}$, which is identical in terms of valuations up to shifting indices by $t$.
    Hence, proving our claim establishes that the dynamics starting from $\pi_0$ must cycle.
    Without loss of generality, we prove our claim for the case $i=0$, and set $\pi := \pi_0$, and $C_j := C_j^i$ for $j\in [2]$.

    First, we show that $a_0$ has a $(1, \qin)$-\SVS deviation to join coalition $C_1$.
    Indeed, observe that $\sum_{a \in C_2\setminus \{a_0\}} \vf_{a_0}(a) = (t-1) + 1 - t = 0$, while $\sum_{a \in C_1} \vf_{a_0}(a) = t \cdot 0 + 1 = 1$, and agent $a_0$ strictly prefers $C_1$ over $C_2$ in all game classes.
    
    Next, observe that all agents in $N_2^0 \cup \{a_{t+1}\}$ have a utility strictly greater than $0$ for coalition $C_2$, but a valuation of~$0$ for agent $a_0$, and hence do not object to agent $a_0$ leaving.
    
    Further, for $1\le i\le t-1$, it holds that $\vf_{a_i}(a_0) = 1$. 
    Hence, $a_i$ is in favor of $a_0$ joining in the case of an \ASHG. 
    Otherwise, in case of an \FHG or \MFHG, we have $\uf_{a_i}(C_1) = \frac{i-1}{\delta} < \frac{i}{\delta+1} = \uf_{a_i}(C_1\cup\{a_0\})$ with $\delta \in \{t, t+1\}$.
    It follows that $a_i$ is also in favor of $a_0$ joining. 
    Moreover, only $a_t$ is against the deviation. 
    Hence, by the choice of~$t$, the deviation of $a_0$ is a $(1, \qin)$-\SVS deviation.

    It remains to show that no other \NS deviation is possible.
    First, as $a_0$ is in an individually rational coalition, she cannot deviate to her singleton coalition.

    Moreover, for every $i\in [2t]$, 
    agents $a_i$ and $a_{i + t + 1}$ are in different coalitions.
    Thus, $a_i$ is in an individually rational coalition, and cannot receive a positive utility by deviating to the other nonempty coalition. 
    Therefore, $a_i$ cannot perform an \NS deviation.
\end{proof}

Once again, we can modify the example but extracting the first deviating agent from the initial partition and place her in a singleton coalition to obtain a game suitable for applying \Cref{thm:unified_NCD}.

\begin{lemma}\label{lem:SVS_quota_in_no_convergence_adapted}
    Let $\chi$ be a stability notion such that $\chi \devimplies \NS$ and $(1,\qin)$-\SVS  $\devimplies \chi$ for some $\qin \in [0, 1[$. 
    Then, there exists an \ASHG, \FHG, and \MFHG $G = (N, \vf)$, and partition $\pi$ of $N$ that contains a singleton coalition $\{a\} \in \pi_\chi$, such that the $\chi$ dynamics can cycle on $(G_\chi, \pi_\chi)$, but must converge on $(G_\chi - a, \pi_\chi - a)$.
\end{lemma}

\begin{proof}
    Consider the game $G = (N, \vf)$ constructed in the proof of \Cref{lem:SVS_quota_in_no_convergence}. 
    We alter $\pi^0$ to obtain a starting partition $\pi'$ by placing agent $a_0$ in her singleton coalition while leaving all other agents untouched, i.e., we define 
    $\pi' = \{\{a_0\}\}\cup \{C_1,C_2\}$ with $C_1 = \{g_1\} \cup N_1^0$ and $C_2 = \{a_t+1,g_2\} \cup N_2^0$.
    We claim that the altered partition satisfies the assertion.

    First, to see that the $\chi$ dynamics on $(G, \pi')$ can cycle, observe that $\uf_{a_0}(C_1\cup \{a_0\}) > 0$. 
    Now, by virtue of the first deviation of $a_0$ to $C_1$ being a $(1, \qin)$-\SVS deviation when $a_0$ was abandoning $C_2$, this deviation is still a $(1, \qin)$-\SVS deviation when $a_0$ abandons her singleton coalition. 
    Afterwards, due to the arguments given in the proof of \Cref{lem:SVS_quota_in_no_convergence}, the $\chi$-dynamics must cycle. 

    Next, on $(G - a_0, \pi' - a_0)$, observe that for every $i\in [2t]\setminus \{t\}$, agents $a_i$ and $a_{i+t+1}$ are still in different coalitions and, therefore, $a_i$ has no \NS deviation, as argued in the proof of \Cref{lem:SVS_quota_in_no_convergence}.
    Further, for agent $a_t$, it holds that $\sum_{a \in C_1\setminus \{a_t\}} \vf_{a_t}(a) = t \geq 3$, while $\sum_{a \in C_2} \vf_{a_t}(a) = 1$, while $|C_1| < |C_2 \cup \{a_t\}|$. 
    Thus, $a_t$ strictly prefers to stay in $C_1$ and has no \NS deviation.
    Hence, $\pi' - a_0$ is an \NS partition and, therefore, a $\chi$ partition on $G - a_0$, and the $\chi$
    dynamics converge without a single deviation.
\end{proof}

\section{Missing Proofs in Section~\ref{sec:CIS}}\label{app:CIS}

In this section, we present the missing proof concerning contractual individual stability.

\subsection{Proof of \Cref{thm:CIS_IR_NP}}\label{app:CIS_IR_NP}

In this section, we present the full proof of \Cref{thm:CIS_IR_NP}.
We provide a reduction from \problemname{Independent Set (\INDSET)}, defined as follows. 

\problemdefijcai
{Independent Set (\INDSET)}
{An undirected graph $G = (V, E)$ and a positive integer $k$.}
{Does $G$ admit an independent set $I\subseteq V$ of size at least $k$, i.e.,$|I|\ge k$ and $\{v,w\}\notin E$ for all $v,w\in I$?}

It is well-known that \INDSET is \NP-complete \citep{Karp72a}.

\CisIrNp* 

\begin{proof}
    \begin{figure}[ht!]
        \centering
        \resizebox{.4\textwidth}{!}{
            \begin{tikzpicture}
	\begin{pgfonlayer}{nodelayer}
		\node [style=agent] (0) at (-6, 4) {$a$};
		\node [style=agent] (1) at (-2, 4) {$b$};
		\node [style=agent] (2) at (2, 4) {$c$};
		\node [style=agent] (3) at (6, 4) {$d$};
		\node [style=agent] (4) at (-2, -1) {$g^+$};
		\node [style=agent] (6) at (2, -1) {$g^-$};
		\node [style=agent] (7) at (0, -5) {$y$};
		\node [style=agent] (8) at (0, -9) {$p$};
		\node [style=agent] (9) at (-2, -7) {$x$};
		\node [style={agent_group}] (10) at (2, -7) {$D$};
		\node [style=agent] (11) at (-6, 7) {$r_a^+$};
		\node [style=agent] (12) at (-6, 10) {$r_a^-$};
		\node [style=agent] (13) at (-2, 7) {$r_b^+$};
		\node [style=agent] (14) at (-2, 10) {$r_b^-$};
		\node [style=agent] (15) at (2, 7) {$r_c^+$};
		\node [style=agent] (16) at (2, 10) {$r_c^-$};
		\node [style=agent] (17) at (6, 7) {$r_d^+$};
		\node [style=agent] (18) at (6, 10) {$r_d^-$};
	\end{pgfonlayer}
	\begin{pgfonlayer}{edgelayer}
		\draw [style=utility-symmetricc] (4) to (0);
		\draw [style=utility-symmetricc] (4) to (1);
		\draw [style=utility-symmetricc] (4) to (2);
		\draw [style=utility-symmetricc] (4) to (3);
		\draw [style=utility-symmetricc] (4) to (6);
		\draw [style=utility-symmetricc] (0) to (6);
		\draw [style=utility-symmetricc] (1) to (6);
		\draw [style=utility-symmetricc] (2) to (6);
		\draw [style=utility-symmetricc] (3) to (6);
		\draw [style=utility-symmetricc] (7) to (4);
		\draw [style=utility] (6) to (7);
		\draw [style=utility] (7) to (9);
		\draw [style=utility-symmetricc] (10) to (8);
		\draw [style=utility] (7) to (10);
		\draw [style=utility-symmetricc] (0) to (1);
		\draw [style=utility-symmetricc] (1) to (2);
		\draw [style=utility] (0) to (11);
		\draw [style=utility-symmetricc] (11) to (12);
		\draw [style=utility-symmetricc] (13) to (14);
		\draw [style=utility-symmetricc] (15) to (16);
		\draw [style=utility-symmetricc] (17) to (18);
		\draw [style=utility] (1) to (13);
		\draw [style=utility] (2) to (15);
		\draw [style=utility] (3) to (17);
		\draw [style=utility-symmetricc, in=90, out=-120] (2) to (7);
		\draw [style=utility-symmetricc, in=90, out=-60] (1) to (7);
		\draw [style=utility-symmetricc, bend left] (3) to (7);
		\draw [style=utility-symmetricc, bend right] (0) to (7);

        \draw[thick,CoalitionColor, fill=CoalitionColor!50, fill opacity=0.2]  \convexpath{0, 12}{1.6cm};
        \draw[thick,CoalitionColor, fill=CoalitionColor!50, fill opacity=0.2]  \convexpath{1, 14}{1.6cm};
        \draw[thick,CoalitionColor, fill=CoalitionColor!50, fill opacity=0.2]  \convexpath{2, 16}{1.6cm};
        \draw[thick,CoalitionColor, fill=CoalitionColor!50, fill opacity=0.2]  \convexpath{3, 18}{1.6cm};
        \draw[thick,CoalitionColor, fill=CoalitionColor!50, fill opacity=0.2]  \convexpath{4, 6}{1.6cm};
        \draw[thick,CoalitionColor, fill=CoalitionColor!50, fill opacity=0.2]  \convexpath{7, 10, 8, 9}{1.7cm};
	\end{pgfonlayer}
\end{tikzpicture}
            }
        \caption{Illustration of reduction in \Cref{thm:CIS_IR_NP}. The reduced instance for the source instance $(\mathcal{G} = (V, E), k)$ is displayed, where $V = \{a, b, c, d\}$, $E = \{\{a, c\}, \{a, d\}, \{b, d\}, \{c, d\}\}$, and $k = 2$.
        An edge from agent $a$ to agent $b$ represents the valuation $\vf_a(b) = f^+(n)$. We note that the bold node $D$ represents the whole set $D$ of agents, and we omit the positive valuations among all agents in $D$. All other omitted edges represent the valuation $f^-(n)$. The initial coalitions according to $\pi$ are displayed by blue boxes.}
        \label{fig:CIS_IR_NP}
    \end{figure}
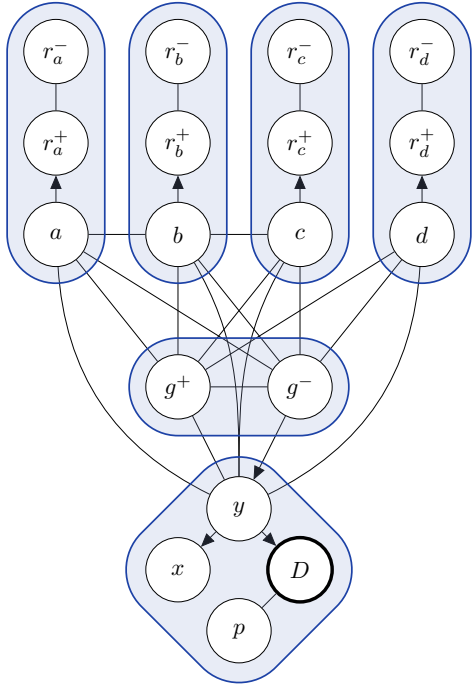

    We reduce from \INDSET. 
    Let $(\mathcal{G} = (V, E), k)$ be an instance of \INDSET, and without loss of generality, assume that $k \geq 3$. 
    We construct a game $G = (N, \vf)$ with a set $N = N_V \cup R_V \cup D \cup \{g^+, g^-, x, y, p\}$ of agents.
    There, $N_V = \{a_v\}_{v \in V}$ is a set of \emph{vertex agents}, $R_V = \{r_v^+, r_v^-\}_{v \in V}$ is a set of \emph{restricting vertex agents}, and $D = \{d_i\}_{i \in [k-1]}$ is a set of \emph{dummy agents}. 
    We define valuations as follows:
    \begin{enumerate}[1.]
        \item Let $\vf_y(a) = f^+(n)$ for all $a \in D \cup N_V \cup \{x, g^+\}$.
        \item Let $\vf_d(d') = \vf_d(p) = \vf_p(d) = f^+(n)$ for all $d, d' \in D$ with $d \neq d'$.
        \item Let $\vf_{g^+}(a) = \vf_{g^-}(a) = f^+(n)$ for all $a \in N_V \cup \{y\}$, and let $\vf_{g^+}(g^-) = \vf_{g^-}(g^+) = f^+(n)$.
        \item Let $\vf_{a_v}(a_w) = f^+(n)$ for all $v, w \in V$ with $v \neq w$ and $\{v, w\} \not\in E$.
        \item Let $\vf_{a_v}(r_v^+) = \vf_{r_v^+}(r_v^-) = \vf_{r_v^-}(r_v^+) = f^+(n)$.
        \item Let all other valuations be $f^-(n)$.
    \end{enumerate}
    We set the initial partition $\pi = \{\{x, y, p\} \cup D, \{\{a_v, r_v^+, r_v^-\} \mid v \in V\}, \{g^+, g^-\}\}$, i.e., we place agents $g^+$ and $g^-$ together, agents $x$, $y$, and $p$ with $D$, and all vertex agents with their corresponding restricting vertex agents. 
    We will sometimes refer to the coalition that contains agents $g^+$ and $g^-$ as the \emph{grouping coalition} $\Gamma$. 
    We depict the schematic of our construction in \Cref{fig:CIS_IR_NP}.\footnote{Note that our reduction assumes that $k\ge 3$, but we chose an example with $k = 2$ to avoid overloading the illustration.}

    Note that $\pi$ is neither a \CIS partition, as any vertex candidate can deviate to join $g^+$, nor individually rational, as, e.g., $x$ would rather be in her singleton coalition.

    Further, since $G$ only contains one negative and one positive valuation value, we can make the following observation.

    \begin{observation}\label{obs:CIS_IR_NP_fin_fout}
        Given an agent $a \in N$, a coalition $C$, and an agent $c \in C \setminus \{a\}$, for the game classes \ASHG and \FHG, it holds that $c \in \Fin(C, a)$ if $\vf_c(a) = f^+(n)$, and $c \in \Fout(C, a)$ otherwise. For the game class \MFHG, it holds that $c \in \Fin(C, a)$ if $\vf_c(a) = f^+(n)$ and $\exists c' \in C: \vf_c(c') = f^-(n)$ or $C \setminus \{a\} = \{c\}$, and $c \in \Fout(C, a)$ if $\vf_c(a) = f^-(n)$ and $\exists c' \in C: \vf_c(c') = f^+(n)$ or $C \setminus \{a\} = \{c\}$.
    \end{observation}

    We now begin our in-depth analysis of the \CIS dynamics of the construced instance $(G, \pi)$ by proving several claims. 

    \begin{claim}\label{claim:CIS_IR_NP_D_x_p_no_dev}
        No agent in $D \cup \{x, p\}$ can perform a CIS deviation abandoning coalition $D \cup \{x, y, p\}$.
    \end{claim}
    \begin{claimproof}
        Observe that agent $y$ has valuation $\vf_y(a) = f^+(n)$ for all agents in $a \in D \cup \{x\}$, and $\vf_y(p) = f^-(n)$. 
        Hence, by \cref{obs:CIS_IR_NP_fin_fout}, it holds that $y \in \Fin(D \cup \{x, y, p\}, a)$, and no such agent $a$ can perform a \CIS deviation from $D \cup \{x, y, p\}$. 
        Further, since $k \geq 3$, there exists an agent $d \in D$.
        Using \cref{obs:CIS_IR_NP_fin_fout} again, and since $\vf_d(p) = f^+(n)$ and $\vf_d(y) = f^-(n)$, one can see that $p$ cannot deviate.
    \end{claimproof}

    \begin{claim}\label{claim:CIS_IR_NP_R_no_dev}
        Let $v \in V$. Then no restricting vertex agent $r_v \in \{r_v^+, r_v^-\}$ can deviate from her initial coalition $\{a_v, r_v^+, r_v^-\}$, or from the coalition $\{r_v^+, r_v^-\}$.
    \end{claim}
    \begin{claimproof}
        Observe that both $r_v^+$ and $r_v^-$ get nonnegative utility from both coalitions $\{a_v, r_v^+, r_v^-\}$ and $\{r_v^+, r_v^-\}$, since $\vf_{r_v^+}(r_v^-) = \vf_{r_v^-}(r_v^+) = f^+(n)$, and $f^+(n) \geq |f^-(n)|$. Hence, neither agent can deviate to her singleton coalition. Further, since both agents $r_v^+$ and $r_v^-$ have valuation $f^-(n)$ for all agents in $N \setminus \{r_v^+, r_v^-\}$, any deviation to another nonempty coalition would strictly decrease their utility. 
    \end{claimproof}

    \begin{claim}\label{claim:CIS_IR_NP_V_dev_iff}
        A vertex agent $a_v \in N_V$ can \CIS-deviate from her initial coalition $\{a_v, r_v^+, r_v^-\}$ to join a nonempty coalition $C \subseteq \{g^+, g^-\} \cup N_V \cup \{y\}$ if and only if there is no $a_w \in N_V \cap C$ with $\{v, w\} \in E$. Further, $a_v$ cannot deviate to her singleton coalition, a coalition $\{r_w^-, r_w^+\}$ for some $w \in V \setminus \{v\}$, or any coalition $C' \in \pi \setminus \{\{a_v, r_v^+, r_v^-\}\}$.
    \end{claim}
    \begin{claimproof}
        To see that $a_v$ can deviate to such a coalition $C$ that contains no such agent $a_w$, observe that, by \cref{obs:CIS_IR_NP_fin_fout}, and since $\vf_c(a_v) = f^+(n)$ for all $c \in C$, it holds that $\Fout(C, a_v) = \emptyset$. Further, since $a_v$ has positive valuation for all (at least one) agents in $C$, but dislikes $r_{v}^-$ in her original coalition, she strictly improves her utility in all game classes.
        
        Second, assume that $C$ contains such an agent $a_w$. Then, since $\vf_{a_w}(a_v) = f^-(n)$, again by \cref{obs:CIS_IR_NP_fin_fout}, agent $a_v$ cannot join $C$.
        Finally, to see that $a_v$ cannot deviate to her singleton coalition, observe that, since $f^+(n) \geq |f^-(n)|$, agent $a_v$ has positive utility for her initial coalition $\{a_v, r_v^+, r_v^-\}$. Further, $a_v$ cannot join coalition $D \cup \{x, y, p\}$, since, e.g., $\vf_p(a_v) = f^-(n)$, and $a_v$ cannot join $\{r_w^+, r_w^-\}$ or $\{a_w, r_w^+, r_w^-\}$ with $w \in V \setminus \{v\}$ by \cref{obs:CIS_IR_NP_fin_fout} since, e.g., $\vf_{r_w^+}(a_v) = f^-(n)$ but $\vf_{r_w^+}(r_w^-) = f^+(n)$.
    \end{claimproof}

    \begin{claim}\label{claim:CIS_IR_NP_y_dev_iff}
        Let $N_{V}' \subseteq N_V$ be a subset of $m \in \NN$ vertex agents, and let $\Tilde{\pi}$ be the partition that resulted from $\pi$ after all agents in $N_{V}'$ joined the grouping coalition $\Gamma$. Then, $y$ can deviate to join $\Gamma$ if and only if $m \geq k$, and otherwise, $y$ cannot deviate to join any coalition in $\emptyset \cup \Tilde{\pi} \setminus \{\Tilde \pi(y)\}$.
    \end{claim}
    \begin{claimproof}
        Due to \cref{obs:CIS_IR_NP_fin_fout}, and because $y$ is ``disliked'' by all agents in $\Tilde\pi(y) = \pi(y)$, and ``liked'' by all agents in $\Tilde{\pi}(g^+)$, no agent will ever object to $y$ performing such a deviation. Hence, we will focus our argument on whether $y$ can improve her utility. 
    
        Observe that $\sum_{c \in \Tilde\pi(y)} \vf_y(c) = k \cdot f^+(n) - f^-(n)$, and $\sum_{c \in \Tilde{\pi}(g^+)} \vf_y(c) = (m+1) \cdot f^+(n) - f^-(n)$. Then, for the game class \ASHG, the statement is immediate. Otherwise, for the game classes \FHG and \MFHG, it holds that $\uf_y(\Tilde{\pi}(y)) = \frac{k \cdot f^+(n) - f^-(n)}{k + s}$ and $\uf_y(\Tilde{\pi}(g^+) \cup \{y\}) = \frac{(m+1) \cdot f^+(n) - f^-(n)}{m+s+1}$ with $s = 2$ and $s = 1$, respectively. Then, it holds that:
        \begin{align*}
            & \uf_y(\Tilde{\pi}(g^+) \cup \{y\}) - \uf_y(\Tilde{\pi}(y)) \\
            = & \frac{(m+1) \cdot f^+(n) - f^-(n)}{m+s+1} - \frac{k \cdot f^+(n) - f^-(n)}{k + s} \\
            = & \frac{(m+1) \cdot (s \cdot f^+(n) + f^-(n)) - k \cdot (s \cdot f^+(n) + f^-(n))}{(m+s+1) \cdot (k + s)} \\
            = & (m+1 - k) \cdot \frac{s \cdot f^+(n) + f^-(n)}{(m+s+1) \cdot (k + s)}.
        \end{align*}
        Now, since $s, k > 0$, $m \geq 0$, and $f^+(n) \geq |f^-(n)|$, we have $\frac{s \cdot f^+(n) + f^-(n)}{(m+s+1) \cdot (k + s)} > 0$, and hence, $\uf_y(\Tilde{\pi}(g^+) \cup \{y\}) > \uf_y(\Tilde{\pi}(y))$ if and only if $m \geq k$.

        Now, to see that $y$ cannot deviate to her singleton coalition, observe that, since $\sum_{c \in \Tilde \pi(y)} \vf_y(c) = k \cdot f^+(n) - f^-(n) > 0$, it holds that $\uf_y(\pi(y)) > 0$ for all game classes. Finally, by \cref{obs:CIS_IR_NP_fin_fout}, $y$ cannot deviate to join a coalition $\{r_v^+, r_v^-\}$ or $\{a_v, r_v^+, r_v^-\}$ for some $v \in V$, since, e.g., $\vf_{r_v^+}(y) = f^-(n)$ but $\vf_{r_v^+}(r_v^-) = f^+(n)$.
    \end{claimproof}

    We claim that there is a sequence of \CIS deviations on $(G, \pi)$ that leads to an individually rational \CIS partition if and only if there is a size-$k$ independent set in~$\mathcal{G}$. 

    ($\Rightarrow$) Assume that there is a sequence of \CIS deviations on $(G, \pi)$ that leads to an individually rational \CIS partition. Observe that no agent in $\{g^+, g^-\}$ has ever has incentive to deviate, as they both like each other in their original coalition, but, due to the behaviour described in the above claims, the agents in $N_V \cup \{y\}$ (i.e., the only other agents that are liked by $g^+$ and $g^-$) can only deviate to join $\{g^+, g^-\}$. Now, because $x$ is not individually rational in the original coalition, and due to \cref{claim:CIS_IR_NP_D_x_p_no_dev,,claim:CIS_IR_NP_R_no_dev}, it is easy to see that $y$ needs to deviate away to join the grouping coalition. However, as shown in \cref{claim:CIS_IR_NP_y_dev_iff}, this can only happen if there is a size-$k$ independent set in $\mathcal{G}$.

    ($\Leftarrow$) Assume that there is a size-$k$ independent set $I \subseteq V$. Then, by \cref{claim:CIS_IR_NP_V_dev_iff}, all agents in $\{a_v \in N_V\}_{v \in I}$ can deviate to join $G$. Afterwards, due to \cref{claim:CIS_IR_NP_y_dev_iff} and since $|I| = k$, agent $y$ can deviate to join $\Gamma$. Finally, let $x$ deviate from $D \cup \{x, p\}$ to her singleton coalition, which is a \CIS deviation, since $\vf_x(c) = \vf_c(x) = f^-(n)$ for all $c \in D \cup \{p\}$. 
    
    Now, to see all agents have nonnegative utility in the resulting coalition, observe that all agents in $D \cup \{p\}$ like each other, all restricting agents in a coalition $\{r_v+, r_v^-\}$ (i.e., $v \in I$) like each other, and all agents in a coalition $\{a_v, r_v+, r_v^-\}$ (i.e. $v \not\in I'$) have positive utility for their coalitions, since $f^+(n) \geq f^-(n)$. Further, the statement trivially holds for agents $g^+$, $g^-$, and for all agents that have deviated. Now, due to \cref{obs:CISconvergence,,obs:CIS_dyn_maintain_IR}, it is clear that the \CIS dynamics must converge to an individually rational \CIS partition.
\end{proof}

\subsection{Proof of \Cref{claim:CIS_dyn_shortcut_singleton_exactly_one}}

\CISshortcut*

\begin{proof} 
    We prove the statement by induction over the number of deviations. 
    Clearly, the statement holds for the base case of the singleton partition. 
    Now, assume that the statement holds for all executions of the \CIS dynamics of length $k$.
    Then, in case the next deviation is performed by an agent that has previously deviated, it cannot be to her singleton coalition, as each \CIS deviation must strictly increase the deviator's utility, and weakly increase that of everyone else.
    Hence, the assertion follows after the $(k+1)$st deviation.
    
    Otherwise, in case the subsequent deviation is performed by an agent that has never deviated, observe that such an agent can only deviate from her singleton coalition.
    Indeed, it is easy to see that 
    as long as such an agent is in a coalition that contains an agent that \CIS-deviated to join her, she is prohibited from leaving via a \CIS deviation. 
    But then, she must join a different, nonempty coalition, which, by the induction hypothesis, contains exactly one agent that never deviated. 
    Then, the abandoned coalition is empty, and the welcoming coalition is joined by an additional agent who now performed a deviation. 
    Hence, the statement holds for all \CIS deviation sequences of length $k+1$. 
\end{proof}

\subsection{Proof of \Cref{thm:CIS_EXP}}\label{app:CIS_EXP}

In this section, we present the full proof of \Cref{thm:CIS_EXP} restated as follows.

\CisExp* 

\begin{proof}
    \begin{figure*}[ht]
        \centering
        \begin{subfigure}[b]{0.32\textwidth}
            \centering
            \resizebox{1\textwidth}{!}{
                \begin{tikzpicture}
	\begin{pgfonlayer}{nodelayer}
		\node [style=none] (0) at (-1.5, 5) {$i-1$};
		\node [style=none] (1) at (1.5, 5) {$i$};
		\node [style=none] (2) at (4, 5) {$\ldots$};
		\node [style=none] (3) at (6.5, 5) {$T$};
		\node [style={large_agent}] (5) at (6.5, 2) {$t$};
		\node [style={large_agent}] (6) at (6.5, -1.5) {$\mathit{t}_1$};
		\node [style=none] (7) at (6.5, -3.75) {$\vdots$};
		\node [style={large_agent}] (8) at (6.5, -6.5) {$\mathit{t}_i$};
		\node [style={large_agent}] (9) at (6.5, -10) {$\mathit{t}_{i+1}$};
		\node [style=none] (10) at (-4, 5) {$\ldots$};
		\node [style=none] (21) at (-6.5, 5) {$1$};
		\node [style={large_group}] (22) at (-1.5, 2) {$B_{i-1}$};
		\node [style={large_group}] (23) at (-6.5, 2) {$B_1$};
		\node [style=none] (24) at (-4, 2) {$\ldots$};
		\node [style=none] (25) at (4, 2) {$\ldots$};
		\node [style=none] (30) at (4, -1.5) {$\ldots$};
		\node [style=none] (35) at (4, -6.5) {$\ldots$};
		\node [style=none] (40) at (4, -10) {$\ldots$};
		\node [style=none] (41) at (4, -3.75) {$\vdots$};
		\node [style={large_group}] (4) at (1.5, 2) {$B_i$};
	\end{pgfonlayer}
    \begin{pgfonlayer}{edgelayer}
        \draw[thick,CoalitionColor, fill=CoalitionColor!50, fill opacity=0.2]  \convexpath{5, 23}{1.6cm};
        \draw[thick,CoalitionColor, fill=CoalitionColor!50, fill opacity=0.2]  \convexpath{6, 30}{1.6cm};
        \draw[thick,CoalitionColor, fill=CoalitionColor!50, fill opacity=0.2]  \convexpath{8, 35}{1.6cm};
        \draw[thick,CoalitionColor, fill=CoalitionColor!50, fill opacity=0.2]  \convexpath{9, 40}{1.6cm};
    \end{pgfonlayer}
\end{tikzpicture}
            }
            \caption{Initial state.}
            \label{subfig:CIS_EXP1}
        \end{subfigure}
        \hspace{0.001\textwidth}
        \begin{subfigure}[b]{0.32\textwidth}
            \centering
            \resizebox{1\textwidth}{!}{
                \begin{tikzpicture}
	\begin{pgfonlayer}{nodelayer}
		\node [style=none] (0) at (-1.5, 5) {$i-1$};
		\node [style=none] (1) at (1.5, 5) {$i$};
		\node [style=none] (2) at (4, 5) {$\ldots$};
		\node [style=none] (3) at (6.5, 5) {$T$};
		\node [style={large_group}] (4) at (1.5, 2) {$B_i$};
		\node [style={large_agent}] (5) at (6.5, 2) {$t$};
		\node [style={large_agent}] (6) at (6.5, -1.5) {$\mathit{t}_1$};
		\node [style=none] (7) at (6.5, -3.75) {$\vdots$};
		\node [style={large_agent}] (9) at (6.5, -10) {$\mathit{t}_{i+1}$};
		\node [style=none] (10) at (-4, 5) {$\ldots$};
		\node [style=none] (21) at (-6.5, 5) {$1$};
		\node [style={large_group}] (22) at (-1.5, -6.5) {$B_{i-1}$};
		\node [style={large_group}] (23) at (-6.5, -6.5) {$B_1$};
		\node [style=none] (24) at (-4, -6.5) {$\ldots$};
		\node [style=none] (25) at (4, 2) {$\ldots$};
		\node [style=none] (30) at (4, -1.5) {$\ldots$};
		\node [style=none] (35) at (4, -6.5) {$\ldots$};
		\node [style=none] (40) at (4, -10) {$\ldots$};
		\node [style=none] (41) at (4, -3.75) {$\vdots$};
		\node [style={large_agent}] (8) at (6.5, -6.5) {$\mathit{t}_i$};
	\end{pgfonlayer}
    \begin{pgfonlayer}{edgelayer}
        \draw[thick,CoalitionColor, fill=CoalitionColor!50, fill opacity=0.2]  \convexpath{5, 4}{1.6cm};
        \draw[thick,CoalitionColor, fill=CoalitionColor!50, fill opacity=0.2]  \convexpath{6, 30}{1.6cm};
        \draw[thick,CoalitionColor, fill=CoalitionColor!50, fill opacity=0.2]  \convexpath{8, 23}{1.6cm};
        \draw[thick,CoalitionColor, fill=CoalitionColor!50, fill opacity=0.2]  \convexpath{9, 40}{1.6cm};
    \end{pgfonlayer}
\end{tikzpicture}
            }
            \caption{State after the first recursive sub-call.}
            \label{subfig:CIS_EXP2}
        \end{subfigure}
        
        \vskip\baselineskip
        
        \begin{subfigure}[b]{0.32\textwidth}
            \centering
            \resizebox{1\textwidth}{!}{
                \begin{tikzpicture}
	\begin{pgfonlayer}{nodelayer}
		\node [style=none] (0) at (-1.5, 5) {$i-1$};
		\node [style=none] (1) at (1.5, 5) {$i$};
		\node [style=none] (2) at (4, 5) {$\ldots$};
		\node [style=none] (3) at (6.5, 5) {$T$};
		\node [style={large_agent}] (4) at (1.5, -1.5) {$b_{i}^1$};
		\node [style={large_agent}] (5) at (6.5, 2) {$t$};
		\node [style={large_agent}] (6) at (6.5, -1.5) {$\mathit{t}_1$};
		\node [style=none] (7) at (6.5, -3.75) {$\vdots$};
		\node [style={large_agent}] (8) at (6.5, -10) {$\mathit{t}_i$};
		\node [style={large_agent}] (9) at (6.5, -13.5) {$\mathit{t}_{i+1}$};
		\node [style=none] (10) at (-4, 5) {$\ldots$};
		\node [style=none] (21) at (-6.5, 5) {$1$};
		\node [style={large_group}] (22) at (-1.5, -10) {$B_{i-1}$};
		\node [style={large_group}] (23) at (-6.5, -10) {$B_1$};
		\node [style=none] (24) at (-4, -10) {$\ldots$};
		\node [style=none] (25) at (4, 2) {$\ldots$};
		\node [style=none] (30) at (4, -1.5) {$\ldots$};
		\node [style=none] (35) at (4, -10) {$\ldots$};
		\node [style=none] (40) at (4, -13.5) {$\ldots$};
		\node [style={large_agent}] (42) at (1.5, -6.5) {$b_{i}^{i-1}$};
		\node [style=none] (43) at (1.5, -3.75) {$\vdots$};
		\node [style={large_agent}] (48) at (1.5, -13.5) {$b_{i}^{i}$};
		\node [style=none] (53) at (4, -6.5) {$\ldots$};
		\node [style=none] (54) at (4, -3.75) {$\vdots$};
		\node [style={large_agent}] (41) at (6.5, -6.5) {$\mathit{t}_{i-1}$};
	\end{pgfonlayer}
    \begin{pgfonlayer}{edgelayer}
        \draw[thick,CoalitionColor, fill=CoalitionColor!50, fill opacity=0.2]  \convexpath{5, 25}{1.6cm};
        \draw[thick,CoalitionColor, fill=CoalitionColor!50, fill opacity=0.2]  \convexpath{6, 4}{1.6cm};
        \draw[thick,CoalitionColor, fill=CoalitionColor!50, fill opacity=0.2]  \convexpath{42, 41}{1.6cm};
        \draw[thick,CoalitionColor, fill=CoalitionColor!50, fill opacity=0.2]  \convexpath{8, 23}{1.6cm};
        \draw[thick,CoalitionColor, fill=CoalitionColor!50, fill opacity=0.2]  \convexpath{9, 48}{1.6cm};
    \end{pgfonlayer}
\end{tikzpicture}
            }
            \caption{State before the second recursive sub-call.}
            \label{subfig:CIS_EXP3}
        \end{subfigure}
        \hspace{0.01\textwidth}
        \begin{subfigure}[b]{0.32\textwidth}
            \centering
            \resizebox{1\textwidth}{!}{
                \begin{tikzpicture}
	\begin{pgfonlayer}{nodelayer}
		\node [style=none] (0) at (-1.5, 5) {$i-1$};
		\node [style=none] (1) at (1.5, 5) {$i$};
		\node [style=none] (2) at (4, 5) {$\ldots$};
		\node [style=none] (3) at (6.5, 5) {$T$};
		\node [style={large_agent}] (4) at (1.5, -1.5) {$b_{i}^1$};
		\node [style={large_agent}] (5) at (6.5, 2) {$t$};
		\node [style={large_agent}] (6) at (6.5, -1.5) {$\mathit{t}_1$};
		\node [style=none] (7) at (6.5, -3.75) {$\vdots$};
		\node [style={large_agent}] (8) at (6.5, -10) {$\mathit{t}_i$};
		\node [style={large_agent}] (9) at (6.5, -13.5) {$\mathit{t}_{i+1}$};
		\node [style=none] (10) at (-4, 5) {$\ldots$};
		\node [style=none] (21) at (-6.5, 5) {$1$};
		\node [style={large_group}] (22) at (-1.5, -13.5) {$B_{i-1}$};
		\node [style={large_group}] (23) at (-6.5, -13.5) {$B_1$};
		\node [style=none] (24) at (-4, -13.5) {$\ldots$};
		\node [style=none] (25) at (4, 2) {$\ldots$};
		\node [style=none] (30) at (4, -1.5) {$\ldots$};
		\node [style=none] (35) at (4, -10) {$\ldots$};
		\node [style=none] (40) at (4, -13.5) {$\ldots$};
		\node [style={large_agent}] (41) at (6.5, -6.5) {$\mathit{t}_{i-1}$};
		\node [style={large_agent}] (42) at (1.5, -6.5) {$b_{i}^{i-1}$};
		\node [style=none] (43) at (1.5, -3.75) {$\vdots$};
		\node [style={large_agent}] (48) at (1.5, -13.5) {$b_{i}^i$};
		\node [style=none] (53) at (4, -6.5) {$\ldots$};
		\node [style=none] (54) at (4, -3.75) {$\vdots$};
	\end{pgfonlayer}
    \begin{pgfonlayer}{edgelayer}
        \draw[thick,CoalitionColor, fill=CoalitionColor!50, fill opacity=0.2]  \convexpath{5, 25}{1.6cm};
        \draw[thick,CoalitionColor, fill=CoalitionColor!50, fill opacity=0.2]  \convexpath{6, 4}{1.6cm};
        \draw[thick,CoalitionColor, fill=CoalitionColor!50, fill opacity=0.2]  \convexpath{42, 41}{1.6cm};
        \draw[thick,CoalitionColor, fill=CoalitionColor!50, fill opacity=0.2]  \convexpath{8, 35}{1.6cm};
        \draw[thick,CoalitionColor, fill=CoalitionColor!50, fill opacity=0.2]  \convexpath{9, 23}{1.6cm};
    \end{pgfonlayer}
\end{tikzpicture}
            }
            \caption{State after the second recursive sub-call.}
            \label{subfig:CIS_EXP4}
        \end{subfigure}
        \hspace{0.01\textwidth}
        \begin{subfigure}[b]{0.32\textwidth}
            \centering
            \resizebox{1\textwidth}{!}{
                \begin{tikzpicture}
	\begin{pgfonlayer}{nodelayer}
		\node [style=none] (0) at (-1.5, 5) {$i-1$};
		\node [style=none] (1) at (1.5, 5) {$i$};
		\node [style=none] (2) at (4, 5) {$\ldots$};
		\node [style=none] (3) at (6.5, 5) {$T$};
		\node [style={large_agent}] (5) at (6.5, 2) {$t$};
		\node [style={large_agent}] (6) at (6.5, -1.5) {$\mathit{t}_1$};
		\node [style=none] (7) at (6.5, -3.75) {$\vdots$};
		\node [style={large_agent}] (8) at (6.5, -10) {$\mathit{t}_i$};
		\node [style={large_agent}] (9) at (6.5, -13.5) {$\mathit{t}_{i+1}$};
		\node [style=none] (10) at (-4, 5) {$\ldots$};
		\node [style=none] (21) at (-6.5, 5) {$1$};
		\node [style={large_group}] (22) at (-1.5, -13.5) {$B_{i-1}$};
		\node [style={large_group}] (23) at (-6.5, -13.5) {$B_1$};
		\node [style=none] (24) at (-4, -13.5) {$\ldots$};
		\node [style=none] (25) at (4, 2) {$\ldots$};
		\node [style=none] (30) at (4, -1.5) {$\ldots$};
		\node [style=none] (35) at (4, -10) {$\ldots$};
		\node [style=none] (40) at (4, -13.5) {$\ldots$};
		\node [style={large_agent}] (41) at (6.5, -6.5) {$\mathit{t}_{i-1}$};
		\node [style={large_group}] (48) at (1.5, -13.5) {$B_{i}$};
		\node [style=none] (53) at (4, -6.5) {$\ldots$};
		\node [style=none] (54) at (4, -3.75) {$\vdots$};
	\end{pgfonlayer}
    \begin{pgfonlayer}{edgelayer}
        \draw[thick,CoalitionColor, fill=CoalitionColor!50, fill opacity=0.2]  \convexpath{5, 25}{1.6cm};
        \draw[thick,CoalitionColor, fill=CoalitionColor!50, fill opacity=0.2]  \convexpath{6, 30}{1.6cm};
        \draw[thick,CoalitionColor, fill=CoalitionColor!50, fill opacity=0.2]  \convexpath{41, 53}{1.6cm};
        \draw[thick,CoalitionColor, fill=CoalitionColor!50, fill opacity=0.2]  \convexpath{8, 35}{1.6cm};
        \draw[thick,CoalitionColor, fill=CoalitionColor!50, fill opacity=0.2]  \convexpath{9, 23}{1.6cm};
    \end{pgfonlayer}
\end{tikzpicture}
            }
            \caption{Final state.}
            \label{subfig:CIS_EXP5}
        \end{subfigure}
        \caption{Illustration of the recursive procedure of \Cref{thm:CIS_EXP}. The individual sub-figures display the partition of the constructed game after certain steps of the recursive procedure. Bold nodes refer to groups of agents.}
        \label{fig:CIS_EXP}
    \end{figure*}
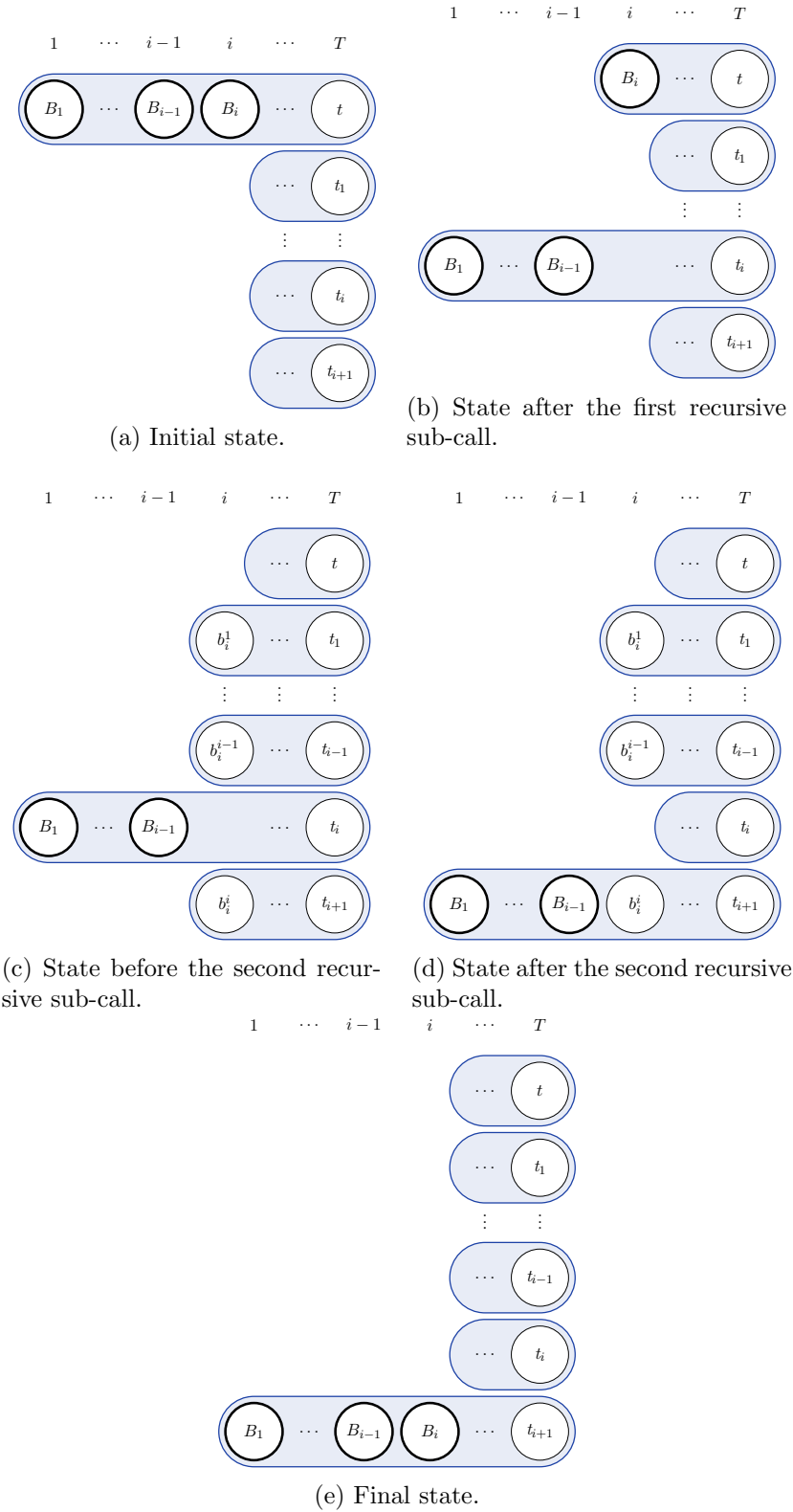

    Let $k \in \NN$ be a natural number with $k \geq 5$. We will construct an \ASHG $G$ such that the size of $G$ is polynomially bounded in $k$, whereas, in the $\chi$ dynamics on $G$ when starting from the singleton partition, there exists a sequence of at least $2^{k-1}$ deviations. Let $G = (N, \vf)$ be a game with a set $N = B \cup T$ of agents, where $B = \bigcup_{i \in [k]} B_i$ with $B_i = \{b_i^j\}_{j \in [i]}$ for each $i \in [k]$ is a set of \emph{counting agents}, and $T = \{t_i\}_{i \in [k+1]}$ is a set of \emph{track agents}.

    Let $f: \NN \rightarrow \NN$ be a function with $f(x) = x^{kx}$. We define valuations as follows: 
    \begin{enumerate}[1.]
        \item For each $b \in B$ and $t_i \in T$, let $\vf_b(t_i) = i$.
        \item For each $b_i \in B_i$ and $b_{j} \in B_{j}$ with $i < j$, let $\vf_{b_i}(b_{j}) = f(j)$.
        \item Let all other valuations be $0$.
    \end{enumerate}

    Note that $G$ consists of $|T| + \sum_{i \in [k]}|B_i| \leq k + 1 + \sum_{i \in [k]}k = k + 1 + k^2$ agents. Further, the maximum utility in $G$ is $f(k) = k^{k^2}$, and can thus be represented with $\log_2(k^{k^2}) = k^2 \cdot \log_2(k)$ bits. Hence, the space required to represent $G$ is polynomial in $k$; specifically, it holds that $|G| \in \bigO(k^5)$.

    The high-level idea is that we count up in binary, where a counting agent from group $B_i$ contributes not $2^i$ to the total score, but $f(i) = i^{ki}$. Further, after each step, we will ``collapse'' all counting agents from particular groups onto a single track, which ensures that the polynomial number of tracks suffices to represent an exponential number of different values.

    As the name suggests, we will use the track agents to form coalitions in the form of ``tracks'', where at each point during the relevant part of the dynamics, there exist exactly $k+1$ coalitions. We often denote to the coalition that some agent $t_i \in T$ belongs to by \emph{track-$t_i$ coalition} $C_{t_i}$ or simply by \emph{track-$i$ coalition} $C_{i}$. 
    
    Given a coalition $C \subseteq N$, and an agent $a \in N$, we denote to the value $\uf_a(C) - (\uf_a(C) \mod (k+1))$ by the \emph{truncated utility} of $a$ for $C$. Since no two track agents can ever be part of the same coalition, and as $f(2) = 2^{2k} > k+1$, the truncated utility value directly corresponds to a counting agents utility for a coalition when not taking the track agents into account. Given an $i \in [k]$ and a boolean compactor relation $\vartriangleleft \subseteq \NN \times \NN$, we define $B_{\vartriangleleft i} = \{ B_j \in B \mid j \vartriangleleft i\}$ (e.g., we write $B_{\leq i}$ to denote the set $\{B_j \in B \mid j \leq i\}$).

    In the following, we will show that, when starting from the singleton partition, in the $\chi$-dynamics on $G$, there is a sequence of at least $2^{k-1}$ deviations. Since $\CIS \devimplies \chi$, it suffices to show that each performed deviation is a \CIS deviation.

    Observe that, as there are no negative valuations, no agent will ever block an agent from joining a coalition, and only agents in $B_i$ can be blocked from leaving a coalition if it contains some agent in $B_j$ with $j < i$. Thus, when arguing that a deviation is a \CIS deviation, we will focus solely on showing that the latter part is not satisfied and that it is an \NS deviation.

    Starting from the singleton partition, we first let all agents in $B_{<k}$ deviate to join the track-$t_1$ coalition, and each agent in $B_k$ to join a unique track-$j$ coalition with $j > 1$. Then, we call the following recursive procedure with parameter $i = k-1$ with respect to the set $\{t_j\}_{j > 1}$. We illustrate the procedure in \Cref{fig:CIS_EXP}.

    Let $i \in [k]$ and assume that the following conditions hold: 
    \begin{enumerate}[I.]
        \item all agents in $B_{\leq i}$ belong to the track-$t$ coalition $C_t$ for some arbitrary $t \in T$, and
        \item there exists a set $\Tilde{T} \subseteq T \setminus \{t\}$ of $i+1$ track agents, such that for each $b \in B_{\leq i}$ and $\Tilde{t}, \Tilde{t}' \in \Tilde{T}$, it holds that
        \begin{enumerate}[(a)]
            \item $b$ has the same truncated utility for the two track coalitions $C_{\Tilde{t}}, C_{\Tilde{t}'}$, and
            \item $b$ has a strictly higher utility for $C_{\Tilde{t}}$ than for her current coalition $C_t$.
        \end{enumerate}
    \end{enumerate}

    Then, our procedure will produce a partition such that all agents in $B_{>i}$, as well as the track agents in $T$ remain the coalitions they were in before the procedure was called on $i$, and all agents in $B_{\leq i}$ move to the track coalition $C_{\Tilde{t}_{\mathrm{max}}}$ where $\Tilde{t}_{\mathrm{max}} \in \Tilde{T}$ and such that no agent has a higher valuation for any other agent in $\Tilde{T}$ (i.e., the track agent with the highest index among those in $\Tilde{T}$). 

    First, let $\Tilde{T}' = \Tilde{T} \setminus \{t_{\mathrm{max}}\}$ be a collection of $i$ track agents, and observe that all above conditions are also satisfied for $i-1$ with respect to $\Tilde{T}'$. Hence, in case $i > 1$, we call the procedure recursively on $i-1$. Afterwards, as per the above assumption, all agents in $B_{<i}$ belong to some coalition $C_{\Tilde{t}'}$ for a $\Tilde{t}' \in \Tilde{T}'$, and let $\Tilde{T}'' = \Tilde{T} \setminus \{\Tilde{t}'\}$.
    
    Second, we let each agent $b_i \in B_i$ deviate to a unique track-$\Tilde{t}''$ coalition for some $\Tilde{t}'' \in \Tilde{T}''$. To verify that these are indeed \CIS deviations, observe that any agent with nonzero valuation for any agent in $B_i$ is part of the track-$\Tilde{t}'$ coalition and $b$ strictly increases her utility as per condition II.b. 
    
    Next, in case $i > 1$, we recursively call the procedure again on $i-1$ with respect to the set $\Tilde{T}''$, where conditions I., and II.a trivially follow from the fact that they initially held for $i$ with respect to $\Tilde{T}$. To see that II.b also holds, observe that, as II.a initially held for $i$ with respect to $\Tilde{T}$, the truncated utility of any agent $b_j \in B_{<i}$ for a track-$\Tilde{t}''$ coalition with $\Tilde{t}'' \in \Tilde{T}''$ with respect to only the agents in $B_{> i}$ must be exactly $f(i)$ higher than that for her current track-$\Tilde{t}'$ coalition. In addition, agent $b_j$ gets at most $k+1$ further utility points from the track agents, and $\sum_{j \in [i-1]} j \cdot f(j) \leq (i-1)^2 \cdot f(i-1) \leq k^2 \cdot f(i-1)$ further utility points from the counting agents in $B_{<i}$. Therefore, with $i\geq 2$ and $k \geq 5$, we can follow that 
    \begin{align*}
        & \uf_{b_j}(C_{\Tilde{t}''}) - \uf_{b_j}(C_{\Tilde{t}'}) \\
        \geq \quad & f(i) - (k+1 + k^2 \cdot f(i-1)) \\
        = \quad & k^{ki} - (k+1 + k^2 \cdot k^{k(i-1)}) \\
        \geq \quad & k^{k} \cdot k^{k(i-1)} - (k^2 + k^2 \cdot k^{k(i-1)}) \\
        \geq \quad & k^{5} \cdot k^{k(i-1)} - (k^4 \cdot k^{k(i-1)}) \\
        \geq \quad & k^{k(i-1)} \\
        \geq \quad & 1.
    \end{align*}
    Hence, condition II.b is also satisfied for $i-1$ with respect to $\Tilde{T}''$.

    Finally, as per the above assumption, after the recursive call has terminated, all agents in $B_{< i}$ belong to coalition $C_{\Tilde{t}_{\mathrm{max}}}$. Hence, as per condition II.b, and based on how $\Tilde{t}_{\mathrm{max}}$ was chosen, all agents in $B_i$ (apart from the one that already belongs to $C_{\Tilde{t}_{\mathrm{max}}}$) can perform \CIS deviations to join $C_{\Tilde{t}_{\mathrm{max}}}$. Thus, our claim about the final partition after the recursive procedure on $i$ terminated was justified.

    \smallskip 

    One can verify that the specified conditions hold for the initial call on $i = k-1$ with respect to $\{t_j\}_{j > 1}$. Further, each recursive call requires at least one \CIS deviation, and, in case $i \in [2, k-1]$, induces two recursive calls on $i-1$. Thus, the sequence requires at least $2^{k-1}$ deviations, which concludes the proof.
\end{proof}

\subsection{Proof of \Cref{thm:CIS_FFPT_asymm_0}}\label{app:CISffpt}

\CISffpt* 

\begin{proof}
    Define $X := \{a \in N \mid \exists b \in N : \vf_a(b) > 0 \land \vf_b(a) = 0\}$ as the set of agents defining $s(G)$.
    Moreover, let $Y := N \setminus X$ be the set of the remaining agents.

    We claim that the statement holds for any execution of the \CIS dynamics that is constructed according to the following three phases:

    \begin{enumerate}[1.]
        \item Iterate through the agents in $Y$ in an arbitrary order. For each agent $y \in Y$, let $y$ perform a \CIS deviation if possible.
        \item Iterate through the agents in $X$ in an arbitrary order. For each agent $x \in X$, let $x$ perform a \CIS deviation to join an agent in $Y$ if possible, where we choose the deviation that maximizes $x$'s utility in case there are multiple such deviations.
        \item Perform arbitrary \CIS deviations until a \CIS partition is reached.
    \end{enumerate}

    We begin by showing the following two auxiliary claims.

    \begin{claim}\label{claim:CIS_FFPT_asymm_0_Y_only_one_deviation}
        After the first phase, no agent in $Y$ can perform a further \CIS deviation.
    \end{claim}
    \begin{claimproof}
        First, observe that no agent that is placed in a nonsingleton coalition after the first phase can perform any further \CIS deviation. 
        This is because, for an agent $a \in N$ such that some $y \in Y$ deviated to join $a$'s singleton coalition, it must have hold that $\vf_y(a) > 0$.
        In addition, since $y \not \in X$, it holds that $\vf_a(y) \neq 0$.
        Hence, since $a$ approved of $y$ joining, it follows that $\vf_a(y) > 0$. 
        Hence, $y$ and $a$ block each other from leaving.
        Similarly, for any $y'\in Y$ that joined some nonsingleton coalition $C$, there must have been an $a' \in C$ with $\vf_{y'}(a') > 0$. 
        Hence, $\vf_{a'}(y') > 0$ as above and $a'$ blocks $y'$ from leaving.

        Next, assume for contradiction that there is an agent $y \in Y$ that is in her singleton coalition after the first phase but can perform a \CIS deviation to join some coalition $C$ at some later point in the dynamics. 
        Let $C' \subseteq C$ be the agents that joined $C$ before $y$'s turn in the first phase. 
        If $\uf_y(C') > 0$, then $y$ could have joined $C'$ at her turn in the first phase. Hence, $\uf_y(C') \leq 0$ must hold and there must be some agent $a \in C \setminus C'$ such that $\vf_y(a) > 0$. 
        Note that if $a$ had ended in a nonsingleton coalition in the first phase, she would not have been able to perform another deviation, and the coalition $C$ would never have formed.  
        Hence, $a$ must have been in her singleton coalition throughout the first phase.
        Thus, when $y$ had her turn in the first phase, she could have joined $a$, a contradiction to $y$ ending up in a singleton coalition after the first phase.
    \end{claimproof}

    \begin{claim}\label{claim:CIS_FFPT_asymm_0_Y_stays_single}
        Each agent in $Y$ that is in her singleton coalition after the second phase will stay in her singleton coalition for the remaining dynamics.
    \end{claim}
    \begin{claimproof}
        Assume for contradiction that some agent $y\in Y$ is in a singleton coalition at the end of the second phase but is joined by an agent in the third phase.
        By \Cref{claim:CIS_FFPT_asymm_0_Y_only_one_deviation}, the joining agent must be an agent $x \in X$.
        In addition, since $y$ is in a singleton coalition at the end of the second phase, she was in a singleton coalition throughout the second phase.
        Now, recall that in \CIS dynamics utilities weakly increase after each deviation.
        Hence, $x$ obtains a higher utility than at the end of the second phase.
        But then, when it was the turn of agent~$x$ in the second phase, she did not perform her best deviation, because joining $y$ was an option.
    \end{claimproof}
    
    Now, because of \Cref{claim:CIS_FFPT_asymm_0_Y_only_one_deviation}, all deviations in the third phase must be performed by agents in $X$. Further, because of \Cref{claim:CIS_FFPT_asymm_0_Y_stays_single}, each of these deviations must be to join one of the coalitions that contained at least one agent in $X$ after the second phase.
    Note that there are at most $s(G)$ such deviations.
    Hence, after the first and second phases are complete, there are at most $s(G)^{s(G)}$ unique partitions reachable in the \CIS dynamics. 
    Since the \CIS dynamics are acyclic, there can be at most $n$ deviations in the first two phases and $s(G)^{s(G)}$ in the third phase, until the dynamics converge.
\end{proof}

\end{document}